\documentclass[12pt, a4paper]{article}
\usepackage[sort&compress, round, semicolon, authoryear]{natbib}
\usepackage{graphicx, lscape}
\usepackage{graphics, color}
\usepackage{epsfig}
\usepackage{epstopdf}
\usepackage{wrapfig}
\usepackage{caption}
\usepackage{subcaption}
\usepackage{amsmath, amsthm, amssymb, amscd}
\usepackage{latexsym}
\usepackage{multirow} 
\usepackage[running, switch*, displaymath, mathlines]{lineno}
\usepackage{array}
\usepackage{color}
\usepackage{rotating}
\usepackage[hyphens]{url}
\usepackage[hidelinks]{hyperref}  
\usepackage{cleveref}
\usepackage{fullpage}
\usepackage{fancyvrb}
\usepackage{pdfpages}
\usepackage{fmtcount}
\usepackage{verbatim}
\usepackage[english]{babel}
\usepackage[utf8]{inputenc}
\usepackage{tikz} 
\usepackage{float}
\usetikzlibrary{arrows, decorations.pathmorphing, backgrounds, fit, positioning, shapes.symbols, chains}
\usetikzlibrary{decorations.markings}
\usepackage{lscape}
\usepackage{algpseudocode}
\usepackage{algorithm}
\usepackage{soul}
\usepackage{booktabs}
\usepackage{enumerate}
\usepackage{colortbl}
\long\def\comment#1{} 
\usepackage{booktabs}
\usepackage{yhmath}  
\usepackage{mathtools}
\usepackage{rotating}
\usepackage{diagbox}
\usepackage{bm}
\usepackage{changepage}
\usepackage{rotating}
\usepackage{dsfont}
\usepackage{setspace}

\usepackage{etoolbox}
\apptocmd{\normalsize}{%
  \setlength{\abovedisplayskip}{8pt}%
  \setlength{\belowdisplayskip}{8pt}%
  \setlength{\abovedisplayshortskip}{8pt}%
  \setlength{\belowdisplayshortskip}{8pt}%
}{}{}

\newtheorem{thm}{Theorem}[]
\newtheorem{assumption}{Assumption}[]
\newtheorem{lemma}{Lemma}[section]

\newtheoremstyle{remarkstyle}    
  {3pt}    
  {3pt}    
  {\normalfont}  
  {}       
  {\bfseries}  
  {.}      
  { }      
  {}       
\theoremstyle{remarkstyle}
\newtheorem{remark}{Remark}[]

\newcommand\norm[1]{\left\lVert#1\right\rVert}  
\renewcommand{\baselinestretch}{1.5}

\def\spacingset#1{\renewcommand{\baselinestretch}%
{#1}\small\normalsize} \spacingset{1}

\usepackage{mathtools}

\usepackage{enumitem}
\setlist[itemize]{noitemsep, topsep=0pt}
\setlist[enumerate]{noitemsep, topsep=0pt}

\newcommand{\defeq}{\vcentcolon=}
\usepackage[compact]{titlesec}

\DeclareMathOperator*{\rank}{rank}
\DeclareMathOperator{\E}{\mathbb{E}}
\DeclareMathOperator{\M}{\mathbb{M}}
\DeclareMathOperator{\W}{\mathbb{W}}
\DeclareMathOperator{\R}{\mathbb{R}}
\DeclareMathOperator{\I}{\mathbb{I}}

\newcommand\lam[2]{\lambda_{#1,#2}}
\renewcommand{\th}{\theta}
\newcommand{\Lam}{\Lambda}
\newcommand{\Lams}{\Lambda_k^s}
\newcommand{\Lamz}{\Lambda_k^0}
\newcommand{\Lamsp}{\Lambda_k^{s^{\prime}}}
\newcommand{\Lamzp}{\Lambda_k^{0^{\prime}}}
\newcommand{\Lamsr}{\Lambda_k^{s,r}}
\newcommand{\Lamsnr}{\Lambda_k^{s,-r}}

\newcommand{\sumkit}{\sum\limits_{k=1}^K\sum\limits_{i=1}^N\sum\limits_{t=1}^T}
\newcommand{\sumit}{\sum\limits_{i=1}^N\sum\limits_{t=1}^T}
\newcommand{\sumk}{\sum\limits_{k=1}^K}
\DeclareRobustCommand{\p}{^\prime}
\newcommand{\h}{\expandafter\hat}

\begin{document}

\begin{center}
        \pagenumbering{arabic} 
    
        {\Large{\textbf{A Data-Adaptive Factor Model Using Composite Quantile Approach}}\medskip}
        \vskip 7mm
        {\sc Seeun Park and Hee-Seok Oh\footnote{Corresponding author: heeseok@stats.snu.ac.kr}}\\
{Department of Statistics, Seoul National University, Seoul 08826, Korea}
\end{center}
\vskip 5mm
\noindent 
{\bf Abstract}: 
This paper proposes a data-adaptive factor model (DAFM), a novel framework for extracting common factors that explain the structures of high-dimensional data. 
DAFM adopts a composite quantile strategy to adaptively capture the full distributional structure of the data, thereby enhancing estimation accuracy and revealing latent patterns that are invisible to conventional factor models. 
In this paper, we develop a practical algorithm for estimating DAFM by minimizing an objective function based on a weighted average of check functions across quantiles. We also establish the theoretical properties of the estimators, including their consistency and convergence rates. Furthermore, we derive their asymptotic distributions by introducing approximated estimators from a kernel-smoothed objective function, and propose two consistent methods for determining the number of factors. Simulation studies demonstrate that DAFM outperforms existing factor models across different data distributions, and real data analyses on volatility and forecasting further validate its effectiveness.


\doublespacing
\section{Introduction} \label{sec:Introduction}

Factor models are an essential tool for reducing the dimensionality of data and summarizing its features through a few latent factors. The classical factor model described by \cite{lawley1962factor} has been extended in numerous ways to address different data characteristics and application needs, such as the approximate factor model (AFM), the dynamic factor model (DFM), and others \citep{geweke1977dynamic, chamberlain1982arbitrage,mcardle1994structural,ansari2000bayesian}. These extensions have significantly increased the flexibility and applicability of factor models, including their use for high-dimensional panel data \citep{ forni2000generalized,forni2017dynamic,stock2002forecasting}. When dealing with large datasets in factor models, principal component analysis (PCA) methods are often employed for estimation due to their computational simplicity \citep{stock2002forecasting, bai2003inferential}. However, PCA relies on ordinary least squares, which limits its applicability to data that meet specific moment constraints. Much research has been done to address this limitation. For example, \cite{he2022large} proposed a robust factor estimation method that assumes an elliptical distribution for the joint distribution of factors and idiosyncratic errors. This method applies PCA to a spatial Kendall’s tau matrix instead of the sample covariance matrix and fits a linear regression model to obtain factor scores. In addition, \cite{chen2021quantile} introduced the quantile factor model (QFM), which minimizes a loss function based on the check function rather than the traditional quadratic loss function. This approach allows QFM to extract factors corresponding to a specific quantile of data and improves the robustness of the factor model. As demonstrated by \cite{chen2021quantile}, QFM with a quantile level of 0.5 successfully captures the true factors in data with heavy-tailed distributions. In a similar vein, \cite{ando2020quantile} developed a framework designed to capture quantile-dependent factor structures tailored explicitly for application to large financial panel datasets. Unlike \cite{chen2021quantile}, their model also incorporates observable regressors.

While QFM is effective at capturing the true factors in heavy-tailed data, it may not capture some of the hidden information in the data because it only represents a single quantile of the factor structure. Additionally, selecting an inappropriate quantile level that does not align with the intrinsic characteristics of the data can lead to incomplete or misleading factor estimates. In this study, we propose a new factor model, called the data-adaptive factor model (DAFM), 
which employs a composite quantile approach that adaptively reflects the underlying structure of the entire data distribution. By aggregating information across multiple quantile levels, the proposed model more effectively captures the overall structure of the data, distinguishing itself from mean-based or single-quantile factor models.
Our concept aligns with that of \cite{zou2008composite}, which extends quantile regression to composite quantile regression by considering multiple quantile levels simultaneously. The key idea is to define an objective function that combines the check functions at individual quantile levels. For the proposed model, we estimate the common factors and quantile-dependent loadings that minimize the objective function. This study also introduces a flexible weighting scheme for the check functions; for example, these weights can be estimated based on the importance of the information at each quantile level, or they can be chosen uniformly if there is no specific rationale for the choice. By utilizing DAFM, we expect the factors to incorporate more sophisticated information and reveal some features that may be hidden at a particular single-quantile level.

We note that \cite{huang2023composite} used a similar approach to extract factors by considering multiple quantile levels; however, this model assumes that the loadings are fixed across quantiles and only the intercept varies across quantiles. In contrast, our model differs in that it allows the loadings to vary over quantiles, which accommodates a broader range of data-generating processes and provides greater flexibility in capturing the underlying structure of data. In particular, the factor model presented by \cite{huang2023composite} is a special case of the proposed model because the quantile-varying intercepts can be expressed through the quantile-varying loadings. Furthermore, \cite{huang2023composite} assumed that the idiosyncratic errors are independent and identically distributed (i.i.d.), whereas, in our study, it is sufficient for them to be independent. 
\cite{chen2025universalfactormodels} discussed a universal factor model (UFM) in their recent working paper, which considers a model form similar to DAFM. Although the factor structure of UFM appears similar to that of DAFM, the objectives and theoretical developments of the two methods differ significantly. The primary contribution of UFM is to unify AFM and QFM under a common framework, thereby showing that these models can be induced as special cases of UFM. They relax the conventional strong factor assumption and allow factors to have only a weak influence on the mean or on conditional quantiles at most quantile levels. The simulation focuses on simplified one-factor settings, verifying that the method can recover factors and remain robust to weak factors without testing different error distributions or conducting real data analysis. In contrast, DAFM is driven by the need to manage data with a broad spectrum of distributions, including those with heavy tails, heteroskedasticity, and asymmetry. We combine information across multiple quantile levels to obtain improved factors, validate performance under complex error structures, and further confirm effectiveness in economic and financial applications. Another difference lies in how the objective function and asymptotic theory are developed. DAFM explicitly defines the objective as a weighted average of check functions, using kernel smoothing only to achieve asymptotic normality. Conversely, UFM incorporates kernel smoothing into its definition and additionally applies inverse density weighting to derive asymptotic distributions, thus imposing stronger assumptions from the outset. Moreover, its implementation requires more tuning parameters, such as bandwidths.

The more specific contributions of this paper are as follows. First, we present an algorithm for estimating factors and quantile-dependent loadings by minimizing an objective function through iterative parameter updates. Second, we establish the consistency of the estimators and derive their convergence rates, ensuring that our approach provides reliable estimates. The asymptotic distributions are also investigated by introducing modified estimators that minimize a kernel-smoothed objective function. Third, we propose two consistent estimators of the number of factors: one is based on an information criterion, and the other is related to the ranks of the estimated loadings. We prove that both methods accurately identify the true number of factors with probability converging to 1. Fourth, we conduct a simulation study to evaluate the performance of the proposed model in two scenarios: a basic location-shift model with an additional error term and a location-scale-shift model where the error term is scaled by a common factor. In both cases, we show the benefits of considering multiple quantile levels compared to QFMs applied to each quantile. Finally, we apply the proposed model to two economic datasets, demonstrating its ability to extract the underlying structure of the data and its effectiveness in forecasting.

The R codes used in the experiments are available at \url{https://github.com/p-seeun/data-adaptive-factor}. All proofs are given in the supplement.


\section{The proposed model and its estimation} \label{sec:CQFM}
\subsection{Data-adaptive factor model} 

\cite{chen2021quantile} proposed the following quantile factor model (QFM) 
\[ 
Q_{X_{it}}(\tau|f_t(\tau))=\lambda^{\prime}_i(\tau) f_t(\tau),\quad i=1,\ldots,N \text{ and } t=1,\ldots,T,
\] 
where $\lambda_i(\tau)$ and $f_t(\tau) \in \mathbb{R}^r$ are loadings and factors at a given quantile $\tau \in (0,1)$. The loadings and factors can be estimated by minimizing the objective function, 
\[ 
\frac{1}{NT}\sum\limits_{i=1}^N \sum\limits_{t=1}^T \rho_{\tau}(X_{it}-\lambda_{i}^{'}f_t),
\]
where $\rho_\tau(\epsilon)=(\tau-\mathds{1}\{ \epsilon \leq 0 \})\epsilon $ is the check function. We note that the estimated factor $\hat{f}_t$ depends on $\tau$ and describes a single quantile of the data. In contrast, the proposed data-adaptive factor model (DAFM) considers the data structure simultaneously at multiple quantiles $0 < \tau_1, \tau_2, \ldots, \tau_K < 1$ to estimate the common factor that captures a more accurate distributional
features of the data. To do this, we model that the quantiles of data are described by a common factor structure with quantile-dependent loadings, denoted as 
\begin{equation} \label{eq:CQFM str}
    Q_{X_{it}}(\tau_k|f_t) = \lambda_{k,i}^{\prime} f_t,\quad i=1,\ldots,N \text{ and } t=1,\ldots,T, 
\end{equation}
for each quantile index $k=1,\ldots,K$. This modeling is reasonable, as it encompasses both the basic form of a factor model and more complex variations. For example, a data-generating process with three factors, 
$X_{it}=\lambda_{1i}f_{1t}+\lambda_{2i}f_{2t}+\lambda_{3i}f_{3t}+\epsilon_{it}$, can be expressed as
\[
Q_{X_{it}}(\tau_k|f_t) = [\lambda_{1i} \, \lambda_{2i} \, \lambda_{3i} \, Q_{\epsilon_{it}}(\tau_k)] \cdot [f_{1t} \, f_{2t} \, f_{3t} \, 1]^{\prime},\quad k=1,\ldots,K.
\]
Also, for a location-scale-shift model where the error term is scaled by a common factor, $X_{it}=\lambda_{1i}f_{1t}+\lambda_{2i}f_{2t}+\lambda_{3i}f_{3t}\epsilon_{it}$, it can be expressed as
\begin{equation} \label{eq:LS CQFM}
    Q_{X_{it}}(\tau_k|f_t) = [\lambda_{1i} \; \lambda_{2i} \; \lambda_{3i}\cdot Q_{\epsilon_{it}}(\tau_k)] \cdot [f_{1t} \; f_{2t} \; f_{3t}]^{\prime},\quad k=1,\ldots,K.
\end{equation} 
These two models will be discussed in detail in the simulation study. For implementing DAFM, we define the objective function based on $K$ quantile levels as  
\begin{equation} \label{eq:obj ftn}
\frac{1}{NT}\sum\limits_{k=1}^K \sum\limits_{i=1}^N \sum\limits_{t=1}^T  w_k \rho_{\tau_k}(X_{it}-\lambda_{k,i}^{'} \, f_t),
\end{equation}
where $w_k$ denotes the positive weight of each check function at quantile level $\tau_k$. The estimators of loadings $\lambda_{1,1},\ldots,\lambda_{1,N},\ldots,\lambda_{K,1},\ldots,\lambda_{K,N}$ and factors $f_1, \ldots f_T$ in (\ref{eq:CQFM str}) are obtained by minimizing the objective function of (\ref{eq:obj ftn}). It is important to note that the factors $\{f_t\}$, unlike the loadings, do not depend on the quantile index $k$. These are the composite quantile factors obtained by aggregating the loss functions across multiple quantile factor models at the given quantile levels, which describe the overall data distribution. Additionally, for $K=1$, the proposed model reduces to the QFM structure of \cite{chen2021quantile}.

Since DAFM considers multiple quantile levels, the extracted factors are expected to capture more information on data, such as skewness and heavy tails, while uncovering features that might be hidden when using only a single quantile level. For example, if we apply QFM at $\tau=0.5$ to the location-scale-shift model with three factors of (\ref{eq:LS CQFM}) and the median of error $\epsilon_{it}$ is close to zero, the third factor $f_{3t}$ may not be well-identified. Conversely, simultaneously incorporating quantiles distant from 0.5 can significantly enhance the identification of $f_{3t}$ and improve the overall factor estimation.

In this paper, we adopt a fixed-effects methodology that treats the true factors as given parameters rather than random variables. This approach is commonly used in some studies of factor models, particularly those involving quantiles (see, e.g., \cite{ma2021estimation}, \cite{ando2020quantile}, \cite{chen2021quantile}, and \cite{kong2022matrix}). While extending it to a random setting poses challenges to the technical proofs of its theoretical properties, it remains theoretically valid and practically useful within the current framework. We leave further exploration of the random setting for future work. Here, we assume that the number of factors, $r$, is known; how to estimate $r$ is discussed in Section \ref{sec:Factor number}. 

\subsection{Normalization of factors and loadings} \label{subsec:Normalization}

To rewrite the model of (\ref{eq:CQFM str}) in a matrix form, let $F=[ f_1,\ldots,f_T ]^{\prime} \in \mathbb{R}^{T \times r}$ and $\Lambda_k=[\lambda_{k,1},\ldots,\lambda_{k,N}]^{\prime} \in \mathbb{R}^{N \times r},~k=1,\ldots,K$. The matrix representing the $\tau_k$-th quantile of $X$ given $F$, denoted as $Q_X(\tau_k|F)$, is given by
$Q_X(\tau_k|F) = \Lambda_k F^{\prime}.$ It is a well-known fact that the loadings and factors are not separately identifiable \citep{bai2002determining}, that is, $\Lambda_k F^\prime = \Lambda_k H H^{-1} F^{\prime}$ for an arbitrary $r \times r$ invertible matrix $H$, so a normalization method is needed to determine $\Lambda_k$ and $F$ \citep{bai2013principal}. We take the following normalization,
\begin{equation} \label{eq:normalization}
    \frac{F^{\prime}F}{T} = I_r \; \text{ and } \; \frac{\Lambda_{k^*}^{\prime} \Lambda_{k^*} }{N}\; \text{is a diagonal matrix with non-increasing diagonal elements}
\end{equation}
for a specific quantile index $k^* \in \{1,\ldots,K\}$. The condition about $k^*$ will be mentioned in Section \ref{sec:Theoretical properties} when analyzing the theoretical properties of the estimators. This normalization allows to identify $\Lambda_{k^*}$ and $F$, and once $F$ is determined, the loadings $\Lambda_k$ for all $k \neq k^{*}$ can be also identified. 

\subsection{Estimation algorithm} \label{subsec:Estimation of the model}

We present an algorithm to minimize the objective function of (\ref{eq:obj ftn}). The non-convexity of the objective function with respect to the parameters makes it difficult to find a global minimizer. However, given $F=[f_1, \ldots, f_T]'$, the objective function becomes convex for each $\lambda_{k,i}$. Similarly, when $\Lambda_k=[\lambda_{k,1}, \ldots, \lambda_{k,N}]^{\prime}$ is given for $k=1,\ldots,K$, the function becomes convex with respect to each $f_t$. Thus, we propose the following iterative steps to exploit these convexities:
\begin{adjustwidth}{3.5em}{0em}
\begin{enumerate}
    \item[Step 1.] Choose an initial parameter $F^{(0)}=[f_1^{(0)}, \ldots, f_T^{(0)}]^{\prime}$ and set $l=1$.
    \item[Step 2.] Given $F^{(l-1)}=[f_1^{(l-1)},\ldots,f_T^{(l-1)}]^{\prime}$, estimate the loading matrices $\Lambda^{(l-1)}_k=[\lambda_{k,1}^{(l-1)},\ldots,\lambda_{k,N}^{(l-1)}]^{\prime}$ for $k=1,\ldots,K$ by minimizing
    $\sum\limits_{k=1}^{K}\sum\limits_{t=1}^{T}\sum\limits_{i=1}^{N} w_k\rho_{\tau_k}(X_{it}-\lambda^{'}_{k,i}f^{(l-1)}_t)$.
    \item[Step 3.] Given $\Lambda_k^{(l-1)}$ for $k=1,\ldots,K$, estimate the factor matrix $F^{(l)}$ by minimizing
     $\sum\limits_{k=1}^{K}\sum\limits_{t=1}^{T}\sum\limits_{i=1}^{N} w_k\rho_{\tau_k}(X_{it}-\lambda^{(l-1)'}_{k,i}f_t)$.
    Then, set $l=l+1$.
    \item[Step 4.] Repeat Steps 2 and 3 until the objective function converges.
    \item[Step 5.]  Normalize the final factor matrix $F^{(L)}$ and the loading matrices $\Lambda_1^{(L)},\ldots,\Lambda_K^{(L)}$ according to the normalization condition in (\ref{eq:normalization}), where $L$ denotes the total number of iterations.
\end{enumerate}
\end{adjustwidth}
\vspace{3mm}

We have some remarks about Steps 2, 3, and 5. 
\begin{itemize}
\item[(i)] The optimization problem in Step 2 is equivalent to solving
\[
 \lambda^{(l-1)}_{k,i}=\underset{\lambda}{\operatorname{argmin}}\sum\limits_{t=1}^{T} \rho_{\tau_k}(X_{it}-\lambda^{'}f_t^{(l-1)})
\]
for each $k$ and $i$, which is also equivalent to fitting the $\tau_k$-th quantile regression between $X_{it}$ and $f_t^{(l-1)}$ for $t=1,\ldots,T$.

\item[(ii)] The optimization problem in Step 3 can be expressed as
\begin{center}
    $f^{(l)}_{t}=\underset{f}{\operatorname{argmin}}\sum\limits_{k=1}^{K}\sum\limits_{i=1}^{N} w_k\rho_{\tau_k}(X_{it}-\lambda_{k,i}^{(l-1)'}f)$
\end{center}
for each $t$. This problem can be solved by the alternating direction method of multipliers (ADMM) algorithm after reformulating it as
\begin{center}
    $\underset{f,\epsilon}{\operatorname{minimize}}\sum\limits_{k=1}^K\sum\limits_{i=1}^N \rho_{\tau_k}(\epsilon_{ki})$~~subject to $\Lambda^*f+\epsilon=X^*$,
\end{center}
where
\begin{align*}
    \Lambda^{*}&=(w_1\lambda_{1,1}^{(l)}, \ldots, w_1\lambda_{1,N}^{(l)}, \ldots, w_K\lambda_{K,1}^{(l)}, \ldots, w_K\lambda_{K,N}^{(l)})^{\prime}\in \mathbb{R}^{NK \times r},\\
    \epsilon&= (\epsilon_{11},\ldots,\epsilon_{1N},\ldots,\epsilon_{K1},\ldots,\epsilon_{KN})^{\prime} \in \mathbb{R}^{NK},\\
    X^{*}&=(w_1 X_{1t}, \ldots,w_1 X_{Nt}, \ldots, w_K X_{1t}, \ldots, w_K X_{Nt})^\prime \in \mathbb{R}^{NK}.
\end{align*} 
Note that the constraint implies $w_k \lambda_{k,i}^{\prime} f_t+ \epsilon_{ki}=w_k X_{it}$ in elementwise form. \cite{pietrosanu2021advanced} suggested the ADMM algorithm for estimating the composite quantile regression. We have partially adapted and reformulated their algorithm for this study.

\item[(iii)] Suppose that $F^{(L)}$ and $\Lambda_1^{(L)},\ldots,\Lambda_K^{(L)}$ are obtained after Step 4. The normalized estimators according to (\ref{eq:normalization}), denoted as $\Tilde{F}^{(L)}$ and $\Tilde{\Lambda}^{(L)}_1,\ldots,\Tilde{\Lambda}^{(L)}_K$, are obtained as follows. 
First, we perform the diagonalization, 
\[
\left(\frac{F^{(L)^{\prime}} F^{(L)}}{T}\right)^{\frac{1}{2}} \cdot \left(\frac{\Lambda_{k^*}^{(L)^{\prime}} \Lambda_{k^*}^{(L)}}{N}\right) \cdot \left(\frac{F^{(L)^{\prime}} F^{(L)}}{T}\right)^{\frac{1}{2}} = U D U^{\prime}, 
\]
and then define $H$ as $H = \left(\frac{F^{(L)^{\prime}} F^{(L)}}{T}\right)^{-\frac{1}{2}} U$. Using $H$, the normalized factors and loadings are given by
$\tilde{F}^{(L)} = F^{(L)} H$ and $\tilde{\Lambda}^{(L)}_1 = \Lambda^{(L)}_1 \left(H^{-1}\right)^{\prime}, \ldots, \tilde{\Lambda}^{(L)}_K = \Lambda^{(L)}_K \left(H^{-1}\right)^{\prime}$.  It can be easily verified that $\Tilde{F}^{(L)}$ and $\Tilde{\Lambda}^{(L)}_1,\ldots,\Tilde{\Lambda}^{(L)}_K$ satisfy the normalization condition, and that $\Tilde{\Lambda}^{(L)}_k \Tilde{F}^{(L)^\prime} = \Lambda_k^{(L)}F^{(L)^{\prime}}$ holds for all $k=1,\ldots,K$. 
\end{itemize}

\section{Theoretical properties of the proposed estimators} \label{sec:Theoretical properties}


\subsection{Consistency and convergence rate}

We denote the true factors and loadings as $\{f_t^0:t=1,\ldots,T\}$ and $\{\lambda_{k,i}^0: k=1,\ldots,K \text{ and } i=1,\ldots,N\}$, respectively. Define a vector of loadings and factors as 
\[
\theta=\big(\lambda_{1,1}^{\prime},\ldots,\lambda^{\prime}_{1,N},\ldots,\lambda_{K,1}^{\prime},\ldots,\lambda^{\prime}_{K,N}, f_1^{\prime}, \ldots, f_T^{\prime}\big)^\prime \in \mathbb{R}^{(KN+T)r},
\]
and the true parameter as 
\[
\theta^0=\big(\lambda_{1,1}^{0^\prime},\ldots,\lambda_{1,N}^{0^\prime},\ldots,\lambda_{K,1}^{0^\prime},\ldots,\lambda_{K,N}^{0^\prime}, f_1^{0^\prime}, \ldots, f_T^{0^\prime}\big)^\prime \in \mathbb{R}^{(KN+T)r}.
\]
Let $\mathcal{A}$, $\mathcal{F} \subset \mathbb{R}^{r}$ be the sets of possible loadings and factors, respectively. We further define a parameter space as
\[
\Theta^r=\big\{\theta \in \mathbb{R}^M: \lambda_{k,i} \in \mathcal{A}, f_t \in \mathcal{F} \text{ for all } k,i,t; \{\lambda_{k,i}\} \text{ and } \{f_t\} \text{ satisfy the normalizations in (\ref{eq:normalization})} \big\}.
\]
Let $M_{NT}(\theta)$ denote the objective function of (\ref{eq:obj ftn}), i.e.,
\[ 
M_{NT}(\theta) = \frac{1}{NT}\sum\limits_{k=1}^K \sum\limits_{i=1}^N \sum\limits_{t=1}^T  w_k\rho_{\tau_k}\big(X_{it}-\lambda_{k,i}^{'} \, f_t\big). 
\]
The estimator $\hat{\theta}=\big(\hat{\lambda}_{1,1}^{\prime},\ldots,\hat{\lambda}_{1,N}^{\prime},\ldots,\hat{\lambda}_{K,1}^{\prime},\ldots,\hat{\lambda}_{K,N}^{\prime}, \hat{f}_1^{\prime}, \ldots, \hat{f}_T^{\prime}\big)^\prime$ is then defined as
\[
\hat{\theta} = \underset{\theta \in \Theta^r}{\operatorname{argmin}} \, M_{NT}(\theta).
\]
Let $\Lambda_k^0=[\lambda^0_{k,1},\ldots,\lambda^0_{k,N}]^{\prime}$ and $F^0=[f^0_1,\ldots,f^0_T]^{\prime}$ denote the true matrices of loadings and factors, respectively, and $\hat{\Lambda}_k=[\hat{\lambda}_{k,1},\ldots,\hat{\lambda}_{k,N}]^{\prime}$ and $ \hat{F}=[\hat{f}_1,\ldots,\hat{f}_T]^{\prime}$ be the corresponding estimators. Additionally, let $L_{NT}=\min\{\sqrt{N},\sqrt{T}\}$ and $\mbox{sgn}(A)$ denote a diagonal matrix where each diagonal element is the sign of the corresponding diagonal element of matrix $A$. The norm $\|\cdot\|$ denotes the Frobenius norm, defined as $\|A\|_F \defeq \sqrt{\operatorname{Trace}(A^\prime A)}$.

We establish the consistency of the estimator $\hat{\theta}$ under the following assumptions. 

\begin{assumption}\label{assump:consistency}

\begin{itemize}\setlength{\itemsep}{0ex}
\item[(a)] $\mathcal{A}$ and $\mathcal{F}$ are compact.
\item[(b)] $\theta_0 \in \Theta^r$ and for $\Lambda^{0^{\prime}}_{k^*} \Lambda^0_{k^*}/N = \text{diag}(\nu_1(N),\ldots,\nu_r(N))$, $\nu_i(N) \rightarrow \nu_i$ as $N \rightarrow \infty$, where $\nu_1 > \nu_2 > \cdots >\nu_r>0$. 
\item[(c)] Let $h_{it}$ be the conditional probability density function of $x_{it}$ given $f_t^0$. For any compact set $C\subset \mathbb{R}$, there exists some constant $c$ such that $h_{it}(x) \geq c$, $\forall x \in C$ for all $i,~t$. 
\item[(d)] Given $\{f_t^0: t=1,\ldots,T\}$, for any $k$, the idiosyncratic error $\epsilon_{k,it}=X_{it}-\lambda_{k,i}^{\prime}f_t$ is independent across $i,~t$.
\end{itemize}
\end{assumption}

\begin{thm}[Consistency]\label{thm:consistency} 
Under Assumption \ref{assump:consistency},
\[
  \big\|\hat{F}-F^0 \hat{S}\big\|/\sqrt{T}=O_p\big(L_{NT}^{-1}\big) \text{\, and \,} \big\|\hat{\Lambda}_k- \Lambda^0_k \hat{S}\big\|/\sqrt{N} = O_p\big(L_{NT}^{-1}\big) \, , \, k=1,\ldots,K, 
\]
where $\hat{S}=\mbox{sgn}\big(\hat{F}^{\prime}F^0/T\big)$. 
\end{thm}

\begin{remark}
Assumption \ref{assump:consistency}(a) ensures that the norm of each $f_t$ and $\lambda_{k,i}$ constituting the elements of $\Theta^r$ is bounded and is also used in the proof to establish the separability of the stochastic process for $\theta$. Assumption \ref{assump:consistency}(b) is necessary for ordering the factors, as applied in \cite{bai2003inferential}. 
Assumption 1(d) is required to apply Hoeffding's inequality in the proof and can also be found in \cite{ando2020quantile} and \cite{chen2021quantile}.
\end{remark}

\begin{remark}
Compared to the proof in \cite{chen2021quantile}, our analysis presents additional technical challenges due to the composite nature of the objective function. While \cite{chen2021quantile} dealt with a single quantile level, enabling the use of quantile-specific empirical process techniques, our framework involves multiple quantile levels, which require controlling the behavior of aggregated quantities across quantiles. As a result, key elements such as estimation errors and distances become more complex, and additional technical steps are involved in the proofs. For instance, to establish that $\underset{\theta \in \Theta^r}{sup}|\mathbb{W}_{NT}(\theta)|=o_p(1)$ in Lemma 1 of the supplement, we apply the Cauchy–Schwarz inequality to derive a new upper bound on the probability related to $\mathbb{W}_{NT}(\theta)$ after invoking Hoeffding's inequality (see the definition of $\mathbb{W}_{NT}(\theta)$ on page 1 of the supplement).
    
\end{remark}

\begin{remark}\label{remark:b} In practice, $k^*$ is unknown. However, it is not essential to identify $k^*$ in the normalization step to ensure consistency for the following reason: if we assume only the existence of $k^*$ such that Assumption \ref{assump:consistency}(b) is satisfied when the estimators are normalized with respect to $k^*$ as in (\ref{eq:normalization}), then the consistency of the estimator still holds, regardless of the quantile index $k$ selected for the condition (\ref{eq:normalization}) in estimation. Suppose that we have $\check{F}, \check{\Lambda}_1, \ldots, \check{\Lambda}_K$ by normalizing the minimizer of $M_{NT}(\theta)$ with an arbitrary quantile level index instead of $k^*$ that is unknown. Define $\check{U}$ from the diagonalization of $\frac{\check{\Lambda}^{\prime}_{k^*}\check{\Lambda}_{k^*}}{N}$, $\check{U} \check{D} \check{U}^{\prime}.$ It can be verified that the new estimators $\bar{F}=\check{F}\check{U}, \bar{\Lambda}_1 = \check{\Lambda}_1 \check{U}, \ldots, \bar{\Lambda}_K=\check{\Lambda}_K \check{U}$ satisfy the original normalization condition (\ref{eq:normalization}) with $k^*$. Then, for $\check{S}:=\mbox{sgn}((\check{F}{\check{U}})^{\prime}F^0/T)=\mbox{sgn}(\bar{F}^{\prime}F^0/T)=:\bar{S}$, it holds that 
\[
   \|\check{F}-F^0 \check{S}{\check{U}}^{\prime}\|/\sqrt{T} = \|\bar{F}-F^0 \bar{S}\|/\sqrt{T} =   O_p\big(L_{NT}^{-1}\big).
\]
For loading matrices, the same result holds as   
\[
\|\check{\Lambda}_k- \Lambda^0_k \check{S}{\check{U}}^{\prime}\|/\sqrt{N} = O_p(L_{NT}^{-1}) \, , \, k=1,\ldots,K.
\]
The above facts imply that the estimators, when normalized with respect to an arbitrary $k$, are consistent for the true parameters up to a rotation matrix.
\end{remark}

\subsection{Asymptotic distributions}
A significant challenge in achieving the asymptotic distributions of the estimators is that the objective function, $M_{NT}(\theta)$, is not differentiable and, therefore, can not be expanded using the Taylor series method. In the literature, a common strategy to address this difficulty is using a kernel smoothing technique to make the objective function differentiable. For instance, \cite{horowitz1998bootstrap} modified the least absolute deviations (LAD) objective function by applying kernel smoothing to attain asymptotic refinements for the bootstrap of a median regression model. \cite{galvao2016smoothed} investigated the asymptotic properties of a new estimator, a fixed effect smoothed quantile regression (FE-SQR) estimator, obtained by minimizing the smoothed quantile regression objective function. Furthermore, \cite{chen2021quantile} employed a similar method by introducing kernel smoothing to analyze the asymptotic distributions of estimators in their quantile factor model. Following this approach, we introduce a new objective function using kernel smoothing to handle the non-smoothness of the original objective function, 
\[
S_{NT}(\theta) = \frac{1}{NT}\sum\limits_{k=1}^K \sum\limits_{i=1}^N \sum\limits_{t=1}^T  w_k\varrho_{\tau_k}\big(X_{it}-\lambda_{k,i}^{'} \, f_t\big),
\]
where
$\varrho_{\tau_k}(\epsilon) = \left\{\tau_k-K\left( \frac{\epsilon}{h} \right) \right\}\cdot \epsilon.$
Here, $K(\epsilon) = \int_\epsilon^1 k(s) \,ds$ serves as the survival function of $k$, and $k(\cdot)$ is a continuous kernel function with support on $[-1,1]$. The bandwidth parameter $h$ satisfies $h \rightarrow 0$ as $N,T \rightarrow \infty$. The smoothness of $S_{NT}(\theta)$ is achieved by replacing the indicator function $\mathds{1}\{\epsilon \leq 0\}$ in the check function $\rho_{\tau_k}(\epsilon)$ with the smoothed term $K(\epsilon/h)$. We then define the new estimator as
$
\tilde{\theta}=\big(\tilde{\lambda}_{1,1}^{\prime},\ldots,\tilde{\lambda}^{\prime}_{1,N},\ldots,\tilde{\lambda}_{K,1}^{\prime},\ldots,\tilde{\lambda}^{\prime}_{K,N}, \tilde{f}_1^{\prime}, \ldots, \tilde{f}_T^{\prime}\big)^\prime \in \mathbb{R}^{(KN+T)r},
$
which is obtained by minimizing the new objective function,
$\tilde{\theta} = \underset{\theta \in \Theta^r}{\operatorname{argmin}} \, S_{NT}(\theta).$ 
To derive the asymptotic distributions of $\tilde{\lambda}_{k,i}$ and $\tilde{f}_t$,
we define
\[
\Psi_{N,t} = \frac{1}{N}\sum\limits_{k=1}^K \sum\limits_{i=1}^N w_k h_{it}(\lambda_{k,i}^{0^{\prime}}f_t^0)\lambda_{k,i}^0 \lambda_{k,i}^{0^\prime} 
\text{\; and \;} 
\Phi_{T,k,i} = \frac{1}{T} \sum\limits_{t=1}^T h_{it}(\lambda_{k,i}^{0^{\prime}}f_t^0)f_t^0 f_t^{0^\prime},
\]
where $h_{it}$ is the conditional density function of $x_{it}$ given $f^0_t$. For the details of $h_{it}$, see  Assumptions 1(c) and 2(e). Furthermore, we define their limits as
$\Psi_t = \underset{N \rightarrow \infty}{\lim}  \Psi_{N,t}$ and $
\Phi_{k,i} = \underset{T \rightarrow \infty}{\lim} \Phi_{T,k,i}.$
For any function $g$, we write
$g^{(j)}(s) = \left( \frac{\partial}{\partial s}\right)^j g(s)$ for $ j\geq 1$ and $g^{(0)}(s)=g(s).$
For $A>0$ and $B>0$, the symbol $A \sim B$ implies that there exists a constant $c>0$ such that 
$c^{-1}A \leq B \leq cA$.
\begin{assumption}\label{assump:normality}
\begin{itemize} \setlength{\itemsep}{-0.5ex}
\item[(a)] For any $t,~k$, and $i$, $\Psi_t \succ 0$ and $ \Phi_{k,i} \succ 0$.
\item[(b)] For any $k$ and $k^\prime \in \{1,\ldots,K\}$, $\underset{N \rightarrow \infty}{\lim}\frac{\Lambda_k^{0^\prime}\Lambda_{k^{\prime}}^0}{N} = \Sigma_{k k^{\prime}}$.
\item[(c)] \mbox{$\lambda_{k,i}^0$ is an interior point of $\mathcal{A}$ for each $k,~i$, and $f_t^0$ is an interior point of $\mathcal{F}$ for each $t$.}
\item[(d)] $k(\cdot)$ is a kernel function of order $m$, i.e, 
\[
\int k(s) ds = 1 \text{\, and \,} \int s^j k(s) ds = 0 , \; j=0,\ldots,m-1.
\]
Also, $k(s) = k(-s)$ and it is twice continuously differentiable.
\item[(e)] The density function $h_{it}$ is $m+2$ times continuously differentiable. 
For any compact set $C\subset \mathbb{R}$, $h_{it}^{(j)}$ for $j=0,\ldots,m+2$ are uniformly bounded on $C$ across all $i$ and $t$; i.e., there exist constants $h_l$ and $h_u$, depending on $C$, such that for all $i$ and $t$:
\[
h_l \leq h_{it}^{(j)}(s)\leq h_u, \forall s \in C, \text{\; for\; } j=0,\ldots,m+2 .
\]
\item[(f)] For sufficiently large $N$ and $T$, $N \sim T$ and $h \sim T^{-c}$ for $\frac{1}{m}<c<\frac{1}{6}$.
\end{itemize}
\end{assumption}

\begin{thm}[Asymptotic distributions]\label{thm:normality} 
Let $\tilde{S}=\mbox{sgn}\left(\tilde{F}^{\prime} F_0 / T\right)$. Under Assumptions \ref{assump:consistency} and \ref{assump:normality}, for each $t$,
\[
\sqrt{N} \left( \tilde{f}_t-\tilde{S}f_t^0 \right) \xrightarrow{d} \mathcal{N}\left(0, \Psi_t^{-1} \sum\limits_{k=1}^K\sum\limits_{k^\prime=1}^K \left\{ w_k w_{k^\prime} \operatorname{min}(\tau_k, \tau_{k^\prime}) (1-\operatorname{max}(\tau_k, \tau_{k^\prime})) \Sigma_{k k^\prime} \right\} \Psi_t^{-1} \right),
\]
and for each $k$ and $i$,
\[
\sqrt{T}\left(\tilde{\lambda}_{k,i}-\tilde{S} \lambda_{k, i}^0\right) \xrightarrow{d} \mathcal{N}\left(0, \tau_k(1-\tau_k) \Phi_{k,i}^{-2}\right).
\]
\end{thm}

\begin{remark} The assumptions follow those used in  \cite{chen2021quantile}, except for (b), which is necessary for obtaining the limit of covariance of smoothed check functions at two quantile levels, each multiplied by the corresponding loading vector. 
Similar assumptions can also be found in the literature on smoothed quantile regression, such as \cite{galvao2016smoothed} and \cite{horowitz1998bootstrap}. Assumption \ref{assump:normality}(f) implies that, unlike in Theorem \ref{thm:consistency}, $N$ and $T$ are required to grow at the same rate to ensure the asymptotic normality of the estimators.
\end{remark}

\begin{remark}
    For the proof of Theorem \ref{thm:consistency}, we consider the vector $\mathcal{S}(\theta)\in \mathbb{R}^{(KN+T)r}$, defined on page 13 of the supplement, to capture key quantities. The stochastic expansion is derived from the following equation:
    \[
    \mathcal{S}(\tilde{\theta})=\mathcal{S}\left(\theta_{0}\right)+\mathcal{H}(\theta_0) \cdot\left(\tilde{\theta}-\theta_{0}\right)+0.5 \mathcal{R}(\tilde{\theta}),
    \]
    where $\mathcal{R}(\theta)$ denotes the remainder term and $\mathcal{H}(\theta)$ is the derivative of $\mathcal{S}(\theta)$. 
    Unlike in \cite{chen2021quantile}, $\mathcal{S}(\theta)$ is constructed to incorporate information from loadings across all quantile levels, and as a result, the derivative $\mathcal{H}(\theta)$ becomes substantially more complex, with dimension $(KN+T)r \times (KN+T)r$. Establishing its invertibility and other essential properties requires new definitions of parameter-dependent functions and more refined analytical techniques. Moreover, the resulting stochastic expansion of $\tilde{f}_t -f_t^0$ is more complex than in \cite{chen2021quantile}, as it incorporates contributions from the loadings and check functions across multiple quantiles. In particular, additional covariance terms arise, and we must handle terms of the form
    \[E\left(\varrho_k^{(1)}(X_{i t}-\lambda_{k, i}^{0^{\prime}} f_{t}^0) \cdot \varrho_{k^\prime}^{(1)}(X_{i t}-\lambda_{k^\prime, i}^{0^{\prime}} f_{t}^0)\right) \, , \quad k \neq k^\prime, \] which are absent in a single-quantile setting. 
    
\end{remark}

\begin{remark}
Suppose that we normalized the estimators using the quantile level index $\ell$, instead of $k^*$ in condition (\ref{eq:normalization}). Let $\check{F}, \check{\Lambda}_1, \ldots, \check{\Lambda}_K$ be the estimators obtained by normalizing the minimizer of $S_{NT}(\theta)$ with index $\ell$. Define a rotation matrix $\check{U}$ from the diagonalization $\frac{\check{\Lambda}^{\prime}_{k^*}\check{\Lambda}_{k^*}}{N} = \check{U} \check{D} \check{U}^{\prime}$, where the diagonal elements of $\check{D}$ are arranged in non-increasing order, and $\check{U}$ has non-negative diagonal elements. Similarly, as in Remark \ref{remark:b}, let $\bar{F}=\check{F}\check{U}$ and $\bar{\Lambda}_k = \check{\Lambda}_k \check{U}$ for all $k$, so that $\bar{F}, \bar{\Lambda}_{1} , \ldots, \bar{\Lambda}_K$ satisfy the original normalization condition. Note that $\bar{\lambda}_{k,i} = \check{U}^\prime \check{\lambda}_{k,i}$ and $\bar{f}_t =\check{U}^\prime \check{f}_t$. Under the additional assumption that $\Sigma_{\ell\ell}$ has distinct eigenvalues, it can be shown that $\check{U} \xrightarrow{p} U$, where $U$ is defined as the eigenvector matrix of $\Sigma_{\ell\ell}$. Specifically, $\Sigma_{\ell\ell}$ is diagonalized as $\Sigma_{\ell\ell} = UBU^\prime$, where $B$ is a diagonal matrix with non-increasing elements, and the diagonal entries of $U$ are non-negative. Using the relation $\frac{\bar{\Lambda_\ell^\prime}\bar{\Lambda}_\ell}{N} = \check{U} \, \frac{\check{\Lambda_\ell^\prime}\check{\Lambda}_\ell}{N} \, \check{U}^{\prime} $ and the fact that $\frac{\check{\Lambda_\ell^\prime}\check{\Lambda}_\ell}{N}$ is diagonal with non-increasing elements, we establish $\check{U} \xrightarrow{p} U$ by applying perturbation theory for eigenvectors (see Section 6.2 of \cite{franklin2012matrix}) and noting that $\|\bar{\Lambda}_\ell- \Lambda^0_\ell \bar{S}\|/\sqrt{N} = O_p(L_{NT}^{-1})$, where $\bar{S}=\mbox{sgn}\left(\bar{F}^{\prime} F_0 / T\right)$. As a result, we have
\begin{gather*}
\sqrt{N}\big(\check{f}_t- \check{U} \check{S}f_t^0\big) \xrightarrow{d}  \mathcal{N} \left(0, U \, \Psi_t^{-1} \sum\limits_{k=1}^K\sum\limits_{k^\prime=1}^K \left\{ w_k w_{k^\prime} \operatorname{min}(\tau_k, \tau_{k^\prime}) (1-\operatorname{max}(\tau_k, \tau_{k^\prime})) \Sigma_{k k^\prime} \right\} \Psi_t^{-1} \, U^\prime \right), \\[2pt]
\sqrt{T}\left(\check{\lambda}_{k,i}- \check{U}\check{S} \lambda_{k, i}^0\right) \xrightarrow{d} \mathcal{N}\left(0, \tau_k(1-\tau_k) U \, \Phi_{k,i}^{-2} \, U^\prime \right),
\end{gather*}
where $\check{S} = \mbox{sgn}\left((\check{F}\check{U})^{\prime} F_0 / T\right)$.
\end{remark}

\section{Factor number estimation} \label{sec:Factor number}

We propose two consistent estimators of the number of factors. The first estimator is based on the information criterion proposed by \cite{bai2002determining} for approximate factor models. The second one is related to the ranks of the estimated loading matrices, which is based on the criterion of \cite{chen2021quantile}. To derive these estimators, we need to analyze the estimators of factors and loadings, assuming that the number of factors is $s \in \{1,\ldots,m\}$ for $m>r$. When the number of factors is $s$, let $\lambda_{k,i}^s, f_t^s \in \mathbb{R}^s$ be vectors of loadings and factors, respectively, and let $\Lambda_k^s=[\lambda_{k,1}^s,\ldots,\lambda_{k,N}^s]^{\prime} \in \mathbb{R}^{N \times s}$ and $F=[f_1,\ldots,f_T]^{\prime} \in \mathbb{R}^{T \times s}$ be the matrices of loadings and factors. In a similar sense, we define a vector of parameters,
\[
\theta^s=\big(\lambda_{1,1}^{s^\prime},\ldots,\lambda^{s^\prime}_{1,N},\ldots,\lambda_{K,1}^{s^\prime},\ldots,\lambda^{s^\prime}_{K,N}, f_1^{s^\prime}, \ldots, f_T^{s^\prime}\big) \in \mathbb{R}^{(KN+T)s},
\]
and let $\mathcal{A}^s$ and $\mathcal{F}^s$ be compact subsets of $\mathbb{R}^s$. We set the normalization condition for these parameters as
\begin{equation} \label{eq:normalization2}
    \frac{1}{T}(F^s)^{\prime}F^s = I_r \; \text{ and } \; \frac{1}{N}(\Lambda_{k^*}^s)^{\prime} \Lambda_{k^*}^s \; \text{is a diagonal matrix with non-increasing elements}.
\end{equation}
Also, define the new parameter space,
\begin{equation*}
\resizebox{1.0\hsize}{!}{$\Theta^s=\{\theta \in \mathbb{R}^{(KN+T)s} : \lambda^s_{k,i} \in \mathcal{A}^s, f^s_t \in \mathcal{F}^s \text{ for all } k,i,t, \{\lambda^s_{k,i}\} \text{ and } \{f^s_t\} \text{ satisfy the normalizations in (\ref{eq:normalization2})} \}$}.
\end{equation*}
The estimator 
$\hat{\theta}^s=\big(\hat{\lambda}_{1,1}^{s^\prime},\ldots,\hat{\lambda}^{s^\prime}_{1,N},\ldots,\hat{\lambda}_{K,1}^{s^\prime},\ldots,\hat{\lambda}^{s^\prime}_{K,N}, \hat{f}_1^{s^\prime}, \ldots, \hat{f}_T^{s^\prime}\big) \in \mathbb{R}^{(KN+T)s}$
is then defined to be
$\hat{\theta}^s =  \underset{\theta \in \Theta^s}{\operatorname{argmin}} \, M_{NT}(\theta).$
Now, Theorem \ref{thm:factornumber1} below provides an information criterion that estimates the factor number. 

\begin{assumption}\label{assump:facnum}
\begin{itemize}
\item[(a)] For any compact set $C\subset \mathbb{R}$, there exists some constant $c>0$ such that $h_{it}(x) \leq c$, $\forall x \in C$ for all $i,t$, for $h_{it}$ in Assumption \ref{assump:consistency}(c).

\item[(b)]  For $\ell>r$, let $\Lambda^{*}_{k}=[\Lambda^0_{k} \; \mathbf{0}_{N \times (\ell-r)}] \in \mathbb{R}^{N \times \ell}$ and $F^{*}=[F^0 \; V] \in \mathbb{R}^{T \times \ell}$ for some $V$ such that $F^{*^{\prime}}F^*/T=I_r$. Denoting each row of $\Lambda^*_{k}$ and $F^*$ as $\lambda^*_{k,i}$ and $f_t^*$, respectively, it holds that $\lambda^*_{k,i} \in \mathcal{A^\ell}$ and $f^*_t \in \mathcal{F^\ell}$.
\end{itemize}
\end{assumption}

\begin{thm}[]\label{thm:factornumber1} 
Suppose that Assumptions \ref{assump:consistency} and \ref{assump:facnum} hold. For a sequence $\{P_{NT}\}$ such that $P_{NT} \rightarrow 0$ and $P_{NT}L_{NT}^2 \rightarrow \infty$ as $N,~T \rightarrow \infty$, define
\[
\hat{r}=\underset{1 \leq \ell \leq s}{\operatorname{argmin}} M_{NT}(\hat{\theta}^{\ell}) + \ell \cdot P_{NT},
\]
for $s>r$. Then, it holds that
$P(\hat{r}=r) \rightarrow 1$ as $N,~T \rightarrow \infty$
\end{thm}

\begin{remark} Theorem \ref{thm:factornumber1} does not require the computation of the normalized estimators because $M_{NT}(\hat{\theta}^{\ell})$ only relies on the estimate of the common components $\hat{\Lambda}_k^{\ell} \hat{F}^{\ell^{\prime}}$ for all $k$. 
\end{remark}

\begin{remark} Assumption \ref{assump:facnum}(a) combined with Assumption \ref{assump:consistency}(c) implies that the conditional probability density function is uniformly bounded both above and away from zero. Assumption \ref{assump:facnum}(b) along with Assumption \ref{assump:consistency}(b) ensures that
\[ 
\theta^*= \big(\lambda^*_{1,1},\ldots,\lambda^*_{1,N},\ldots,\lambda^*_{K,1},\ldots,\lambda^*_{K,N},f^*_1,\ldots, f^*_T\big) \in \Theta^\ell.
\]
Then, using $\Lambda^{*}_k F^{*^{\prime}}=\Lambda_k^0 F^{0^\prime}$, we obtain $M_{NT}(\hat{\theta}^\ell)-M_{NT}(\theta^0)=M_{NT}(\hat{\theta}^\ell)-M_{NT}(\theta^*)\leq 0$, which is utilized in the proof.
\end{remark}

The next theorem addresses the selection of factor numbers based on the rank of the estimated loading matrices. This method involves computing the eigenvalues of $\hat{\Lambda}^{s^\prime}_k\hat{\Lambda}^s_k/N$ for some fixed $s>r$. Therefore, unlike Theorem \ref{thm:factornumber1}, the normalizing step is inevitable. However, it is important to note that the eigenvalues of $\hat{\Lambda}^{s\prime}_k\hat{\Lambda}^s_k/N$ for all $k$ remain consistent regardless of the choice of $k^*$ used in the normalization condition (\ref{eq:normalization}) in the estimation. Thus, the estimators can be normalized with any arbitrary $k$ instead of $k^*$ to estimate the number of factors using the following theorem. We define $\hat{\sigma}^k_{1}(N) \geq \cdots \geq \hat{\sigma}^k_{s}(N)$ as the eigenvalues of $\hat{\Lambda}^{s^\prime}_k\hat{\Lambda}^s_k/N$.

\begin{thm}[]\label{thm:factornumber2} 
Suppose that Assumptions \ref{assump:consistency} and \ref{assump:facnum} hold. For a sequence $\{\kappa_{NT}\}$ such that $\kappa_{NT} \rightarrow 0$ and $\kappa_{NT}L_{NT} \rightarrow \infty$ as $N,~T \rightarrow \infty$, define
\[
\hat{r} = \underset{k=1,\ldots,K}{\operatorname{max}}\sum\limits_{i=1}^s I\big({\hat{\sigma}}^k_{i}(N)>\kappa_{NT}\big),
\]
for $s>r$. Then, it holds that $P(\hat{r}=r) \rightarrow 1$ as $N,~T \rightarrow \infty$.
\end{thm}

\begin{remark} This estimator is related to the ranks of the estimated loading matrices because $\hat{r}$ can be represented as the maximum among the rank estimators of $\hat{\Lambda}^{s^\prime}_{k}\hat{\Lambda}^s_{k}/N$ for $k=1,\ldots,K$, which in fact, $k^*$ yields the maximum value with probability converging to 1. In the proof of Theorem \ref{thm:factornumber2}, it can be shown that
\[
P\Big(\hat{r} = \sum\limits_{i=1}^s I\big(\hat{\sigma}^{k^*}_{i}(N)>\kappa_{NT}\big)\Big) \rightarrow 1 \text{ as } N,~T \rightarrow \infty.
\]
Specifically, we mainly use the property that $\hat{\sigma}^{k}_{r+1}(N) = O_p\big(L_{NT}^{-1}\big)$ for all $k$, and that $P(\hat{\sigma}_r^{k^*}(N) \leq \kappa_{NT}) \rightarrow 0$ since $\hat{\sigma}_r^{k^*}(N) > \frac{\nu_{r}}{2}>0$ for sufficiently large $N$ in the proof. Then, $\kappa_{NT}$ is served as a threshold 
for the eigenvalues to distinguish between $\hat{\sigma}_r^{k^*}(N)$ and $\hat{\sigma}_{r+1}^{k^*}(N)$.
\end{remark}

\section{Simulation study} \label{sec:Simulation study}

For the simulation study, we generate data from two model types, each with different error distributions. The first model is a location-shift model, where the factors only affect the location of the data, and the second model is a location-scale-shift model, where some factors also influence the scale of the data. The accuracy of the true factor estimation in each case is measured by calculating the adjusted $R^2$ value of regressing each true factor on the estimated factor. 
For comparison, we include the robust PCA method of \cite{he2022large}, the QFM of \cite{chen2021quantile}, and the model proposed by \cite{huang2023composite}, hereafter referred to as `CQF-H'. 
Here, the number of factors is assumed to be known.
We note that although the UFM of \cite{chen2025universalfactormodels} is relevant to our framework, we do not include it in the comparison because our simulation settings and objectives address substantially different issues from those of UFM.

\subsection{Location-shift model} \label{subsec:L shift model}

We first consider a location-shift model generated by three factors. The data-generating process is adopted from \cite{chen2021quantile} as 
$X_{it}=\lambda_{1i}f_{1t}+\lambda_{2i}f_{2t}+\lambda_{3i}f_{3t}+\epsilon_{it}$, where $f_{1t}=0.8 f_{1,t-1}+ e_{1t}$,  $f_{2t}=0.5 f_{2,t-1}+ e_{2t}$, $f_{3t}=0.2 f_{3,t-1}+ e_{3t}$,~ $e_{1t},~e_{2t},~e_{3t} \overset{i.i.d}{\sim} N(0,1)$, and $\lambda_{1i},~\lambda_{2i},~\lambda_{3i} \overset{i.i.d}{\sim} N(0,1)$. The model can be written in the DAFM form as $Q_{X_{it}}(\tau_k|F) = \lambda_{k,i}^{\prime} f_t$ with $\lambda_{k,i}=(\lambda_{1i}, \lambda_{2i}, \lambda_{3i}, Q_{\epsilon_{it}}(\tau_k))^{\prime}$ and $f_t = (f_{1t}, f_{2t}, f_{3t}, 1)^{\prime}$. The error term $\epsilon_{it}$ is generated from three distributions with their own characteristics:
\begin{itemize}
    \item $\epsilon_{it} \sim t(2)$: A $t$-distribution with 2 degrees of freedom, characterized by heavy tails.
    \item $\epsilon_{it} \sim 0.5 \cdot N(2,0.5) + 0.5 \cdot N(-2,0.5)$: A Gaussian mixture distribution that is bimodal with modes at 2 and $-2$.
    \item $\epsilon_{it} \sim skewt(\xi=0, \omega=4, \alpha=4, \nu=3)$: A skewed $t$-distribution that is right-skewed with location parameter $\xi=0$, scale parameter $\omega=4$, shape parameter $\alpha=4$, and degrees of freedom $\nu=3$.
\end{itemize}
For analysis, the theoretical mean of each error term has been subtracted. We consider several combinations of $(N,T)$ with values $(50,50), (100,50), (100,100), (100,150)$, and $ (150,150).$ To implement DAFM, we select a set of quantile levels $(\tau_1, \tau_2, \tau_3, \tau_4, \tau_5)=(0.1, 0.3, 0.5, 0.7, 0.9)$ to capture a diverse range of quantiles within the data distribution. In this scenario, we assign weights based on the known error distribution. To reflect the influence of each quantile level when capturing the factors, we set $w_k = f(F^{-1}(\tau_k))$, where $f$ and $F$ represent the probability density function and the cumulative density function of the error term $\epsilon_{it}$, respectively. Table \ref{table:L shift} lists the adjusted $R^2$ values of regressing each true factor on the factors estimated by DAFM, as well as those obtained by the robust PCA method, QFM, and CQF-H. For QFM, the model is fitted at five individual quantile levels for comparison, denoted as QFM($\tau$), $\tau=0.1,0.3,0.5,0.7,0.9$, in the first column. Also, CQF-H is implemented using the same set of quantile levels as those selected for DAFM. For each distribution of $\epsilon_{it}$, the results are computed for each combination of $(N,T)$. 
We observe that DAFM outperforms all the comparison methods in most cases. Notably, the performance of QFM varies across quantile levels, with the optimal quantile level depending on the distribution of $\epsilon_{it}$. For example, when $\epsilon_{it}$ follows the $t$-distribution, QFM performs best at $\tau=0.5$. In contrast, for the skewed $t$-distribution, QFM tends to perform better at lower quantile levels, $\tau=0.1$ or $\tau=0.3$, compared to higher quantile levels. For the Gaussian mixture distribution, it performs well at $\tau=0.1$ or $0.9$. The proposed DAFM consistently outperforms QFM in most cases by incorporating information from all quantile levels simultaneously. CQF-H also benefits from combining quantile levels and generally outperforms QFMs at quantile levels; however, in certain cases, it shows lower $R^2$ values than QFM at the optimal quantile level, while DAFM consistently achieves higher $R^2$ values in all scenarios. 

\begin{sidewaystable}[] 
\setlength{\tabcolsep}{11pt}
\def\arraystretch{1.3}
\caption{Location-shift model: the adjusted $R^2$ of the true factors fitted by several methods according to different settings of $(N,T)$.}
\label{table:L shift}
\scalebox{0.77}{
\begin{tabular}{llllllllllllllll}
\hline
(N,T) & \multicolumn{3}{c}{(50,50)} & \multicolumn{3}{c}{(100,50)} & \multicolumn{3}{c}{(100,100)} & \multicolumn{3}{c}{(100,150)} & \multicolumn{3}{c}{(150,150)} \\ \cmidrule(rl){2-4} \cmidrule(rl){5-7} \cmidrule(rl){8-10} \cmidrule(rl){11-13} \cmidrule(rl){14-16} 
 & $f_{1t}$ & $f_{2t}$ & $f_{3t}$ & $f_{1t}$ & $f_{2t}$ & $f_{3t}$  & $f_{1t}$ & $f_{2t}$ & $f_{3t}$ & $f_{1t}$ & $f_{2t}$ & $f_{3t}$  & $f_{1t}$ & $f_{2t}$ & $f_{3t}$  \\ \hline
\multicolumn{14}{c}{$\epsilon_{it} \sim t(2)$}  \\ \hline

\multicolumn{1}{l|}{Rob.PCA}  & 0.889 & 0.833 & \multicolumn{1}{l|}{0.76}  & 0.934 & 0.88 & \multicolumn{1}{l|}{0.854} & 0.947 & 0.9   & \multicolumn{1}{l|}{0.863}         & 0.951 & 0.903 & \multicolumn{1}{l|}{0.877}         & 0.967 & 0.931 & 0.911          \\
\multicolumn{1}{l|}{QFM(0.1)} & 0.8   & 0.64  & \multicolumn{1}{l|}{0.551} & 0.888 & 0.787         & \multicolumn{1}{l|}{0.75}  & 0.934 & 0.861 & \multicolumn{1}{l|}{0.82} & 0.938 & 0.878 & \multicolumn{1}{l|}{0.844}         & 0.962 & 0.919 & 0.901          \\
\multicolumn{1}{l|}{QFM(0.3)} & 0.965 & 0.936 & \multicolumn{1}{l|}{0.921} & 0.984 & 0.971         & \multicolumn{1}{l|}{0.966} & 0.987 & 0.975 & \multicolumn{1}{l|}{0.969}         & 0.988 & 0.975 & \multicolumn{1}{l|}{0.97} & 0.992 & 0.984 & 0.981          \\
\multicolumn{1}{l|}{QFM(0.5)} & \textbf{0.974} & \textbf{0.956} & \multicolumn{1}{l|}{\textbf{0.949}} & 0.988 & 0.98 & \multicolumn{1}{l|}{0.976} & 0.991 & 0.982 & \multicolumn{1}{l|}{0.978}         & 0.991 & 0.983 & \multicolumn{1}{l|}{0.979}         & 0.994 & 0.989 & 0.986          \\
\multicolumn{1}{l|}{QFM(0.7)} & 0.963 & 0.936 & \multicolumn{1}{l|}{0.925} & 0.984 & 0.97 & \multicolumn{1}{l|}{0.965} & 0.987 & 0.975 & \multicolumn{1}{l|}{0.969}         & 0.988 & 0.975 & \multicolumn{1}{l|}{0.969}         & 0.992 & 0.984 & 0.98           \\
\multicolumn{1}{l|}{QFM(0.9)} & 0.801 & 0.648 & \multicolumn{1}{l|}{0.563} & 0.899 & 0.798         & \multicolumn{1}{l|}{0.734} & 0.933 & 0.863 & \multicolumn{1}{l|}{0.826}         & 0.939 & 0.879 & \multicolumn{1}{l|}{0.846}         & 0.96  & 0.919 & 0.899          \\
\multicolumn{1}{l|}{CQF-H}     & 0.974 & 0.954 & \multicolumn{1}{l|}{0.947} & 0.988 & 0.978 & \multicolumn{1}{l|}{0.975} & 0.990 & 0.982 & \multicolumn{1}{l|}{0.977} & 0.991 & 0.981 & \multicolumn{1}{l|}{0.976} & 0.994 & 0.987 & 0.985 \\
\multicolumn{1}{l|}{DAFM}     & 0.973 & 0.954 & \multicolumn{1}{l|}{0.947} & \textbf{0.989} & \textbf{0.98} & \multicolumn{1}{l|}{\textbf{0.977}} & \textbf{0.991} & \textbf{0.983} & \multicolumn{1}{l|}{\textbf{0.98}} & \textbf{0.992} & \textbf{0.983} & \multicolumn{1}{l|}{\textbf{0.98}} & \textbf{0.995} & \textbf{0.989} & \textbf{0.987}

\\ \hline
\multicolumn{14}{c}{$\epsilon_{it} \sim 0.5 \, N(-2,0.5)+0.5 \, N(2,0.5)$}  \\ \hline
\multicolumn{1}{l|}{Rob.PCA}  & 0.955 & 0.924 & \multicolumn{1}{l|}{0.905} & 0.976 & 0.96  & \multicolumn{1}{l|}{0.953} & 0.982 & 0.966 & \multicolumn{1}{l|}{0.957} & 0.982 & 0.967 & \multicolumn{1}{l|}{0.958} & 0.989 & 0.977 & 0.972 \\
\multicolumn{1}{l|}{QFM(0.1)} & 0.954 & 0.904 & \multicolumn{1}{l|}{0.876} & 0.987 & 0.974 & \multicolumn{1}{l|}{0.97}  & 0.994 & 0.987 & \multicolumn{1}{l|}{0.985} & 0.994 & 0.989 & \multicolumn{1}{l|}{0.986} & 0.997 & 0.993 & 0.992 \\
\multicolumn{1}{l|}{QFM(0.3)} & 0.922 & 0.872 & \multicolumn{1}{l|}{0.834} & 0.974 & 0.953 & \multicolumn{1}{l|}{0.942} & 0.985 & 0.972 & \multicolumn{1}{l|}{0.966} & 0.988 & 0.976 & \multicolumn{1}{l|}{0.97}  & 0.994 & 0.988 & 0.985 \\
\multicolumn{1}{l|}{QFM(0.5)} & 0.857 & 0.758 & \multicolumn{1}{l|}{0.661} & 0.91  & 0.813 & \multicolumn{1}{l|}{0.784} & 0.911 & 0.825 & \multicolumn{1}{l|}{0.771} & 0.908 & 0.825 & \multicolumn{1}{l|}{0.778} & 0.928 & 0.85  & 0.815 \\
\multicolumn{1}{l|}{QFM(0.7)} & 0.923 & 0.877 & \multicolumn{1}{l|}{0.832} & 0.974 & 0.951 & \multicolumn{1}{l|}{0.945} & 0.986 & 0.973 & \multicolumn{1}{l|}{0.967} & 0.988 & 0.977 & \multicolumn{1}{l|}{0.97}  & 0.994 & 0.988 & 0.985 \\
\multicolumn{1}{l|}{QFM(0.9)} & 0.96  & 0.924 & \multicolumn{1}{l|}{0.903} & 0.989 & 0.967 & \multicolumn{1}{l|}{0.968} & 0.994 & 0.988 & \multicolumn{1}{l|}{0.985} & 0.994 & 0.989 & \multicolumn{1}{l|}{0.986} & 0.996 & 0.993 & 0.991 \\
\multicolumn{1}{l|}{CQF-H}  & 0.966 & 0.939 & \multicolumn{1}{l|}{0.923 } & 0.987 &  0.978 & \multicolumn{1}{l|}{0.974} & 0.991 & 0.983 & \multicolumn{1}{l|}{0.978} & 0.992 & 0.985 & \multicolumn{1}{l|}{0.980} & 0.994 & 0.982 & 0.986 \\
\multicolumn{1}{l|}{DAFM} & \textbf{0.974} & \textbf{0.957} & \multicolumn{1}{l|}{\textbf{0.943}} & \textbf{0.991} & \textbf{0.984} & \multicolumn{1}{l|}{\textbf{0.982}} & \textbf{0.995} & \textbf{0.991} & \multicolumn{1}{l|}{\textbf{0.988}} & \textbf{0.996} & \textbf{0.992} & \multicolumn{1}{l|}{\textbf{0.99}} & \textbf{0.997} & \textbf{0.995} & \textbf{0.994}

\\
\hline

\multicolumn{14}{c}{$\epsilon_{it} \sim skewt(\xi=0, \omega=4, \alpha=4, df=3) $}  \\ \hline
\multicolumn{1}{l|}{Rob.PCA}  & 0.709 & 0.507 & \multicolumn{1}{l|}{0.398} & 0.828 & 0.659 & \multicolumn{1}{l|}{0.577} & 0.871 & 0.746 & \multicolumn{1}{l|}{0.687} & 0.881 & 0.771 & \multicolumn{1}{l|}{0.719} & 0.918 & 0.837 & 0.798          \\
\multicolumn{1}{l|}{QFM(0.1)} & 0.785 & 0.624 & \multicolumn{1}{l|}{0.53}  & 0.912 & 0.797 & \multicolumn{1}{l|}{0.767} & 0.956 & 0.917 & \multicolumn{1}{l|}{0.897} & 0.961 & 0.923 & \multicolumn{1}{l|}{0.905} & 0.976 & 0.951 & 0.939          \\
\multicolumn{1}{l|}{QFM(0.3)} & 0.869 & 0.773 & \multicolumn{1}{l|}{0.709} & 0.932 & \textbf{0.891} & \multicolumn{1}{l|}{0.854} & 0.96  & 0.924 & \multicolumn{1}{l|}{0.904} & 0.962 & 0.927 & \multicolumn{1}{l|}{0.91}  & 0.976 & 0.95  & 0.94           \\
\multicolumn{1}{l|}{QFM(0.5)} & 0.815 & 0.682 & \multicolumn{1}{l|}{0.6}   & 0.905 & 0.815 & \multicolumn{1}{l|}{0.779} & 0.936 & 0.873 & \multicolumn{1}{l|}{0.844} & 0.939 & 0.88  & \multicolumn{1}{l|}{0.854} & 0.959 & 0.918 & 0.899          \\
\multicolumn{1}{l|}{QFM(0.7)} & 0.702 & 0.506 & \multicolumn{1}{l|}{0.404} & 0.821 & 0.646 & \multicolumn{1}{l|}{0.56}  & 0.875 & 0.75  & \multicolumn{1}{l|}{0.671} & 0.883 & 0.773 & \multicolumn{1}{l|}{0.717} & 0.917 & 0.829 & 0.795          \\
\multicolumn{1}{l|}{QFM(0.9)} & 0.29  & 0.152 & \multicolumn{1}{l|}{0.1}   & 0.411 & 0.206 & \multicolumn{1}{l|}{0.153} & 0.49  & 0.226 & \multicolumn{1}{l|}{0.172} & 0.472 & 0.227 & \multicolumn{1}{l|}{0.155} & 0.602 & 0.303 & 0.207          \\
\multicolumn{1}{l|}{CQF-H}  & 0.866 & 0.769 & \multicolumn{1}{l|}{0.721} & 0.937 & 0.880 & \multicolumn{1}{l|}{0.863} & 0.954 & 0.913 & \multicolumn{1}{l|}{0.893} & 0.957 & 0.915 & \multicolumn{1}{l|}{0.892} & 0.972 & 0.940 &  0.931 \\
\multicolumn{1}{l|}{DAFM}     & \textbf{0.881} & \textbf{0.796} & \multicolumn{1}{l|}{\textbf{0.728}} & \textbf{0.945} & 0.891 & \multicolumn{1}{l|}{\textbf{0.871}} & \textbf{0.961} & \textbf{0.924} & \multicolumn{1}{l|}{\textbf{0.907}} & \textbf{0.964} & \textbf{0.928} & \multicolumn{1}{l|}{\textbf{0.913}} & \textbf{0.976} & \textbf{0.952} & \textbf{0.942}
         
  \\ \hline

\end{tabular}
}
\end{sidewaystable}

\begin{sidewaystable}[] 
\setlength{\tabcolsep}{11pt}
\def\arraystretch{1.3}
\caption{Location-scale-shift model: the adjusted $R^2$ of the true factors fitted by several methods according to different settings of $(N,T)$.}
\label{table:LS shift}
\scalebox{0.77}{
\begin{tabular}{llllllllllllllll}
(N,T) & \multicolumn{3}{c}{(50,50)} & \multicolumn{3}{c}{(100,50)} & \multicolumn{3}{c}{(100,100)} & \multicolumn{3}{c}{(100,150)} & \multicolumn{3}{c}{(150,150)} \\ \cmidrule(rl){2-4} \cmidrule(rl){5-7} \cmidrule(rl){8-10} \cmidrule(rl){11-13} \cmidrule(rl){14-16} 
 & $f_{1t}$ & $f_{2t}$ & $f_{3t}$ & $f_{1t}$ & $f_{2t}$ & $f_{3t}$  & $f_{1t}$ & $f_{2t}$ & $f_{3t}$ & $f_{1t}$ & $f_{2t}$ & $f_{3t}$  & $f_{1t}$ & $f_{2t}$ & $f_{3t}$  \\ \hline
\multicolumn{14}{c}{$\epsilon_{it} \sim N(0,1)$}  \\ \hline

\multicolumn{1}{l|}{Rob.PCA}  & 0.991 & 0.984 & \multicolumn{1}{l|}{0.048}         & 0.995 & 0.991 & \multicolumn{1}{l|}{0.053} & 0.996 & 0.992 & \multicolumn{1}{l|}{0.041} & 0.996 & 0.992 & \multicolumn{1}{l|}{0.012} & 0.997      & 0.995 & 0.022          \\
\multicolumn{1}{l|}{QFM(0.1)} & 0.985 & 0.967 & \multicolumn{1}{l|}{0.801}         & 0.994 & 0.97  & \multicolumn{1}{l|}{0.889} & 0.986 & 0.983 & \multicolumn{1}{l|}{0.916} & 0.996 & 0.993 & \multicolumn{1}{l|}{0.923} & 0.998      & 0.996 & 0.947          \\
\multicolumn{1}{l|}{QFM(0.3)} & 0.996 & 0.994 & \multicolumn{1}{l|}{0.575}         & 0.998 & 0.997 & \multicolumn{1}{l|}{0.764} & 0.999 & 0.998 & \multicolumn{1}{l|}{0.773} & 0.999 & 0.998 & \multicolumn{1}{l|}{0.778} & 0.999      & 0.999 & 0.842          \\
\multicolumn{1}{l|}{QFM(0.5)} & 0.997 & 0.995 & \multicolumn{1}{l|}{0.001}         & 0.999 & 0.998 & \multicolumn{1}{l|}{0.002} & 0.999 & 0.998 & \multicolumn{1}{l|}{0.004} & 0.999 & 0.998 & \multicolumn{1}{l|}{0.001} & 0.999      & 0.999 & 0.001          \\
\multicolumn{1}{l|}{QFM(0.7)} & 0.996 & 0.993 & \multicolumn{1}{l|}{0.589}         & 0.998 & 0.997 & \multicolumn{1}{l|}{0.765} & 0.999 & 0.997 & \multicolumn{1}{l|}{0.779} & 0.999 & 0.998 & \multicolumn{1}{l|}{0.78}  & 0.999      & 0.999 & 0.843          \\
\multicolumn{1}{l|}{QFM(0.9)} & 0.981 & 0.962 & \multicolumn{1}{l|}{0.815}         & 0.992 & 0.983 & \multicolumn{1}{l|}{0.891} & 0.996 & 0.992 & \multicolumn{1}{l|}{0.917} & 0.996 & 0.993 & \multicolumn{1}{l|}{0.923} & 0.998      & 0.987 & 0.948          \\
\multicolumn{1}{l|}{CQF-H}  & 0.995 & 0.992 & \multicolumn{1}{l|}{0.058} & 0.998 & 0.996 & \multicolumn{1}{l|}{0.067} & 0.998 & 0.996 & \multicolumn{1}{l|}{0.043} & 0.998 & 0.996 & \multicolumn{1}{l|}{0.013} & 0.999 & \textbf{1} & 0.026 \\
\multicolumn{1}{l|}{DAFM}     & \textbf{0.998} & \textbf{0.997} & \multicolumn{1}{l|}{\textbf{0.92}} & \textbf{0.999} & \textbf{0.999} & \multicolumn{1}{l|}{\textbf{0.953}} & \textbf{0.999} & \textbf{0.999} & \multicolumn{1}{l|}{\textbf{0.965}} & \textbf{0.999} & \textbf{0.999} & \multicolumn{1}{l|}{\textbf{0.967}} & \textbf{1} & 0.999 & \textbf{0.977}

 \\ \hline
\multicolumn{14}{c}{$\epsilon_{it} \sim t(2) $}  \\ \hline
\multicolumn{1}{l|}{Rob.PCA}  & 0.913 & 0.84  & \multicolumn{1}{l|}{0.038} & 0.961 & 0.924 & \multicolumn{1}{l|}{0.036} & 0.961 & 0.918 & \multicolumn{1}{l|}{0.028}         & 0.955 & 0.923 & \multicolumn{1}{l|}{0.01}  & 0.973 & 0.949 & 0.007         \\
\multicolumn{1}{l|}{QFM(0.1)} & 0.853 & 0.721 & \multicolumn{1}{l|}{0.528} & 0.932 & 0.833 & \multicolumn{1}{l|}{0.705} & 0.97  & 0.908 & \multicolumn{1}{l|}{0.771}         & 0.975 & 0.939 & \multicolumn{1}{l|}{0.781} & 0.989 & 0.975 & 0.863         \\
\multicolumn{1}{l|}{QFM(0.3)} & 0.992 & 0.986 & \multicolumn{1}{l|}{0.486} & 0.997 & 0.995 & \multicolumn{1}{l|}{0.675} & 0.998 & 0.995 & \multicolumn{1}{l|}{0.731}         & 0.998 & 0.996 & \multicolumn{1}{l|}{0.726} & 0.999 & 0.997 & 0.804         \\
\multicolumn{1}{l|}{QFM(0.5)} & 0.994 & \textbf{0.991} & \multicolumn{1}{l|}{0.002} & \textbf{0.998} & \textbf{0.997} & \multicolumn{1}{l|}{0.002} & 0.998 & 0.997 & \multicolumn{1}{l|}{0.004}         & 0.999 & 0.997 & \multicolumn{1}{l|}{0.001} & 0.999 & 0.998 & 0             \\
\multicolumn{1}{l|}{QFM(0.7)} & 0.991 & 0.986 & \multicolumn{1}{l|}{0.459} & 0.997 & 0.995 & \multicolumn{1}{l|}{0.682} & 0.998 & 0.996 & \multicolumn{1}{l|}{0.72} & 0.998 & 0.995 & \multicolumn{1}{l|}{0.73}  & 0.999 & 0.998 & 0.805         \\
\multicolumn{1}{l|}{QFM(0.9)} & 0.868 & 0.7   & \multicolumn{1}{l|}{0.522} & 0.927 & 0.854 & \multicolumn{1}{l|}{0.7}   & 0.959 & 0.915 & \multicolumn{1}{l|}{0.769}         & 0.975 & 0.938 & \multicolumn{1}{l|}{0.78}  & 0.987 & 0.959 & 0.857         \\
\multicolumn{1}{l|}{CQF-H}  & 0.990 & 0.982 & \multicolumn{1}{l|}{0.053} & 0.996 & 0.993 & \multicolumn{1}{l|}{0.091} & 0.997 & 0.993 & \multicolumn{1}{l|}{0.047} & 0.997 & 0.994 & \multicolumn{1}{l|}{0.018} & 0.998 & 0.996 & 0.018 \\
\multicolumn{1}{l|}{DAFM}     & \textbf{0.995} & 0.99  & \multicolumn{1}{l|}{\textbf{0.797}} & 0.998 & 0.997 & \multicolumn{1}{l|}{\textbf{0.879}} & \textbf{0.999} & \textbf{0.997} & \multicolumn{1}{l|}{\textbf{0.92}} & \textbf{0.999} & \textbf{0.997} & \multicolumn{1}{l|}{\textbf{0.925}} & \textbf{0.999} & \textbf{0.999} & \textbf{0.95}

\\ \hline 
\multicolumn{14}{c}{$\epsilon_{it} \sim 0.5 \, N(-2,0.5)+0.5 \, N(2,0.5)$}  \\ \hline

\multicolumn{1}{l|}{Rob.PCA}  & 0.957 & 0.914 & \multicolumn{1}{l|}{0.041}         & 0.977 & 0.954 & \multicolumn{1}{l|}{0.05}  & 0.982 & 0.963 & \multicolumn{1}{l|}{0.042} & 0.983 & 0.966 & \multicolumn{1}{l|}{0.018} & 0.989 & 0.977 & 0.017          \\
\multicolumn{1}{l|}{QFM(0.1)} & 0.987 & 0.956 & \multicolumn{1}{l|}{0.963}         & 0.981 & 0.981 & \multicolumn{1}{l|}{0.983} & 0.992 & 0.975 & \multicolumn{1}{l|}{0.989} & 0.998 & 0.997 & \multicolumn{1}{l|}{0.99}  & 0.999 & 0.998 & 0.994          \\
\multicolumn{1}{l|}{QFM(0.3)} & 0.985 & 0.97  & \multicolumn{1}{l|}{0.81} & 0.996 & 0.994 & \multicolumn{1}{l|}{0.965} & 0.998 & 0.996 & \multicolumn{1}{l|}{0.972} & 0.998 & 0.996 & \multicolumn{1}{l|}{0.971} & 0.999 & 0.998 & 0.986          \\
\multicolumn{1}{l|}{QFM(0.5)} & 0.942 & 0.836 & \multicolumn{1}{l|}{0.017}         & 0.967 & 0.913 & \multicolumn{1}{l|}{0.011} & 0.975 & 0.944 & \multicolumn{1}{l|}{0.008} & 0.977 & 0.951 & \multicolumn{1}{l|}{0.001} & 0.985 & 0.967 & 0.002          \\
\multicolumn{1}{l|}{QFM(0.7)} & 0.987 & 0.978 & \multicolumn{1}{l|}{0.822}         & 0.996 & 0.994 & \multicolumn{1}{l|}{0.97}  & 0.998 & 0.996 & \multicolumn{1}{l|}{0.977} & 0.998 & 0.996 & \multicolumn{1}{l|}{0.975} & 0.999 & 0.998 & 0.986          \\
\multicolumn{1}{l|}{QFM(0.9)} & 0.977 & 0.937 & \multicolumn{1}{l|}{0.963}         & 0.993 & 0.988 & \multicolumn{1}{l|}{0.986} & 0.987 & 0.991 & \multicolumn{1}{l|}{0.99}  & 0.998 & 0.997 & \multicolumn{1}{l|}{0.99}  & 0.996 & 0.982 & 0.993          \\
\multicolumn{1}{l|}{CQF-H}  & 0.966  & 0.928 & \multicolumn{1}{l|}{0.043} & 0.981 & 0.963 & \multicolumn{1}{l|}{0.065} & 0.985 & 0.970 & \multicolumn{1}{l|}{0.040} & 0.985 & 0.971 & \multicolumn{1}{l|}{0.024} & 0.990 & 0.979 & 0.026 \\
\multicolumn{1}{l|}{DAFM}     & \textbf{0.995} & \textbf{0.991} & \multicolumn{1}{l|}{\textbf{0.99}} & \textbf{0.998} & \textbf{0.997} & \multicolumn{1}{l|}{\textbf{0.995}} & \textbf{0.999} & \textbf{0.998} & \multicolumn{1}{l|}{\textbf{0.996}} & \textbf{0.999} & \textbf{0.998} & \multicolumn{1}{l|}{\textbf{0.996}} & \textbf{0.999} & \textbf{0.999} & \textbf{0.997}
\\ \hline
\end{tabular}
}
\end{sidewaystable}

\vspace{-1mm} 
\subsection{Location-scale-shift model} \label{subsec:LS shift model}
\vspace{-1mm} 
We now consider a data-generating process of a three-factor model,
\vspace{-2mm} 
\[ 
X_{it} = \lambda_{1i} f_{1t} + \lambda_{2i} f_{2t} + \lambda_{3i} f_{3t} \epsilon_{it}. 
\] 
\vspace{-2mm} 
It is a location-scale-shift model where the first two factors affect the location and the third factor, $f_{3t}$, affects the scale of the data by being involved in the error term. We consider the case of $\lambda_{3i}f_{3t}>0$, by generating 
$f_{1t}=0.8 f_{1,t-1}+ e_{1t}, ~\, f_{2t}=0.5 f_{2,t-1}+ e_{2t}, ~\, f_{3t}=|e_{3t}|, \; e_{1t}, e_{2t}, e_{3t} \overset{i.i.d}{\sim} N(0,1)$,
and
$ 
\lambda_{1i}, \lambda_{2i} \overset{i.i.d}{\sim} N(0,1), \quad \lambda_{3i} =|\alpha_i|  \, \text{ for } \, \alpha_i \sim N(0,1).
$
This model can be written in DAFM form as $Q_{X_{it}}(\tau|f_t) = \lambda_{k,i}^{\prime}f_t$, with $\lambda_{k,i}=(\lambda_{1i},\lambda_{2i},\lambda_{3i}\cdot Q_{\epsilon_{it}}(\tau))^{\prime}$ and $f_t=(f_{1t}, f_{2t}, f_{3t})^{\prime}$. In this setting, we consider three noise types of $\epsilon_{it}$, a normal distribution, $N(0,1)$, $t(2)$ distribution, and the Gaussian mixture distribution, $0.5 \cdot N(2,0.5) + 0.5 \cdot N(-2,0.5)$. The combinations of $(N,T)$ and the choice of quantile levels are the same as in the previous location-shift model. The adjusted $R^2$ values of regressing each true factor on the estimated factors are listed in Table \ref{table:LS shift}.

If the $\tau$-th quantile of $\epsilon_{it}$ is zero when performing QFM, then QFM at $\tau$ would fail to identify $f_{3t}$ since $Q_{X_{it}}(\tau|f_t)=\lambda_{i1} f_{1t} + \lambda_{i2} f_{2t}$. In other words, considering the $\tau$-th quantile of the data masks the information in $f_{3t}$. Thus, the $R^2$ values of regressing $f_{3t}$ on the estimated factors by QFM(0.5) are all close to zero. Also, QFMs at $\tau=0.3$ and 0.7, which are near 0.5, do not perform well enough to estimate $f_{3t}$. On the other hand, QFMs at $\tau=0.1$ and 0.9 do a relatively good job of identifying $f_{3t}$ but are less accurate at capturing the other factors, $f_{1t}$ and $f_{2t}$. Next, when applying the proposed DAFM, the weighting scheme used in the location-shift model is not appropriate for this case. Specifically, assigning a high weight to the 0.5-th quantile based on its high density only contributes to the estimation of $f_{1t}$ and $f_{2t}$, while it significantly hinders the estimation of $f_{3t}$. Thus, the weights assigned to the quantile levels affect the estimation of each factor in different ways, so assigning weights based on their density in this case does not guarantee good factor estimation overall. Therefore, we choose uniform weights in this model. As shown in Table \ref{table:LS shift}, DAFM effectively improves the estimation of all three factors compared to QFM and yields the highest $R^2$ values, especially in capturing $f_{3t}$. In contrast, CQF-H fails to identify $f_{3t}$, yielding $R^2$ values close to zero, which indicates that incorporating quantile-dependent intercepts with fixed loadings is insufficient to capture error-embedded factors; however, it can be effectively identified by allowing loadings to vary across quantiles. 

Overall, the first simulation results show that DAFM better handles the cases in which CQF-H works well, and the second results highlight that the proposed model performs effectively even when CQF-H fails. This indicates that the proposed model can accommodate a broader range of data and attain more accurate factor estimation than CQF-H.

\section{Real data analysis} \label{sec:Real data analysis}

In this section, we apply the proposed model to two economic datasets to demonstrate its effectiveness in explaining the underlying data structure and forecasting future data. For real data analysis, we consider not only the quantile levels of 0.1, 0.3, 0.5, 0.7, and 0.9 but also the extreme levels of 0.01 and 0.99, as \cite{chen2021quantile} identified that factors derived from extreme quantiles provide useful information in certain economic datasets. We compare a uniformly weighted DAFM, called $\text{DAFM}_{unif}$, with QFMs applied at each quantile level, CQF-H with the same set of quantile levels as DAFM, and the mean-based factor model. We also consider three different weighting schemes for DAFM as follows: (i)  $\text{DAFM}_{low}$: assigns twice the weight to the lower quantiles, 0.01 and 0.1; (ii) $\text{DAFM}_{med}$: assigns twice the weight to the medium quantiles, 0.3, 0.5, and 0.7;  (iii) $\text{DAFM}_{high}$: assigns twice the weight to the upper quantiles, 0.9 and 0.99.

\subsection{Identifying hidden structures in monthly stock returns} 
\label{subsec:Volatility}

We now present an example of how DAFM effectively extracts information regarding the data structure, particularly the structure of idiosyncratic volatility. The dataset we use consists of monthly stock return data from the Center for Research in Security Prices (CRSP) for the most recent 20 years, from January 2004 to December 2023, which implies $T=240$ time points. It consists of $N=233$ firms with ordinary common shares, which have share codes of 12 and 14, and includes only firms with no missing values. Share codes 12 and 14 represent companies incorporated outside the U.S. and closed-end funds, respectively. We select these share codes to ensure an appropriate data size, as other codes among ordinary common shares include either too many or too few firms for suitable analysis. The CRSP stock return data are not publicly available, but can be accessed through institutions that subscribe to WRDS (Wharton Research Data Services) at \url{https://wrds-www.wharton.upenn.edu/pages/get-data/center-research-security-prices-crsp}.

We first estimate the number of factors. For the mean factor model and CQF-H, the number of factors is determined using the criterion by \cite{bai2002determining} and \cite{huang2023composite}, respectively. For QFM, we use the rank estimator of \cite{chen2021quantile}, with $\kappa_{NT}=L_{NT}^{-2/3}\cdot  \hat{\sigma}_1(N)$, where $\hat{\sigma}_1(N)$ denotes the largest eigenvalue of $\hat{\Lambda}^\prime \hat{\Lambda}/{N} $. The estimated number of factors is notably influenced by the choice of $\kappa_{NT}$. Thus, for DAFM, we select a similar form of the estimator as used in QFM to ensure a fair comparison. Specifically, we use $\kappa_{NT}(k) = L_{NT}^{-2/3} \cdot  \hat{\sigma}_1^k(N)$ in Theorem \ref{thm:factornumber2}. Note that the proof of Theorem \ref{thm:factornumber2} remains valid even though the chosen $\kappa_{NT}$ depends on $k$, given the additional assumption that the eigenvalues $\hat{\sigma}_1^k(N)$ converge as $N \rightarrow \infty$ for all $k$, which can also be implied by Assumption \ref{assump:normality}(b). The estimated number of factors is listed in Table \ref{table:Factor number}. The number of factors is determined to be 2 for DAFM in all four weighting schemes, and for QFMs, 1 or 2 is selected, depending on the quantile level $\tau$. For CQF-H, it is determined to be 1, and the number of mean factors is chosen to be 8.

\begin{table}[]
\centering
\setlength{\tabcolsep}{8pt}
\caption{The estimated number of factors for the CRSP stock return data under different factor models. DAFM, with four weighting schemes, DAFMs, estimates the number of factors as 2.}
\label{table:Factor number} \footnotesize
{\renewcommand{\arraystretch}{1.5}
\begin{tabular}{cccccccccc}
\hline
\multirow{2}{*}{DAFMs} & \multicolumn{7}{c}{QFM($\tau$)}  & \multirow{2}{*}{CQF-H}  &  \multirow{2}{*}{Mean FM} \\ \cline{2-8}
          &  \multicolumn{1}{c}{$0.01$} & \multicolumn{1}{c}{0.1} & \multicolumn{1}{c}{0.3} & \multicolumn{1}{c}{0.5} & \multicolumn{1}{c}{0.7} & \multicolumn{1}{c}{0.9} & 0.99 &  \\ \hline
2  & 1 & \multicolumn{1}{c}{1}    & \multicolumn{1}{c}{2}   & \multicolumn{1}{c}{2}   & \multicolumn{1}{c}{2}   & \multicolumn{1}{c}{1}   & \multicolumn{1}{c}{1}  &1  & 8 \\ \hline
\end{tabular}
}
\end{table}

In the literature, \cite{herskovic2016common} explored a common factor structure in the idiosyncratic volatilities of stocks across U.S. firms. They found that after removing common mean factors, the volatilities of residuals exhibit similar patterns of variation across firms. In particular, the common idiosyncratic volatility (CIV) computed from these residuals can be a common factor of the volatilities and has been identified as containing vital information. The CIV is computed as follows. First, to estimate monthly firm-level idiosyncratic volatility, we fit the 
linear regression model to the daily stock return data to obtain residuals, 
$x_{it} = \beta \hat{F}_t + \epsilon_{it},$ for $i=1,\ldots,233,$ 
where $\hat{F}_t$ is the estimated mean factor from data at day $t$ in the corresponding month. The idiosyncratic volatility for each month is then estimated as the standard deviation of the residuals in that month. Finally, CIV is computed as the equally weighted average of the idiosyncratic volatilities across different firms.
\cite{herskovic2016common} demonstrated that shocks to CIV affect asset prices and help explain some deviations from the expected behavior of asset prices. In addition, they empirically showed that CIV can act as a proxy for the idiosyncratic risks faced by individual consumers, with evidence linking CIV to household income risk, employment risk, house price dispersion, and wage growth per job. They also indicated that a firm's exposure to CIV shocks can account for variations in average stock returns among firms. These findings underscore the significance of CIVs in capturing a broad spectrum of economic risks impacting both firms and households.

\begin{figure}
    \centering
    \includegraphics[width=0.65\linewidth]{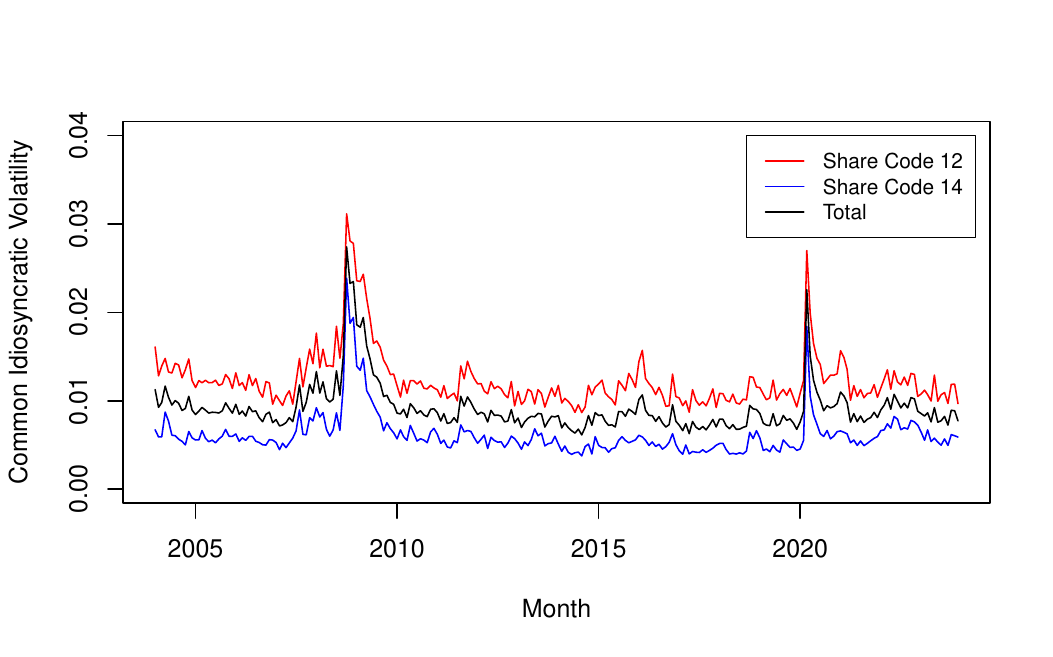}
    \vspace{-3mm}
    \caption{Time series plots of common idiosyncratic volatility computed by averaging the volatilities within three groups of firms.}
    \label{fig:CIVs}
\end{figure}

Figure \ref{fig:CIVs} displays the average idiosyncratic volatilities for all firms and two subsets of firms, identified by share codes 12 and 14, respectively. These CIVs exhibit similar dynamics across all groups of firms, aligning with the findings discussed in \cite{herskovic2016common}, which indicate common variation in firm-level idiosyncratic volatilities regardless of divisions by size and industry category. 

Since there is a strong factor structure in idiosyncratic volatilities even after removing common mean factors, we can expect the stock returns to follow a location-scale shift model involving a factor structure coupled with an error term. In this analysis, the DAFM factors estimated from the monthly stock return data effectively extract some information about CIV, unlike the QFM, CQF-H, and mean PCA factors. To evaluate the explanatory power of the estimated factors for CIV, we compute the adjusted $R^2$ values of regressing CIV on the estimated factors. We calculate CIV not only from the full dataset but also separately for firms with share codes 12 and 14. The results in Table \ref{table:CIV results} indicate that the mean factor and the factors obtained by CQF-H have almost no relationship with CIV, with an $R^2$ value close to zero. The factors estimated by QFM at lower quantiles tend to explain CIV better, but the highest $R^2$ value is only around 0.4. 
\begin{table}[]
\centering
{\renewcommand{\arraystretch}{1.5}
\setlength{\tabcolsep}{6pt}
\caption{Adjusted $R^2$ values of regressing CIV on the estimated factors of various factor models.}
\label{table:CIV results} \footnotesize
\begin{tabular}{lcccccccccc}
\hline 
& \multirow{2}{*}{$\text{DAFM}_{unif}$}  & \multicolumn{7}{c}{QFM($\tau$)}                           & \multirow{2}{*}{$\text{CQF-H}$} & \multicolumn{1}{l}{\multirow{2}{*}{Mean FM}} \\ \cline{3-9}
                    & &  \multicolumn{1}{l}{0.01} & \multicolumn{1}{l}{0.1} & \multicolumn{1}{l}{0.3} & \multicolumn{1}{l}{0.5} & \multicolumn{1}{l}{0.7} & \multicolumn{1}{l}{0.9} & \multicolumn{1}{l}{0.99} & \multicolumn{1}{l}{} &  \\ \hline
Share code 12 & .743 &  .347  & .276 & .110 & .040 & .004 & .007 & .139  & .036 & .066        \\ 
Share code 14 & .712 &  .417 & .369 & .183 & .080 & .012 & .020 & .072 &  .048 & .101     \\
Total & .751 &  .391 & .327 & .143 & .056 & .004 & .048 & .110 &  .043 & .080                 \\ \hline
\end{tabular}
}
\end{table}

On the other hand, the factors estimated by $\text{DAFM}_{unif}$ exhibit significantly higher $R^2$ values exceeding 0.7. This trend remains consistent across the different groups of firms used to compute CIV. We also estimate DAFM factors with different weighting schemes and compare the adjusted $R^2$ values listed in Table \ref{table:CIV results_weight}. Incorporating the three new weighting schemes does not significantly impact $R^2$ values, and the models still explain CIV effectively. However, we observe a consistent pattern that $\text{DAFM}_{low}$ yields the highest $R^2$, followed by $\text{DAFM}_{high}$, and then $\text{DAFM}_{med}$. This is reasonable, given that $R^2$ values obtained from QFM at lower quantiles are notably high, and QFM(0.99) tends to obtain higher $R^2$ values than QFM(0.5) and QFM(0.7). This indicates that the lower, upper, and middle quantiles of the data provide more useful information about CIV in that order.

\begin{table}
\centering
\setlength{\tabcolsep}{8pt}
\caption{Adjusted $R^2$ values of regressing CIV on the estimated factors of DAFMs.}
\label{table:CIV results_weight} \footnotesize
{\renewcommand{\arraystretch}{1.5}
\begin{tabular}{lcccc}
\hline
 & $\text{DAFM}_{unif}$ & $\text{DAFM}_{low}$ & $\text{DAFM}_{med}$ & $\text{DAFM}_{high}$ \\ \hline
 Share code 12 & .745 & .756 & .733 & .748 \\
 Share code 14 & .714 & .739 & .693 & .720 \\
 Total & .745 & .771 & .736 & .758 \\ \hline
\end{tabular}
}
\end{table}

In summary, this analysis shows that the proposed DAFM captures crucial information from the stock return data that is inaccessible through QFMs at a single quantile level, CQF-H, and the mean factor model. These findings are consistent with the simulation results of the location-scale shift model, where DAFM effectively extracts factors compared to other methods, particularly the factor embedded in the error term. In contrast, the robust PCA method and CQF-H fail to identify it. We also expect that DAFM can be effectively applied to firms' cash flow data, as \cite{herskovic2016common} showed the existence of co-movement in the idiosyncratic volatility of cash flows. 

\subsection{Forecasting monthly unemployment rate} 

We aim to make short-term forecasts of the U.S. monthly unemployment rate, which is a key economic indicator that reflects the overall health of the economy and is closely monitored for making policy decisions. A substantial body of literature focuses on forecasting the monthly unemployment rate \citep{maas2020short, costa2024real, wang2023variable}. In this analysis, we extract DAFM factors from panel data of numerous economic variables and use them as predictors for forecasting the unemployment rate. This approach extends the forecasting method of \cite{stock2002forecasting} by replacing the mean PCA factor with DAFM factors. We utilize the FRED-MD data from \cite{mccracken2016fred}, which are publicly available at \url{https://www.stlouisfed.org/research/economists/mccracken/fred-databases}. The dataset consists of monthly observations on 127 economic variables spanning 240 months from January 2000 to December 2019. Data collected after the COVID-19 pandemic was excluded because we observed that the behavior of short-term forecasts was disrupted when using data that may contain potential outliers due to the pandemic. Each time series is transformed to ensure stationarity through a process by \cite{mccracken2016fred}, such as taking a difference or logarithm. Additionally, the unemployment rate time series is non-stationary and requires differentiation to render it stationary.

Let $\{y_t\}$ and $\{\hat{F}_t\}$ denote the time series of the unemployment rate and the estimated factors, respectively. To make a $h$-month ahead forecast, we fit a factor-augmented regression model where the response variable is $y_{t+h}-y_t$ and the predictors include the estimated factors $\hat{F}_t$ along with the lags of $y_t-y_{t-1}$ as,
\begin{equation} \label{eq:lm forecast}
    y_{t+h}-y_t = \alpha+ \sum\limits_{m=0}^p \beta_{y,m}\,(y_{t-m}-y_{t-m-1}) + \beta_F^{\prime} \hat{F}_t+\epsilon_{t+h},
\end{equation}
where $\alpha, \beta_{y,0}, \ldots, \beta_{y,p} \in \mathbb{R}$ and $\beta_F \in \mathbb{R}^r$. The number of lags, $p$, is determined by the BIC of the linear regression model, where the maximum number of lags is set to 4. The forecast horizon $h$ is specified as $h = 1,\ldots,6$. For each prediction, a rolling window of 120 months is used to estimate the factors and fit the linear regression model. For instance, the first rolling window spans from January 2000 to December 2009 and is used to forecast $h$-months from December 2009. Then, for each $h$, we forecast the unemployment rate for the period extending from $h$-months ahead of December 2009 to December 2019, resulting in a total of $(120-h)$ months of forecasts. 
\begin{table}
\centering
\setlength{\tabcolsep}{8pt}
\caption{The number of factors for the FRED-MD data estimated by different factor models. DAFM, with four weighting schemes, DAFMs, estimates the factor number as 7.}
\label{table:Factor number_forecast} \footnotesize
{\renewcommand{\arraystretch}{1.5}
\begin{tabular}{cccccccccc} 
\hline
\multirow{2}{*}{DAFMs} & \multicolumn{7}{c}{QFM($\tau$)}  & \multirow{2}{*}{CQF-H}  &  \multirow{2}{*}{Mean FM} \\ \cline{2-8}
          &  \multicolumn{1}{c}{$0.01$} & \multicolumn{1}{c}{0.1} & \multicolumn{1}{c}{0.3} & \multicolumn{1}{c}{0.5} & \multicolumn{1}{c}{0.7} & \multicolumn{1}{c}{0.9} & 0.99 &  \\ \hline
7   & \multicolumn{1}{c}{1}    & \multicolumn{1}{c}{2}   & \multicolumn{1}{c}{6}   & \multicolumn{1}{c}{6}   & \multicolumn{1}{c}{7}   & \multicolumn{1}{c}{2}  &1 & 6 & 8 \\ \hline
\end{tabular}
}
\end{table}

For comparison, we consider the following four methods: (i) AR: AR model as the benchmark, which incorporates only the lagged values of $y_t$ to forecast $y_{t+h}$. (ii) AR+PCA: use the estimators of the mean PCA factor for $\hat{F}_t$ in (\ref{eq:lm forecast}), which aligns with the forecasting approach of \cite{stock2002forecasting}. (iii) AR+CQF-H: use the estimated factor from CQF-H. (iv) AR+QFM($\tau$): include the estimators of quantile factors at each quantile level $\tau$. The proposed forecasting method, named  `AR+$\text{DAFM}_{(\cdot)}$', uses the estimators obtained by $\text{DAFM}_{(\cdot)}$ for $\hat{F}_t$ in (\ref{eq:lm forecast}), depending on the weighting scheme ($\cdot$). The number of factors is estimated on the full dataset using the same method described in Section \ref{subsec:Volatility}, and the results for different factor models are listed in Table \ref{table:Factor number_forecast}. The number of factors in DAFM remained at 7, regardless of the weighting schemes. This estimated number of factors is applied consistently across all rolling windows for predictions. 

\begin{table}
\centering
{\renewcommand{\arraystretch}{1.2}
\setlength{\tabcolsep}{10pt}
\caption{Relative mean squared error (MSE) results for all forecasts compared to the benchmark AR over different forecast horizons $h=1,\ldots,6$. The actual MSE values of AR are shown in parentheses. For each $h$, the lowest relative MSE is highlighted in bold.}
\label{table:Relative MSE} \footnotesize
\begin{tabular}{lrrrrrr}
\hline
               & $h = 1$ & $h = 2$ & $h = 3$ & $h = 4$ & $h = 5$ & $h = 6$ \\ \hline
AR             & 1 (.160) & 1 (.215)& 1 (.261) & 1 (.296) & 1 (.350) & 1 (.406) \\ 
AR + PCA       & .861  & .790  & .759  & .883  & .926  & .933  \\ \medskip
AR + CQF-H     & .892  & .822  & .768  & .924  & .927 & .935 \\ 
AR + QFM(0.01) & 1.288 & 1.551 & 1.419 & 1.452 & 1.453 & 1.221 \\
AR + QFM(0.1)  & .856  & .873  & .979  & 1.091 & 1.069 & 1.027 \\
AR + QFM(0.3)  & \textbf{.820}  & .733  & .768  & .981  & 1.119 & 1.284 \\
AR + QFM(0.5)  & .896  & .723  & .708  & .837  & .854  & .927  \\
AR + QFM(0.7)  & .915  & .824  & .763  & .876  & .870  & .912  \\
AR + QFM(0.9)  & .975  & .915  & .850  & .959  & .930  & .958  \\ \medskip
AR + QFM(0.99) & 1.076 & 1.103 & 1.142 & 1.133 & 1.080 & 1.014 \\ 
AR + $\text{DAFM}_{unif}$ & .862  & \textbf{.722}  & .712  & .811  & .776  & .818  \\
AR + $\text{DAFM}_{low}$ & .857  & .728  & \textbf{.706}  & \textbf{.787}  & \textbf{.770}  & \textbf{.808}  \\
AR + $\text{DAFM}_{med}$ & .889  & .737  & .720  & .787  & .778  & .816  \\
AR + $\text{DAFM}_{high}$ & .859  & .767  & .743  & .883  & .830  & .871 \\ \hline
\end{tabular}
}
\end{table}

The forecasting results are evaluated by computing the mean squared error (MSE) for all predictions. The relative MSE calculated with respect to the benchmark AR is listed in Table \ref{table:Relative MSE}. The actual MSE values of AR are shown in row 2 in parentheses. We observe that the performance of AR+QFM($\tau$) is highly dependent on the quantile level $\tau$, performing relatively well at the quantile levels of 0.3, 0.5, and 0.7. However, it does not exhibit a consistent pattern across different values of $h$. For $h=2$, 4, 5, and 6, AR+$\text{DAFM}_{unif}$ outperforms all methods that do not include the DAFM factors, with a particularly significant difference in accuracy at higher values of $h$. When employing different weighting schemes for DAFM, AR+$\text{DAFM}_{low}$ tends to perform even better than AR+$\text{DAFM}_{unif}$. On the other hand, AR+$\text{DAFM}_{high}$ achieves relatively low accuracy but still outperforms or is comparable to the other models. Based on these results, we conclude that DAFM factors contribute to more accurate predictions of the unemployment rate compared to the QFM, CQF-H, and mean PCA factors.

\section{Concluding remarks} \label{sec:Conclusion}
In this paper, we propose a new factor model that identifies common factor structures by simultaneously incorporating multiple quantile levels. While the quantile factor model of \cite{chen2021quantile} leverages the check function to extract distributional features at a fixed quantile, the proposed data-adaptive factor model (DAFM) overcomes this limitation by integrating distributional features of the entire data distribution. This composite quantile approach enables DAFM to more effectively identify common factors that represent the overall data structure and uncover richer latent patterns that may remain hidden at the mean or any single quantile level. Consequently, DAFM yields highly data-adaptive factors that better reflect the complexity of real-world data
with various distributions, providing a more comprehensive understanding of the underlying structure.

Several future research directions could further enhance DAFM. First, developing an appropriate weighting scheme that optimizes the contribution of the loss function at each quantile is essential for improving the overall performance of the proposed model. Future work could focus on identifying a weighting scheme that adjusts the influence of each quantile based on its significance in the dataset. Second, exploring an expectile-based factor model that better captures the distributional information of the data at the extreme quantile levels would be beneficial. Finally, DAFM could be extended to dynamic settings or applied to a broader range of data to further broaden its utility.

\section*{Supplementary material}
{\small Web-based supporting materials: Proofs of Theorems 1, 2, 3, and 4 in the paper.}

\section*{Acknowledgments}
\spacingset{1.4}{\small This research was supported by the National Research Foundation of Korea (NRF) funded by the Korea government (2021R1A2C1091357; RS-2024-00337752).}

\setlength{\bibsep}{0pt}
\bibliographystyle{apalike}
\bibliography{cqfm}

\clearpage
\appendix

\numberwithin{equation}{section}
\renewcommand\theequation{\thesection.\arabic{equation}}
\newtheorem{theorem}{Theorem}[section]
\renewcommand\thetheorem{\thesection.\arabic{theorem}}
\renewcommand\thelemma{\thesection.\arabic{lemma}}

\noindent
\begin{center}
\textbf{\Large Supplementary material for ``A Data-Adaptive Factor Model Using Composite Quantile Approach''}
\vskip 7mm


\end{center}

\section{Proof of Theorem \ref{thm:consistency}} 

\noindent For $\th_a, \th_b \in \Theta^r$ and $\th \in \Theta^r$, we define the following notations.
\begin{itemize}
    \item $d(\th_a, \th_b)\defeq \sqrt{\frac{1}{NT}\sumkit (\lam k i^{a^{\prime}} f_t^a -\lam k i^{b\p} f_t^b)^2} $
    \item $M_{NT}^*(\th) \defeq M_{NT}(\th)-M_{NT}(\th_0)$
    \item $\W_{NT}(\th) \defeq M_{NT}^*(\th)-\E[\,M_{NT}^*(\th)\,]$
    \item $\norm{Y}_{\psi_2}\defeq\inf\{C>0:\E\psi_2(|Y|/C)\leq 1\}$, where $\psi_2(x)=e^{x^2}-1$
\end{itemize}
Since $\mathcal{A}$ and $\mathcal{F}$ are compact, $\forall \lam k i \in \mathcal{A}$ and $\forall f_t \in \mathcal{F}$ are bounded. Define $M_1$ as the constant that satisfies $\norm{\lam k i}, \norm{f_t} \leq M_1$ for all $i,~t$, and $k$. Suppose that Assumption 1 holds.

\begin{lemma} \label{lem:consist_1} $d(\h\th,\th_0)=o_p(1)$ as $N,~T \rightarrow \infty$.
\end{lemma}
\begin{proof}
We first show $d^2(\h\th, \th_0) \leq \underset{\th \in \Theta^r}{\sup} |\mathbb{W}_{NT}(\th)|$ and then prove $\underset{\th \in \Theta^r}{\sup} |\mathbb{W}_{NT}(\th)|=o_p(1)$. \\
For any $\lam k i \in \mathcal{A}$ and $f_t \in \mathcal{F}$,
\[
\begin{aligned}
     \rho_{\tau_k}\big(X_{it}-\lam k i\p f_t\big)-\rho_{\tau_k}\big(X_{it}-\lam k i^{0\p} f_t^0\big) &=(\lam k i \p f_t - \lam k i^{0\p}f_t^0)\cdot\{I(X_{it}-\lam k i^{0\p}f_t^0<0)-\tau_k\}\\
     &~~~~~+ \int_0^{\lam k i \p f_t - \lam k i^{0\p}f_t^0} I(X_{it}-\lam k i^{0\p}f_t^0 \leq s)-I(X_{it}-\lam k i^{0\p}f_t^0 \leq 0) ds. 
\end{aligned}
\]
By taking the expectation, since $P(X_{it}-\lam k i ^{0\p}f_t^0 \leq 0)=\tau_k$, it holds that
\[
\begin{aligned}
    \E\big[\,\rho_{\tau_k}(X_{it}-\lam k i\p f_t)-\rho_{\tau_k}(X_{it}-\lam k i^{0\p} f_t^0)\,\big] &=\int_{0}^{\lam k i\p f_t-\lam k i^{0\p} f_t^0} P(\epsilon_{k,it} \leq s)-P(\epsilon_{k,it} \leq 0) \, ds\\
    & = \int_{0}^{\lam k i\p f_t-\lam k i^{0\p} f_t^0} s \, g_{k,it}(e^*) \, ds,  
\end{aligned}
\]
where $g_{k,it}$ is a density function of $\epsilon_{k,it}$ given $f_t^0$ and $e^*$ is some value between 0 and $s$. Using the fact that $|\lam k i^{\p} f_t|$ and $|\lam k i^{0\p} f_t^0|$ are bounded, and $s$ is between 0 and $\lam k i\p f_t-\lam k i^{0\p}f_t^0$, then by Assumption 1(c), $g_{k,it}(e^*) \geq c_1$ for all $e^*$ and  for some constant $c_1$. Therefore,
\begin{equation} \label{eq:exp-d}
     \E[\,w_k\rho_{\tau_k} (X_{it}-\lam k i\p f_t)-w_k\rho_{\tau_k}(X_{it}-\lam k i^{0\p} f_t^0)\,] \geq \frac{1}{2}w_k \cdot c_1 \cdot(\lam k i\p f_t-\lam k i^{0\p} f_t^0)^2.
\end{equation}
Then, by letting $c_2=\frac{1}{2}\cdot {\min}\{w_1,\ldots,w_K\} \cdot c_1 >0$ and summing up (\ref{eq:exp-d}) for all $i$ and $t$, we have
\[
\E[\, \mathbb{M}^*_{NT}(\th)\,] \geq c_2 \cdot \, d^2(\th,\th_0), \quad \forall \th \in \Theta^r.
\]
By definition of $\h \th$,\, $\M^*_{NT}(\h \th)=\M_{NT}(\h \th)-\M_{NT}(\th_0)\leq 0$. Thus, 
\[
0 \leq d^2(\h\th,\th_0) \lesssim \E[ \, \M_{NT}^s*(\h\th)\,]\leq -\W_{NT}(\h\th) \leq   \underset{\th \in \Theta^r}{sup} |\W_{NT}(\th)|.
\]
Now, we show that $\underset{\th \in \Theta^r}{\sup} |\W_{NT}(\th)|=o_p(1)$. Denote $B_r(M_1)$ as the ball of radius $M_1$ with respect to $\norm{\cdot}_F$ in $\R^{r}$. Let
$\{v_{(1)},\ldots, v_{(J)}\}$ be the maximal set of $B_r{(M_1)}$ s.t. $\norm{v_{(i)}-v_{(j)}}>\frac{\epsilon}{M_1}$ for a given $\epsilon>0$ ($i \neq j$). The packing number $J$ is $J=c_3(\frac{\epsilon}{M_1})^r$ for some constant $c_3>0$. For any $\th \in \Theta^r$,
let $\theta^*(\theta)=(\lam11^*,\ldots,\lam1N^*,\ldots,\lam K1^*,\ldots,\lam K N^*, f_1^*,\ldots, f_T^* ) \in \R^{(KN+T)r}$ for some
\begin{equation} \label{eq:theta star}
\lam k i^* \in \Big\{v_{(j)}:\norm{v_{j}-\lam k i} \leq \frac{\epsilon}{M_1}\Big\} \text{ and } f_t^* \in \Big\{v_{(j)}:\norm{v_{j}-f_t}\leq \frac{\epsilon}{M_1} \Big\}.
\end{equation}
Let $A_{it}=\sum_{k=1}^K w_k \rho_{\tau_k}(X_{it}-\lam k i\p f_t)-w_k \rho_{\tau_k}(X_{it}-\lam k i^{*\p} f^*_t)$. Using the inequality, 
$\rho_{\tau}(x)-\rho_{\tau}(y) \leq 2|x-y|, \forall \tau \in [0,1]$,
\[
\begin{aligned}
    |A_{it}| \leq 2  \sum\limits_{k=1}^K w_k \big|\lam ki\p f_t-\lam k i^{*}f_t^*\big| & \leq 2 \sum\limits_{k=1}^K w_k \big\{ \, \norm{\lam k i \p}\cdot \norm{f_t-f_t^*}+\norm{f_t^*}\cdot \norm{\lam k i-\lam k i^*} \, \big\}\\
    & \leq 2 \sum\limits_{k=1}^K w_k \big\{M_1 \cdot \frac{\epsilon}{M_1}+M_1 \cdot \frac{\epsilon}{M_1}\big\} = 4 \epsilon.
\end{aligned}
\]
Then, it holds that
\[
\begin{aligned}
    \big|\W_{NT}(\theta)-\W_{NT}(\theta^*)\big|&=\big|\M_{NT}(\th)-\M_{NT}(\th^*)-\E[\,\M_{NT}(\th)-\M_{NT}(\theta^*)\,]\big|\\
    &=\Big| \frac{1}{NT}\sumit A_{it}-\E[\,\frac{1}{NT}\sumit A_{it}\,]\Big|\leq 4\epsilon+4\epsilon=8\epsilon.
\end{aligned}
\]
Note that 
\[
\W_{NT}(\th^*)=\frac{1}{NT}\sumit \big\{ w_{it}(\th^*)-\E[\,w_{it}(\th^*)\,]\big\}, 
\]
where 
\[
w_{it}(\th^*) = \sum\limits_{k=1}^K w_k \rho_{\tau_k}(X_{it}-\lam k i^*f_t^*)-w_k\rho_{\tau_k}(X_{it}-\lam k i^{0}f_t^{0}).
\]
The bound of $w_{it}(\th^*)$ can be obtained as
\[
|w_{it}(\th^*)| \leq 2\sum\limits_{k=1}^K w_k\cdot |\lam k i^*f_t^*-\lam k i^{0}f_t^{0}|.
\]
Define $B_{it}(\th^*) \defeq  \sum\limits_{k=1}^K w_k\cdot |\lam k i^*f_t^*-\lam k i^{0}f_t^{0}|$. Then, by Cauchy-Schwarz inequality, 
\[
B_{it}^2 \leq \sum_{k=1}^K w_k^2 \cdot \sum_{k=1}^K (\lam k i^*f_t^*-\lam k i^{0}f_t^{0})^2,
\]
and therefore, 
\begin{equation} \label{eq:B_it}
    \frac{1}{NT}\sumit B_{it}^2 \leq \sum_{k=1}^K w_k^2 \cdot d^2(\th^*, \th_0).
\end{equation}
By Hoeffding's inequality and (\ref{eq:B_it}), 
\[
P\big(\sqrt{NT}\cdot |\W_{NT}(\th^*)|>c\big) \leq 2\cdot \exp\Big(-\frac{2NT\cdot c^2}{\sum_{i=1}^N\sum_{t=1}^T (4B_{it})^2}\Big) \leq 2\cdot \exp\Big(-\frac{c^2}{8 \sum_{k=1}^K w_k^2 \cdot d^2(\th^*,\th_0)}\Big).
\]
Then, by Lemma 2.2.1 of \cite{vaart1996weak}, we have
\[
\norm{\W_{NT}(\th^*)}_{\psi_2} \leq \left(\frac{3\cdot8 \sum_{k=1}^K w_k^2 \cdot d^2(\th^*,\th_0)}{NT} \right)^{\frac{1}{2}} \lesssim \frac{d(\theta^*,\th_0)}{\sqrt{NT}}.
\]
Using the fact that $\E[\,Y\,] \leq \norm{Y}_{\psi_2}$ and $\th^*$ can take $J^{(KN+T)}$ values, and by applying Lemma 2.2.2 from \cite{vaart1996weak}, it holds that
\[
\begin{aligned}
\E[\, \underset{\th \in \Theta^r}{\sup}|\W_{NT}(\th^*)|\,] &\leq \|{\underset{\th \in \Theta^r}{\sup}|\W_{NT}(\th^*)|}\|_{\psi_2} \\
&= \|{\underset{\th^*}{\max}\, |\W_{NT}(\th^*)|}\|_{\psi_2} 
\lesssim \psi_2^{-1}(J^{KN+T})\cdot \underset{\th^*}{\max} \|\W_{NT}(\th^*)\|_{\psi_2}.
\end{aligned}
\]
First, $\psi_2^{-1}(J^{KN+T})=\sqrt{\log(J^{KN+T})+1} \lesssim \sqrt{(KN+T)r \log\left( \frac{M_1}{\epsilon} \right)}$ holds for sufficiently small $\epsilon$. Also, $\underset{\th^*}{\max} \|\W_{NT}(\th^*)\|_{\psi_2} \leq \underset{\th^*}{\max}\frac{d(\th^*,\th_0)}{\sqrt{NT}}\lesssim \frac{1}{\sqrt{NT}}$ holds since $d(\th^*,\th_0) \leq d(\th^*,\th)+d(\th,\th_0) \lesssim M_1^2$ for some $\th$ that corresponds to $\th^*$ in the sense of (\ref{eq:theta star}). Therefore,
\[
\E[\, \underset{\th \in \Theta^r}{\sup}|\W_{NT}(\th^*)|\,] \lesssim \frac{1}{L_{NT}}\cdot \sqrt{\log\big(\frac{M_1}{\epsilon}\big)}.
\]
Finally, given $\delta>0$,
\[
\begin{aligned}
P\big(\underset{\th \in \Theta^r}{\sup}|\W_{NT}(\th)|>\delta\big) &\leq P\Big(\underset{\th \in \Theta^r}{\sup}|\W_{NT}(\th^*)|>\frac{\delta}{2}\Big)+P\Big(\underset{\th \in \Theta^r}{\sup}|\W_{NT}(\th)-\W_{NT}(\th^*)|>\frac{\delta}{2}\Big) \\
& \leq \frac{2}{\delta} \E\big[\, \underset{\th \in \Theta^r}{\sup}|\W_{NT}(\th)| \,\big]+P\Big(8\epsilon > \frac{\delta}{2}\Big).
\end{aligned}
\]
By choosing sufficiently small $\epsilon$, the second term becomes 0, and the first term goes to zero as $N,~T \rightarrow \infty$. Therefore, 
$\underset{\th \in \Theta^r}{\sup}|\W_{NT}(\th)| = o_p(1)$, and thus, $d(\h\th ,\th_0)=o_p(1)$. 
\end{proof} 
\vspace{5mm}
 
\noindent For $\delta>0$, define $\Theta^r(\delta) \defeq \{\th \in \Theta^r : d(\th,\th_0)\leq \delta \}$. 

\begin{lemma} \label{lem:consist_2} Let $S=\mbox{sgn}(F\p F^0/T)$. For sufficiently small $\delta>0$, for any $\th \in \Theta^r(\delta) $, it holds that
$$
 \frac{\norm{F-F^0S}}{\sqrt{T}} \, + \, \frac{\norm{\Lambda_1-\Lambda_1^0S}}{\sqrt{N}} \, + \, \frac{\norm{\Lambda_2-\Lambda_2^0S}}{\sqrt{N}} \, +  \cdots \, + \frac{\norm{\Lambda_K-\Lambda_K^0S}}{\sqrt{N}} \lesssim \delta.
$$
\end{lemma}
\begin{proof}
Let $U \in \R^{r \times r}$ be a diagonal matrix with diagonal elements of $1$ or $-1$. For $\th \in \Theta^r(\delta)$,
\[
\begin{aligned}
    \frac{\norm{\Lambda_k-\Lambda_k^0 U}}{\sqrt{N}} &= \frac{\norm{(\Lambda_k-\Lambda_k^0 U)F \p}}{\sqrt{NT}}=\frac{\|{\Lambda_k F\p-\Lambda_k^0 F^{0\p} +\Lambda_k^0 F^{0\p} - \Lambda_k^0 U F\p}\|}{\sqrt{NT}} \\
    & \leq \frac{\norm{\Lambda_k F\p-\Lambda_k^0 F^{0\p}}}{\sqrt{NT}} + \frac{\norm{\Lambda_k^0}}{\sqrt{N}}\cdot \frac{\norm{F^0 U-F}}{\sqrt{T}} \leq d(\th,\th_0)+M_1 \cdot \frac{\norm{F^0U-F}}{\sqrt{T}}
\end{aligned}
\]
holds for all $k=1,\ldots,K$. Then, for any $\theta \in \Theta^r(\delta)$,
$$
 \frac{\norm{F-F^0S}}{\sqrt{T}} \, + \, \frac{\norm{\Lambda_1-\Lambda_1^0S}}{\sqrt{N}} \, + \, \frac{\norm{\Lambda_2-\Lambda_2^0S}}{\sqrt{N}} \, +  \cdots \, + \frac{\norm{\Lambda_K-\Lambda_K^0S}}{\sqrt{N}} \leq K\cdot \delta+(K M_1+1)\cdot \frac{\norm{F^0U-F}}{\sqrt{T}}.
$$
Thus, it suffices to show that $\frac{\norm{F-F^0 U}}{\sqrt{T}} \lesssim d(\th,\th_0)$. 
The rest of the proof is similar to that of Lemma 2 in \cite{chen2021quantile} and is omitted. 
\end{proof}
\vspace{5mm}

\begin{lemma} \label{lem:consist_3} For sufficiently small $\delta$, $\E\big[\underset{\th \in \Theta^r(\delta)}{\sup}|\W_{NT}(\th)|\,\big] \lesssim \frac{\delta}{L_{NT}}$.
\end{lemma}
\begin{proof}
We show that $\E\big[\underset{\th \in \Theta^r(\delta)}{\sup}|\W_{NT}(\th)|\,\big] \lesssim \int_{0}^{\delta} \sqrt{\log{D(\epsilon,d,\Theta^r(\delta))}} d\epsilon$, and then prove \\$\int_{0}^{\delta} \sqrt{\log{D(\epsilon,d,\Theta^r(\delta))}}=O(\sqrt{KN+T}\cdot \delta)$.\\
Since $\Theta^r(\delta)$ is compact and $\W_{NT}(\th)$ is continuous in $\th$, $\{\W_{NT}(\th):\th \in \Theta^r(\delta)\}$ is a separable stochastic process. Also, 
by the same process in Lemma \ref{lem:consist_1}, it can be shown that
\[
\norm{\sqrt{NT}|\W_{NT}(\th_a)-\W_{NT}(\th_b)|}_{\psi_2} \leq M_2 d(\th_a,\th_b).
\]
Then, by Theorem 2.2.4 of \cite{vaart1996weak}, 
\[
\begin{aligned}
\Big\|\underset{\th \in \Theta^r(\delta)}{\sup}  \sqrt{NT}|\W_{NT}(\th)-\W_{NT}(\th_0)|\Big\|_{\psi_2} &\leq \Big\|\underset{\substack{d(\th_a,\th_b)\leq \delta \\ \th_a,\th_b \in \Theta^r(\delta)}}{\sup}\sqrt{NT}|\W_{NT}(\th_a)-\W_{NT}(\th_b)|\Big\|_{\psi_2} \\
& \leq M_2 \cdot \left[\int_0^{\delta}\psi_2^{-1}(D(\epsilon,d,\Theta^r(\delta))d\epsilon + \delta \psi_2^{-1}(D^2(\delta,d,\Theta^r(\delta))\right]\\
& \leq M_2 \cdot \left[\int_0^{\delta}\psi_2^{-1}(D(\epsilon,d,\Theta^r(\delta)) + \psi_2^{-1}(D^2(\epsilon,d,\Theta^r(\delta))  d\epsilon \right]\\
& \leq 3M_2 \cdot \int_0^{\delta}\psi_2^{-1}(D(\epsilon,d,\Theta^r(\delta))). 
\end{aligned}
\]
Therefore, using $\E[Y] \leq \norm{Y}_{\psi_2}$, we get $\E[\underset{\th \in \Theta^r(\delta)}{\sup}|\W_{NT}(\th)|\,] \lesssim \int_0^{\delta}\log(D(\epsilon,d,\Theta^r(\delta))) d\epsilon$.
Next, define 
\[
S \defeq \{U \in \R^{r \times r}: U=\text{diag}(u_1,\ldots, u_r), \, u_i \in \{-1,1\}\}. 
\]
By Lemma \ref{lem:consist_2}, for some constant $M_3>0$, $\Theta^r(\delta) \subset \underset{U \in S}{\cup}\Theta^r(\delta;U)$, where
\[
\Theta^r(\delta;U) \defeq \Big\{\th \in \Theta^r:   \frac{\norm{F-F^0 U}}{\sqrt{T}}\, + \, \sumk \frac{\norm{\Lambda_k-\Lambda_k^0 U }}{\sqrt{N}}\leq M_3 \delta \Big\}.
\]
Since $|S|=2^r$, it suffices to show that $\int_{0}^{\delta} \sqrt{\log{D(\epsilon,d,\Theta^r(\delta;U))}}=O(\sqrt{KN+T}\cdot \delta)$ for $U \in S$. Without loss of generality, let $U=\mathbb{I}_r$. For $\th_a,~ \th_b \in \Theta^r$, 
\[
\begin{aligned}
d(\th_a,\th_b) &= \frac{1}{\sqrt{NT}}\sqrt{\sumk \norm{\Lambda_k^a F^{a\p}-\Lambda_k^b F^{b\p}}^2} \\
& \leq \sumk \frac{1}{\sqrt{NT}} \norm{\Lambda_k^a F^{a\p}-\Lambda_k^b F^{b\p}} \\
& \leq \sumk \frac{1}{\sqrt{NT}} \norm{\Lambda_k^a F^{a\p}-\Lambda_k^a F^{b\p}+\Lambda_k^a F^{b\p}-\Lambda_k^b F^{b\p}}\\ 
& \leq \sumk \Big\{ \frac{\norm{\Lambda_k^a}}{\sqrt{N}}\cdot \frac{\norm{F^a-F^b}}{\sqrt{T}} +\frac{\norm{\Lambda_k^a-\Lambda_k^b}}{\sqrt{N}}\Big\} \\
& \leq \max(K M_1,1) \cdot\left[ \, \frac{\norm{F^a-F^b}}{\sqrt{T}} \, + \, \sumk \frac{\norm{\Lambda_k^a-\Lambda_k^b}}{\sqrt{N}} \, \right] \\
& \leq \max(K M_1,1) \cdot \sqrt{K+1}\cdot \sqrt{\frac{\norm{F^a-F^b}^2}{T}+\sumk \frac{\norm{\Lambda_k^a-\Lambda_k^b}^2}{N}}.
\end{aligned}
\]
Define $d^*(\th_a,\th_b)=M_4 \sqrt{\frac{\norm{F^a-F^b}^2}{T}+\sumk \frac{\norm{\Lambda_k^a-\Lambda_k^b}^2}{N}}$ for $M_4=\max(K M_1,1) \cdot \sqrt{K+1}$. Then, $d^*$ is a metric in $\Theta^r$ and $d(\th_a,\th_b)\leq d^*(\th_a,\th_b)$. Also, since 
\[
\sqrt{\frac{\norm{F^a-F^b}^2}{T}+\sumk \frac{\norm{\Lambda_k^a-\Lambda_k^b}^2}{N}} \leq \frac{\norm{F^a-F^b}}{\sqrt{T}} \, + \, \sumk \frac{\norm{\Lambda_k^a-\Lambda_k^b}}{\sqrt{N}},
\] 
if $\th \in \Theta^r(\delta;\I_r)$, then $d^*(\th,\th_0) \leq M_3 M_4 \cdot \delta$. By defining $\Theta^{r*}(\delta)\defeq \{\th: d^*(\th,\th_0) \leq M_3 M_4 \delta\}$, it holds that $\Theta^r(\delta;\I_r) \subset \Theta^{r*}(\delta)$. The rest of the proof is similar to that of Lemma 3 in \cite{chen2021quantile} and is omitted. 
\end{proof}
\vspace{2mm}

\noindent \textbf{Proof of Theorem 1:}
It can be proved in the same way as Theorem 1 in \cite{chen2021quantile} using Lemmas \ref{lem:consist_1}, \ref{lem:consist_2}, and \ref{lem:consist_3}. 

\vspace{5mm}

\section{Proof of Theorem \ref{thm:normality}} 
Now, we prove Theorem \ref{thm:normality}. Suppose that Assumptions 1 and 2 hold. To simplify the notation, we will write $\varrho_k(\cdot)$ instead of $\varrho_{\tau_k}(\cdot)$ in the proof.
We define 
\[ \mathbb{S}_{N T}^{*}(\theta)=\frac{1}{N T} \sum_{k=1}^K \sum_{i=1}^{N} \sum_{t=1}^{T} w_k \left[\varrho_k \left(X_{i t}-\lambda_{k,i}^{\prime} f_{t}\right)-\varrho_k\left(X_{i t}-\lambda_{k,i}^{0\p} f_{t}^0\right)\right].\]
Next, we denote $\varrho_k^{(j)}(\epsilon) = (\frac{\partial}{\partial \epsilon})^{j} \varrho_k(\epsilon)$ for $j=1,2,3$. For fixed $\lambda_{k,i} \in \mathcal{A}$ and $ f_{t} \in \mathcal{F}$,  we define the following notations:
\begin{itemize}
    \item $\bar{\varrho}_k^{(j)}\left(X_{i t}-\lambda_{i}^{\prime} f_{t}\right)=\mathbb{E}\left[\varrho_k^{(j)}\left(X_{i t}-\lambda_{i}^{\prime} f_{t}\right)\right] \quad$
    \item $\tilde{\varrho}_k^{(j)}\left(X_{i t}-\lambda_{i}^{\prime} f_{t}\right)=\varrho_k^{(j)}\left(X_{i t}-\lambda_{i}^{\prime} f_{t}\right)-\bar{\varrho}_k^{(j)}\left(X_{i t}-\lambda_{i}^{\prime} f_{t}\right)$ 
\end{itemize}
If one of the functions $\varrho_k, \bar{\varrho}_k,$ and $\tilde{\varrho}_k$ is computed at true parameters, $\lambda_{k,i}^0$ and $f_t^0$, we simply add a subscript, as $\varrho_{k,it}, \bar{\varrho}_{k,it},$ and $ \tilde{\varrho}_{k,it}$ respectively, instead of writing the full arguments. It is also the same for the derivatives. For example, $\varrho_{k,it}=\varrho_k\left(X_{i t}-\lambda_{k, i}^{0\p} f_{t}^0\right)$ and $ \tilde{\varrho}_{k,i t}^{(j)}=\tilde{\varrho}_k^{(j)}\left(X_{i t}-\lambda_{k, i}^{0\p} f_{t}^0\right)$. Also, we let $f(i,t) = \bar{O}(g)$ denote that $f(i,t)/g$ is uniformly bounded by across all $i$ and $t$. The stochastic order $\bar{O}_p(\cdot)$ is defined in a similar manner.
In this proof, we assume that $\tilde{S}=\mbox{sgn}(\tilde{F}\p F^0/T)=\mathbb{I}_r$ for simplicity. 

Lemma \ref{lem:normality_1} is based on the properties of kernel density estimators, and the proof is omitted for brevity.
\begin{lemma} \label{lem:normality_1}
    \begin{itemize}
    \item[(i)] For any $k$, $h^{j-1} \cdot \sup_{\epsilon}\left|\varrho_k^{(j)}(\epsilon)\right|$ is bounded for $j=1,2,3$.
    \item[(ii)] $\bar{\varrho}_{k,it}^{(1)}=\bar{O}\left(h^{m}\right), ~\bar{\varrho}^{(2)}_k\left(X_{i t}-\lambda_{k,i}^{\prime} f_{t}\right)=h_{i t}\left(\lambda_{k,i}^{\prime} f_{t}\right)+\bar{O}\left(h^{m}\right)$, and $\bar{\varrho}_k^{(3)}\left(X_{i t}-\lambda_{i}^{\prime} f_{t}\right)=-h_{i t}^{(1)}\left(\lambda_{k,i}^{\prime} f_{t}\right)+\bar{O}\left(h^{m}\right)$.
    \item[(iii)] $\mathbb{E}\left(\varrho_{k,i t}^{(1)}\right)^{2}=\tau_k(1-\tau_k)+\bar{O}(h)$, and $\mathbb{E}\left[\left(\varrho_{k,i t}^{(2)}\right)^{2}\right]=\bar{O}(h^{-1})$.
    \item[(iv)] For $k \neq k\p$, $\mathbb{E}\left(\varrho_{k,i t}^{(1)} \varrho_{k\p,i t}^{(1)}\right)=\min(\tau_k, \tau_{k\p})(1-\max(\tau_k, \tau_{k\p}))+\bar{O}(h)$.
\end{itemize}
\end{lemma}

\begin{lemma} \label{lem:normality_2} $d\left(\tilde{\theta}, \theta_{0}\right)=O_{P}\left(1 / L_{N T}\right)$ as $N, T \rightarrow \infty.$
\end{lemma}
\begin{proof}
    By proof of Theorem \ref{thm:consistency}, $d^{2}\left(\tilde{\theta}, \theta_{0}\right) \lesssim \bar{\mathbb{M}}_{N T}^{*}(\tilde{\theta})$. Also by definition of $\tilde{\th}$, we have $\mathbb{S}_{NT}(\tilde{\th}) \leq \mathbb{S}_{NT}(\th_0)$. Then,
    \begin{align*}
        d^{2}\left(\tilde{\theta}, \theta_{0}\right) &\lesssim \mathbb{E}[\,\mathbb{M}_{N T}^{*}(\tilde{\theta})\,] = \mathbb{M}_{N T}^{*}(\tilde{\th}) - \mathbb{W}_{NT}(\tilde{\th}) \\
        & \leq \mathbb{M}_{N T}(\tilde{\theta})-\mathbb{S}_{N T}(\tilde{\theta})+\mathbb{S}_{N T}\left(\theta_{0}\right)-\mathbb{M}_{N T}\left(\theta_{0}\right)-\mathbb{W}_{N T}(\tilde{\theta}) \\
        & \leq 2 \cdot \sup _{\theta \in \Theta^{r}}\left|\mathbb{M}_{N T}(\theta)-\mathbb{S}_{N T}(\theta)\right| + \sup _{\theta \in \Theta^{r}}\left|\mathbb{W}_{N T}(\theta)\right|.
    \end{align*}
    It can be shown that $\sup _{\theta \in \Theta^{r}}\left|\mathbb{M}_{N T}(\theta)-\mathbb{S}_{N T}(\theta)\right|=O_p(h)$, and since $\sup _{\theta \in \Theta^{r}}\left|\mathbb{W}_{N T}(\theta)\right|=o_p(1)$ is proved in Lemma \ref{lem:consist_1}, we have $d\left(\tilde{\theta},\theta_{0}\right)=o_p(1)$. \\
    For $\mathbb{U}_{NT}(\th)=\mathbb{M}_{N T}(\theta)-\mathbb{E}[\,\mathbb{M}_{N T}(\theta)\,]-\mathbb{M}_{N T}(\theta_0)+\mathbb{E}[\,\mathbb{M}_{N T}(\theta_0)\,]$, by a similar way with proof of Lemma \ref{lem:consist_3}, we can prove
    $\mathbb{E}\left[\sup _{\theta \in \Theta^{r}(\delta)}\left|\mathbb{U}_{N T}(\theta)\right|\right] \lesssim \frac{\delta}{L_{N T}}$. Also, using the same process in proof of Theorem \ref{thm:consistency}, we can show that 
    \[
    P\left[L_{N T} \cdot d\left(\tilde{\theta}, \theta_{0}\right)>2^{V}\right] \leq \sum_{j>V, 2^{j-1} \leq \eta L_{N T}} P\left[\inf _{\theta \in S_{j}} \mathbb{S}_{N T}^{*}(\theta) \leq 0\right]+P\left[d\left(\tilde{\theta}, \theta_{0}\right) \geq \eta\right], 
    \]
    where $S_{j}=\left\{\theta \in \Theta^{r}: 2^{j-1}<L_{N T} \cdot d\left(\theta, \theta_{0}\right) \leq 2^{j}\right\}$. First, since $d\left(\tilde{\theta},\theta_{0}\right)=o_p(1)$, $P\left[d\left(\tilde{\theta}, \theta_{0}\right) \geq \eta\right] \rightarrow 0$ as $N, T \rightarrow \infty$ for any $\eta>0$. Then, it suffices to show that \[\sum_{j>V, 2^{j-1} \leq \eta L_{N T}} P\left[\inf _{\theta \in S_{j}} \mathbb{S}_{N T}^{*}(\theta) \leq 0\right] \lesssim \sum\limits_{j>V}2^{-j}.\]
    By expanding $\mathbb{S}_{NT}^*(\th)$ around $\th_0$ and taking expectations, 
    {\small
    \[
    \mathbb{E}[\,\mathbb{S}_{N T}^{*}(\theta)\,]=\frac{1}{N T} \sumk \sum_{i=1}^{N} \sum_{t=1}^{T} w_k \bar{\varrho}_{k,it}^{(1)} \cdot\left(\lambda_{k,i}\p f_{t}-\lambda_{k,i}^{0\p} f_{t}^0\right)+\frac{1}{N T} \sumk \sum_{i=1}^{N} \sum_{t=1}^{T} 0.5\cdot w_k \cdot \bar{\varrho}_k^{(2)}\left(c_{k,i t}^{*}\right) \cdot\left(\lambda_{k,i}^{\prime} f_{t}-\lambda_{k,i}^{0\p} f_{t}^0\right)^{2}, 
    \] 
    }
    where $c_{k,i t}^{*}$ is some value between $\lambda_{k,i}\p f_t$ and $\lambda_{k,i}^{0\p} f_t^0$. By Lemma \ref{lem:normality_1}, for fixed $\th \in S_j$,
    \[
    -\E[\,\mathbb{S}_{N T}^{*}(\theta)\,] \leq-0.5 \cdot \min_k{w_k} \cdot h_\ell \cdot d^2\left(\theta, \theta_{0}\right)+O\left(h^{m}\right) \leq -\min_k{w_k} \cdot h_\ell \cdot \frac{2^{2 j-3}}{L_{N T}^{2}}+O\left(h^{m}\right).
    \]
    The rest of the proof is similar to that of Lemma S.3 in \cite{chen2021quantile} and is omitted. 
\end{proof}

For Lemmas \ref{lem:normality_3} and \ref{lem:normality_4}, we define the following notations:
\begin{itemize}
    \item $\mathbb{S}_{k,i,T}^{*}(\lambda, F)=\frac{1}{T} \sum_{t=1}^{T}\left[\varrho_k\left(X_{i t}-\lambda^{\prime} f_{t}\right)-\varrho_k\left(X_{i t}-\lambda_{k,i}^{0\p} f_t^0\right)\right]$
    \item $\mathbb{M}_{k,i,T}^{*}(\lambda, F)=\frac{1}{T} \sum_{t=1}^{T}\left[\rho_k\left(X_{i t}-\lambda^{\prime} f_{t}\right)-\rho_k\left(X_{i t}-\lambda_{k,i}^{0\p} f_t^0\right)\right]$
    \item $\mathbb{S}_{K,N,t}^{*}(\Lambda_1,\cdots,\Lambda_K, f)=\frac{1}{N} \sumk \sum\limits_{i=1}^{N} w_k\left[\varrho_k\left(X_{i t}-\lambda_{k,i}^{\prime} f\right)-\varrho_k\left(X_{i t}-\lambda_{k,i}^{0\p} f_t^0\right)\right]$
    \item $\mathbb{M}_{K,N,t}^{*}(\Lambda_1,\cdots,\Lambda_K, f)=\frac{1}{N} \sumk \sum\limits_{i=1}^{N} w_k\left[\rho_k\left(X_{i t}-\lambda_{k,i}^{\prime} f\right)-\rho_k\left(X_{i t}-\lambda_{k,i}^{0\p} f_t^0\right)\right]$
\end{itemize}

\begin{lemma} \label{lem:normality_3} $\|\tilde{\lambda}_{k,i}-\lambda_{k,i}^0\|=o_p(1)$ for each $k$,$i$ and $\|\tilde{f}_t-f^0_t\|=o_p(1)$ for each $t$.
\end{lemma}
\begin{proof}
    We first show $\|\tilde{\lambda}_{k,i}-\lambda_{k,i}^0\|=o_p(1)$. By definition of $\tilde{\th}$, it holds that $\tilde{\lambda}_{k,i}=\underset{\lambda \in \mathcal{A}}{\arg \min \,} \mathbb{S}_{k,i,T}^{*}(\lambda, \tilde{F})$. We show that $\sup _{\lambda \in \mathcal{A}}\left|\mathbb{S}_{k,i,T}^{*}(\lambda, \tilde{F})-\E[\mathbb{M}_{k,i,T}^{*}(\lambda,F^0)]\right|=o_{P}(1)$ and $\lambda_{k,i}^0$ is the unique minimizer of $\E[\mathbb{M}_{i, T}^{*}(\lambda,F^0)]$. Then, $\|\tilde{\lambda}_{k,i}-\lambda_{k,i}^0\|=o_p(1)$ is proved by Theorem 2.1 of \cite{newey1994large}.
    
By triangular inequality,
{\small
    \begin{align*}
 \sup _{\lambda \in \mathcal{A}}\left|\mathbb{S}_{k,i,T}^{*}(\lambda, \tilde{F})-\E[\mathbb{M}_{k,i,T}^{*}(\lambda,F^0)] \right| 
 & \leq \sup _{\lambda \in \mathcal{A}}\left|\mathbb{S}_{k,i,T}^{*}(\lambda, \tilde{F})-\mathbb{M}_{k,i,T}^{*}(\lambda, \tilde{F})\right|+\sup _{\lambda \in \mathcal{A}}\left|\mathbb{M}_{k,i,T}^{*}(\lambda, \tilde{F})-\mathbb{M}_{k,i,T}^{*}\left(\lambda, F^0\right)\right| \\
& \quad +\sup _{\lambda \in \mathcal{A}}\left|\mathbb{M}_{k,i,T}^{*}\left(\lambda, F^0\right)-\E[\mathbb{M}_{k,i,T}^{*}(\lambda,F^0)]\right|.
\end{align*}
}
It can be easily shown that 
\begin{align*}
\sup_{\lambda \in \mathcal{A}}&\left|\mathbb{S}_{k,i,T}^{*}(\lambda, \tilde{F})-\mathbb{M}_{k,i,T}^{*}(\lambda, \tilde{F})\right| \\
\quad = &
\sup_{\lambda \in \mathcal{A}}\left|\frac{1}{T} \sum\limits_{t=1}^{T} \left[\varrho_k\left(X_{i t}-{\lambda}^{\prime} \tilde{f}_t\right)-\rho_k\left(X_{i t}-{\lambda}^{\prime} \tilde{f}_t\right)-\left\{\varrho_k\left(X_{i t}-\lambda_{k,i}^{0\p} f_t^0\right)-\rho_k\left(X_{i t}-\lambda_{k,i}^{0\p} f_t^0\right) \right\}\right] \right|
\lesssim h,
\end{align*} 
and ~$\sup_{\lambda \in \mathcal{A}}\left|\mathbb{M}_{k,i,T}^{*}\left(\lambda, F^0\right)-\E[\,\mathbb{M}_{k,i,T}^{*}(\lambda,F^0)\,]\right|  =o_{p}(1)$. Also, using Lemmas \ref{lem:consist_2} and \ref{lem:normality_2},
\[\sup _{\lambda \in \mathcal{A}}\left|\mathbb{M}_{k,i,T}^{*}(\lambda, \tilde{F})-\mathbb{M}_{k,i,T}^{*}\left(\lambda, F^0\right)\right|  \lesssim \sup _{\lambda \in \mathcal{A}}\|\lambda\| \cdot \frac{1}{T} \sum_{t=1}^{T}\left\|\tilde{f}_{t}-f_{0 t}\right\| 
\lesssim \frac{\left\|\tilde{F}-F^0\right\|}{\sqrt{T}} \lesssim d(\tilde{\th},\th_0)=o_{p}(1).\]
Therefore, $\sup _{\lambda \in \mathcal{A}}\left|\mathbb{S}_{k,i,T}^{*}(\lambda, \tilde{F})-\E[\,\mathbb{M}_{k,i,T}^{*}(\lambda,F^0)\,]\right|=o_{P}(1)$.

For $\epsilon>0$, let $B_{k,i}(\epsilon)=\left\{\lambda \in \mathcal{A}:\left\|\lambda-\lambda_{k,i}^0\right\| \leq \epsilon\right\}$. Then, it can be shown that
\[
\inf_{\lambda \in B_{k,i}^{C}(\epsilon)} \E[\,\mathbb{M}_{k,i, T}^{*}\left(\lambda, F^0\right)\,]>\E[\,\mathbb{M}_{k,i, T}^{*}\left(\lambda_{k,i}^0, F^0\right)\,]=0
\]
for arbitrary $\epsilon>0$, using the proof of Proposition 3.1 in \cite{galvao2016smoothed}, which implies the uniqueness of the minimizer of $\E[\,\mathbb{M}_{i, T}^{*}\left(\lambda, F^0\right)\,]$. Therefore, $\|\tilde{\lambda}_{k,i}-\lambda_{k,i}^0\|=o_p(1)$.

Second, we show $\|\tilde{f}_t-f_t^0\|=o_p(1)$. Note that $\tilde{f}_{t}=\underset{f \in \mathcal{F}}{\arg \min \,} \mathbb{S}_{K,N,t}^{*}(\tilde{\Lam}_1,\cdots,\tilde{\Lam}_K, f)$. Similarly, 
{\small
\begin{align*}
    \sup_{f \in \mathcal{F}}&\left|\mathbb{S}_{K,N,t}^{*}(\tilde{\Lam}_1,\cdots,\tilde{\Lam}_K, f)-\mathbb{M}_{K,N,t}^{*}(\tilde{\Lam}_1,\cdots,\tilde{\Lam}_K, f)\right|\\ 
    \leq & \sumk w_k \cdot \sup_{f \in \mathcal{F}}\left| \frac{1}{N} \sum\limits_{i=1}^{N} \left[\varrho_k\left(X_{i t}-\tilde{\lambda}_{k,i}^{\prime} f\right)-\rho_k\left(X_{i t}-\tilde{\lambda}_{k,i}^{\prime} f\right)-\left\{\varrho_k\left(X_{i t}-\lambda_{k,i}^{0\p} f_t^0\right)-\rho_k\left(X_{i t}-\lambda_{k,i}^{0\p} f_t^0\right) \right\}\right] \right| \lesssim h,
\end{align*}
}
and 
$\sup _{\lambda \in \mathcal{A}}\left|\mathbb{M}_{K,N,t}^{*}\left(\Lam_1^0,\cdots,\Lam_K^0, f \right)-\E[\,\mathbb{M}_{K,N,t}^{*}(\Lam_1^0,\cdots,\Lam_K^0, f)\,]\right|  =o_{p}(1)$. Also, 
\begin{align*}
  \sup _{f \in \mathcal{F}}\left|\mathbb{M}_{K,N,t}^{*}(\tilde{\Lam}_1,\cdots, \tilde{\Lam}_K, f)-\mathbb{M}_{K,N,t}^{*}\left(\Lam_1^0,\cdots,\Lam_K^0, f\right)\right| & \lesssim \sumk w_k \cdot \sup _{f \in \mathcal{F}}\|f\| \cdot \frac{1}{N} \sum_{i=1}^{N}\left\|\tilde{\lambda}_{k,i}-\lambda_{k,i}^0\right\| \\
& \lesssim \sumk w_k \frac{\left\|\tilde{\Lam}_k-\Lam_k^0\right\|}{\sqrt{N}} \\
& \lesssim d(\tilde{\th},\th_0)=o_{p}(1).  
\end{align*}
Therefore, $\sup_{\lambda \in \mathcal{A}}\left|\mathbb{S}_{K,N,t}^{*}(\tilde{\Lam}_1,\cdots, \tilde{\Lam}_K,f)-\E[\,\mathbb{M}_{K,N,t}^{*}(\Lam^0_1,\cdots, \Lam^0_K,f)\,]\right|=o_{P}(1)$.
Also, for $B_{t}(\epsilon)=\left\{f \in \mathcal{F}:\left\|f-f_t^0\right\| \leq \epsilon\right\}$, 
\begin{align*}
    \inf_{f \in B_{t}^{C}(\epsilon)} &\E\left[\,\mathbb{M}_{K,N,t}^{*}\left(\Lam_1^0,\cdots,\Lam_K^0, f\right)\,\right] \\ & \geq \sumk w_k  \inf_{f \in B_{t}^{C}(\epsilon)} \E\left[\, \frac{1}{N} \sum\limits_{i=1}^{N} \left\{\rho_k\left(X_{i t}-\lambda_{k,i}^{0\p} f\right)-\rho_k\left(X_{i t}-\lambda_{k,i}^{0\p} f_t^0\right)\right\} \,\right] \\ &> \sumk w_k \cdot 0 = \E\left[\,\mathbb{M}_{K,N,t}^{*}\left(\Lam_1^0,\cdots,\Lam_K^0, f_t^0 \right)\,\right].
\end{align*}
and thus, $\|\tilde{f}_{t}-f_{t}^0\|=o_p(1)$ is proved.
\end{proof}

\newpage
\begin{lemma} \label{lem:normality_4} $\|\tilde{\lambda}_{k,i}-\lambda_{k,i}^0\|=O_{P}\left(\frac{1}{\sqrt{T}h}\right)$ for each $k,i$, and ~$\|\tilde{f}_t-f_t^0\|=O_{P}\left(\frac{1}{\sqrt{N}h}\right)$ for each $t$.
\end{lemma}
\begin{proof}
    We only prove $\|\tilde{f}_t-f_t^0\|=O_{P}\left(\frac{1}{\sqrt{T}h}\right)$, since the proof of $\|\tilde{\lambda}_{k,i}-\lambda_{k,i}^0\|=O_{P}\left(\frac{1}{\sqrt{T}h}\right)$ is similar to Lemma S.4 in \cite{chen2021quantile}. For fixed $k$ and fixed $\lambda_{k,i}$, $f_t$, by expanding $\varrho_k^{(1)}\left(X_{i t}-\lambda_{k,i}^{\prime} f_{t}\right) \lambda_{k,i}$ at $f_t^0$, we can write
\begin{align*}
\varrho_k^{(1)}\left(X_{i t}-\lambda_{k,i}^{\prime} f_{t}\right) \lambda_{k,i} & = \varrho_k^{(1)}\left(X_{i t}-\lambda_{k,i}^{\prime} f_{t}^0\right) \lambda_{k,i}-\varrho_k^{(2)}\left(X_{i t}-\lambda_{k, i}^{\prime} f_{t}^0\right) \lambda_{k,i} \lambda_{k,i}^{\prime} \cdot\left(f_t-f_t^0\right) \\ & \quad + 0.5 \varrho_k^{(3)}\left(X_{i t}-\lambda_{k,i}^{\prime} f_{t}^{*}\right) \lambda_{k,i}\left[\left(f_t-f_t^0\right)^{\prime} \lambda_{k,i}\right]^{2}.
\end{align*}
By expanding $\varrho_k^{(1)}\left(X_{i t}-\lambda_{k,i}^{\prime} f_{t}^0\right) \lambda_{k,i} $ and $ \varrho_k^{(2)}\left(X_{i t}-\lambda_{k, i}^{\prime} f_{t}^0\right) $ at $\lambda_{k,i}^0$ again, we get
\begin{align} \label{eq:Taylorexp}
\begin{split}
    & \varrho_k^{(1)}\left(X_{i t}-\lambda_{k,i}^{\prime} f_{t}\right) \lambda_{k,i}  \\
    & \quad = \varrho_{k,i t}^{(1)} \lambda_{k,i}^0 + \left\{ 
    \varrho_k^{(1)}\left(X_{i t}-\lambda_{k, i}^{*\p} f_t^0\right)-\varrho^{(2)}_k\left(X_{i t}-\lambda_{k,i}^{*\p} f_{t}^{0}\right) \lambda_{k,i}^{*} f_{t}^{0\p}  
    \right\} \left(\lambda_{k,i}-\lambda_{k,i}^0\right) \\
& \quad -\varrho_{k,i t}^{(2)} \lambda_{k,i} \lambda_{k,i}^{\prime}\left(f_t-f_t^0\right) +\varrho_k^{(3)}\left(X_{i t}-\lambda_{k,i}^{*^\prime} f_{t}^{0}\right) \lambda_{k,i} \lambda_{k,i} ^{\prime} \cdot\left(f_t-f_t^0\right) f_t^{0^\prime}\left(\lambda_{k,i}-\lambda_{k,i}^0\right) \\
& \quad + 0.5 \varrho_k^{(3)}\left(X_{i t}-\lambda_{k,i}^{\prime} f_{t}^{*}\right) \lambda_{k,i}\left[\left(f_t-f_t^0\right)^{\prime} \lambda_{k,i}\right]^{2},
\end{split}
\end{align}
where $f_t^*$ lies between $f_t$ and $f_t^0$, and $\lambda_{k,i}^*$ lies between $\lambda_{k,i}$ and $\lambda_{k,i}^0$. By summing the terms across $i$ and taking expectations, we get
\begin{align*}
    &\left(\frac{1}{N} \sum_{i=1}^{N} \bar{\varrho}_{k,i t}^{(2)} \lambda_{k,i} \lambda_{k,i}^{\prime}\right)\left(f_{t}-f_t^0\right) \\ & \quad = -\frac{1}{N}\sum\limits_{i=1}^N \bar{\varrho}_k^{(1)}\left(X_{i t}-\lambda_{k,i}^{\prime} f_{t}\right) \lambda_{k,i} + \frac{1}{N}\sum\limits_{i=1}^N \bar{\varrho}_{k,i t}^{(1)} \lambda_{k,i}^0 \\
    & \quad \quad +\frac{1}{N}\sum\limits_{i=1}^N \left\{ 
    \bar{\varrho}_k^{(1)}\left(X_{i t}-\lambda_{k, i}^{*\p} f_t^0\right)-\bar{\varrho}^{(2)}_k\left(X_{i t}-\lambda_{k,i}^{*\p} f_{t}^{0}\right) \lambda_{k,i}^{*} f_{t}^{0\p} \right\} \left(\lambda_{k,i}-\lambda_{k,i}^0\right) \\
    & \quad \quad + \frac{1}{N}\sum\limits_{i=1}^N \bar{\varrho}_k^{(3)}\left(X_{i t}-\lambda_{k,i}^{*^\prime} f_{t}^{0}\right) \lambda_{k,i} \lambda_{k,i} ^{\prime} \cdot\left(f_t-f_t^0\right) f_t^{0^\prime}\left(\lambda_{k,i}-\lambda_{k,i}^0\right) \\
    & \quad \quad + \frac{1}{N}\sum\limits_{i=1}^N 0.5 \bar{\varrho}_k^{(3)}\left(X_{i t}-\lambda_{k,i}^{\prime} f_{t}^{*}\right) \lambda_{k,i}\left[\left(f_t-f_t^0\right)^{\prime} \lambda_{k,i}\right]^{2}.
\end{align*}
Next, we put $\lambda_{k,i}=\tilde{\lambda}_{k,i}$ and $f_t=\tilde{f}_t$. Then, by using Lemma \ref{lem:normality_1}, Assumption 2 and that $\frac{1}{N}\sum\limits_{i=1}^N \| \lambda_{k,i}-\lambda_{k,i}^0\| \lesssim \frac{1}{\sqrt{N}} \| \Lambda_k-\Lambda_k^0 \|$, it holds that 
\begin{align} \label{lem10_eq1}
\begin{split}
    &\left(\frac{1}{N} \sum_{i=1}^{N} \bar{\varrho}_{k,i t}^{(2)}\tilde{\lambda}_{k,i} \tilde{\lambda}_{k,i}^{\prime}\right)\left(\tilde{f}_{t}-f_t^0\right) \\ & \quad = -\frac{1}{N}\sum\limits_{i=1}^N \bar{\varrho}_k^{(1)}\left(X_{i t}-\tilde{\lambda}_{k,i}^{\prime} \tilde{f}_{t}\right) \tilde{\lambda}_{k,i} + \frac{1}{N}\sum\limits_{i=1}^N \bar{\varrho}_{k,i t}^{(1)} \lambda_{k,i}^0 +O_p\left( \frac{1}{\sqrt{N}} \| \Lambda_k-\Lambda_k^0 \| \right) \\
    & \quad \quad +O_p(\| \tilde{f}_t - f_t^0 \|)\cdot O_p\left( \frac{1}{\sqrt{N}} \| \Lambda_k-\Lambda_k^0 \| \right) + O_p(\| \tilde{f}_t - f_t^0 \|^2).
\end{split}
\end{align}
Note that by Lemmas \ref{lem:normality_1} and \ref{lem:normality_2}, and Assumption 2, 
\begin{equation} \label{eq:Psi_t}
\frac{1}{N} \sumk \sum_{i=1}^{N} w_k \cdot \bar{\varrho}_{k,i t}^{(2)}\tilde{\lambda}_{k,i} \tilde{\lambda}_{k,i}^{\prime} = \frac{1}{N}\sumk \sum_{i=1}^{N} w_k \cdot \bar{\varrho}_{k,i t}^{(2)} \lambda_{k,i}^0 \lambda_{k,i}^{0\p} +o_p(1) = \Psi_t +o_p(1).\end{equation}
By multiplying weight $w_k$ and summing across $k$ for both sides of equation (\ref{lem10_eq1}), it follows from Lemmas \ref{lem:normality_1}, \ref{lem:normality_2} and \ref{lem:normality_3} that
\begin{equation} \label{lem10_eq2}
    \Psi_t(\tilde{f}_{t}-f_t^0) + o_p(\|\tilde{f}_{t}-f_t^0 \|) = -\frac{1}{N} \sumk \sum\limits_{i=1}^N w_k \cdot \bar{\varrho}_k^{(1)}\left(X_{i t}-\tilde{\lambda}_{k,i}^{\prime} \tilde{f}_{t}\right) \tilde{\lambda}_{k,i} + O(h^m) + O_p\left({L_{NT}^{-1}}\right).
\end{equation}
Since $\tilde{\th}$ minimizes $S_{NT}(\th)$, 
\[ 
 \frac{\partial}{\partial f_t} S_{NT}(\tilde{\th}) = \frac{1}{N}\sumk \sum\limits_{i=1}^N w_k \cdot \varrho^{(1)}_k(X_{it}-\tilde{\lambda}_{k,i}^{\prime} \tilde{f}_{t}) \tilde{\lambda}_{k,i} = 0.
\]
Therefore, 
\begin{align*}
-\frac{1}{N}& \sumk \sum\limits_{i=1}^N w_k \cdot \bar{\varrho}_k^{(1)}\left(X_{i t}-\tilde{\lambda}_{k,i}^{\prime} \tilde{f}_{t}\right) \tilde{\lambda}_{k,i} \\
&= \frac{1}{N} \sumk \sum\limits_{i=1}^N w_k \cdot \tilde{\varrho}_k^{(1)}\left(X_{i t}-\tilde{\lambda}_{k,i}^{\prime} \tilde{f}_{t}\right) \tilde{\lambda}_{k,i}  \\
&= \frac{1}{N} \sumk \sum_{i=1}^{N} w_k \cdot \tilde{\varrho}_{k,it}^{(1)} \lambda_{k,i}^0 
+
\frac{1}{N} \sumk \sum_{i=1}^{N} w_k \left[\tilde{\varrho}_k^{(1)}\left(X_{i t}-\tilde{\lambda}_{k,i}^{\prime} \tilde{f}_{t}\right) \tilde{\lambda}_{k,i}-\tilde{\varrho}_{k,it}^{(1)} \lambda_{k,i}^0\right] \\
& = \frac{1}{N} \sumk \sum_{i=1}^{N} w_k \cdot \tilde{\varrho}_{k,it}^{(1)} \lambda_{k,i}^0 
+
\frac{1}{N} \sumk \sum_{i=1}^{N} w_k \left[\tilde{\varrho}_k^{(1)}\left(X_{i t}-\tilde{\lambda}_{k,i}^{\prime} \tilde{f}_{t}\right) \tilde{\lambda}_{k,i} -  \tilde{\varrho}_k^{(1)}\left(X_{i t}-\lambda_{k,i}^{0\p} \tilde{f}_{t}\right) \lambda_{k,i}^0 \right]\\
& \quad +\frac{1}{N} \sumk \sum\limits_{i=1}^N w_k \left[ \tilde{\varrho}_k^{(1)}\left(X_{i t}-\lambda_{k,i}^{0\p} \tilde{f}_{t}\right) -\tilde{\varrho}_{it}^{(1)} \right] \lambda_{k,i}^0.
\end{align*}
By Lemma \ref{lem:normality_1} and Lyapunov's CLT, $\frac{1}{N} \sumk \sum\limits_{i=1}^{N} w_k \cdot \tilde{\varrho}_{k,it}^{(1)} \lambda_{k,i}^0=O_p(N^{-1/2}) $. Also, 
\begin{align*}
 &\frac{1}{N} \sumk \sum_{i=1}^{N} w_k \left[\tilde{\varrho}_k^{(1)}\left(X_{i t}-\tilde{\lambda}_{k,i}^{\prime} \tilde{f}_{t}\right) \tilde{\lambda}_{k,i} -  \tilde{\varrho}_k^{(1)}\left(X_{i t}-\lambda_{k,i}^{0\p} \tilde{f}_{t}\right) \lambda_{k,i}^0 \right] \\
  & \quad = \frac{1}{N} \sum_{i=1}^{N} w_k \tilde{\varrho}_k^{(1)}\left(X_{i t}-{\lambda}_{k,i}^{*\p} \tilde{f}_{t}\right)\left(\tilde{\lambda}_{k,i}-\lambda_{k,i}^0\right) - \frac{1}{N} \sum_{i=1}^{N} w_k \tilde{\varrho}_k^{(2)}\left(X_{i t}-{\lambda}_{k,i}^{*\p} \tilde{f}_{t}\right) \tilde{\lambda}_{k,i}^{*} \tilde{f}_{t}^{\prime}\left(\tilde{\lambda}_{k,i}-\lambda_{k,i}^0\right) \\
  & \quad = O_p(L_{NT}^{-1}) +O_p(L_{NT}^{-1}h^{-1}),
\end{align*}
where the last equality holds by Lemmas \ref{lem:normality_1} and \ref{lem:normality_2}. Finally, it can be shown that
\[
\frac{1}{N} \sumk \sum\limits_{i=1}^N w_k \left[ \tilde{\varrho}_k^{(1)}\left(X_{i t}-\lambda_{k,i}^{0\p} \tilde{f}_{t}\right) -\tilde{\varrho}_{k,it}^{(1)} \right] \lambda_{k,i}^0 = O_p( \| \tilde{f}_{t}- f_{t}^0 \| )\cdot O_p(1/\sqrt{Nh})=o_p(\| \tilde{f}_{t}- f_{t}^0 \|),
\]
by using Lemma B.2 of \cite{galvao2016smoothed}.
Therefore, from (\ref{lem10_eq2}),
\[
    \Psi_t(\tilde{f}_{t}-f_t^0) + o_p(\|\tilde{f}_{t}-f_t^0 \|) = O_p(L_{NT}^{-1}h^{-1}) + O(h^m). 
\]
Since $\sqrt{N}h^{m+1} \rightarrow 0$ and $\Psi_t \succ 0$ by Assumption 2, we get $ \| \tilde{f}_{t}-f_t^0 \|=O_p(\frac{1}{\sqrt{N}h})$.
\end{proof}

To conduct stochastic expansions of $\tilde{\lambda}_{k,i}$ and $\tilde{f}_t$, we define the following notations. For $b>0$,
{\small
\[
\mathbb{P}_{N T}(\theta)=b\left[\frac{1}{2 N} \sum\limits_{p=1}^{r} \sum\limits_{q>p}^{r}\left(\sum\limits_{i=1}^{N} \lambda_{k^*,ip} \lambda_{k^*,i q}\right)^{2}+\frac{1}{2 T} \sum\limits_{p=1}^{r} \sum\limits_{q>p}^{r}\left(\sum\limits_{t=1}^{T} f_{t p} f_{t q}\right)^{2}+\frac{1}{8 T} \sum\limits_{k=1}^{r}\left(\sum\limits_{t=1}^{T} f_{t k}^{2}-T\right)^{2}\right] \in \mathbb{R},
\]
}
where $\lambda_{k^*,ip}$ and $f_{tp}$ denotes the $p$th element of $\lambda_{k^*,i}$ and $f_t$, respectively. Next, define\\
\resizebox{0.93\linewidth}{!}{
\begin{minipage}{\linewidth}
$$
\begin{aligned}
    \mathbb{S}^*(\th)= & \underbrace{ \left[
    -\frac{1}{\sqrt{N T}}  \sum\limits_{t=1}^{T} w_1 \bar{\varrho}_1^{(1)}\left(X_{1 t}-\lambda_{1,1}^{\prime} f_{t}\right) f_{t}^{\prime}, \ldots, -\frac{1}{\sqrt{N T}}  \sum\limits_{t=1}^{T}  w_1 \bar{\varrho}_1^{(1)}\left(X_{N t}-\lambda_{1,N}^{\prime} f_{t}\right) f_{t}^{\prime} \right.}_{1 \times Nr},   \\
     & \ldots \ldots , \underbrace{-\frac{1}{\sqrt{N T}}  \sum\limits_{t=1}^{T} w_K \bar{\varrho}_K^{(1)}\left(X_{1 t}-\lambda_{K,1}^{\prime} f_{t}\right) f_{t}^{\prime}, \ldots, -\frac{1}{\sqrt{N T}}  \sum\limits_{t=1}^{T}  w_K \bar{\varrho}_K^{(1)}\left(X_{N t}-\lambda_{K,N}^{\prime} f_{t}\right) f_{t}^{\prime}, }_{1 \times Nr} \\
     &  \underbrace{\left.-\frac{1}{\sqrt{N T}}  \sumk \sum\limits_{i=1}^{N} w_k\bar{\varrho}_k^{(1)}\left(X_{i 1}-\lambda_{k,i}^{\prime} f_1\right) \lambda_{k,i}^{\prime}, \ldots, -\frac{1}{\sqrt{N T}}  \sumk \sum\limits_{t=1}^{T} w_k \bar{\varrho}_k^{(1)}\left(X_{i T}-\lambda_{k,i}^{\prime} f_{T}\right) \lambda_{k,i}^{\prime}
     \right]\p }_{1 \times Tr} \in \mathbb{R}^{(KN+T)r}.
\end{aligned}
$$
\end{minipage}
}
\\
Also, let
\[
\mathcal{S}(\theta)=\mathcal{S}^{*}(\theta)+\partial \mathbb{P}_{N T}(\theta) / \partial \theta, \quad \mathcal{H}(\theta)=\partial \mathcal{S}^{*}(\theta) / \partial \theta^{\prime}+\partial^{2} \mathbb{P}_{N T}(\theta) / \partial \theta \partial \theta^{\prime} .
\]
By expanding $S(\tilde{\th})$ at $\th_0$, we can write
\begin{equation}\label{eq:expandS}
\mathcal{S}(\tilde{\theta})=\mathcal{S}\left(\theta_{0}\right)+\mathcal{H}(\th_0) \cdot\left(\tilde{\theta}-\theta_{0}\right)+0.5 \mathcal{R}(\tilde{\theta})
\end{equation}
for $\mathcal{R}(\tilde{\theta})=\left\{\sum\limits^{(KN+T)r}_{j=1} \frac{\partial \mathcal{H}\left(\theta^{*}\right)}{\partial \theta_{j}} \cdot\left(\tilde{\theta}_{j}-\theta_{0 j}\right)\right\}\left(\tilde{\theta}-\theta_{0}\right)$, where $\th_j$ denotes the $j$th element of $\th$, and $\th^*$ lies between $\tilde{\th}$ and $\th_0$. 
Next, define
\[
\Phi_{T,k, i}=\frac{1}{T} \sum_{t=1}^{T}  \bar{\varrho}_{k,i t}^{(2)} f_{t}^0 f_t^{0\p}, \quad \Psi_{N, t}=\frac{1}{N} \sumk \sum_{i=1}^{N} w_k \bar{\varrho}_{k,i t}^{(2)} \lambda_{k,i}^0 \lambda_{k,i}^{0\p}, 
\]
and for simplicity of notation, we denote $\Phi_{T,k, i}^* = w_k \Phi_{T,k,i}$.\\
For $\mathcal{H}_d^{\Lam_k} = \frac{\sqrt{T}}{\sqrt{N}} \operatorname{diag}\left\{\Phi_{T,k, 1}^*, \ldots, \Phi_{T,k,i}^*, \ldots, \Phi_{T,k,N}^*\right\} \in \mathbb{R}^{Nr \times Nr}$, $k=1,\ldots,K$ and \\ $\mathcal{H}_{d}^{F} =\frac{\sqrt{N}}{\sqrt{T}} \operatorname{diag}\left[\Psi_{N, 1}, \ldots, \Psi_{N, t}, \ldots, \Psi_{N, T}\right]\in \mathbb{R}^{Tr \times Tr}$, define
$\mathcal{H}_d = \begin{bmatrix}
    \mathcal{H}_d^{\Lam_1} & & & 0\\
       & \ddots &  &   \\ 
       &  & \mathcal{H}_d^{\Lam_K} & \\
      0 &  &  & \mathcal{H}_d^{F} \\
\end{bmatrix} $.  \\

\begin{lemma} \label{lem:normality_5} $\mathcal{H}(\th_0)^{-1}$ exists and $\left\|\mathcal{H}(\th_0)^{-1}-\mathcal{H}_{d}^{-1}\right\|_{\max }=O(T^{-1})$.
\end{lemma}
\begin{proof}
    For simplicity, we only prove the case when $r=2$ and $K=2$. Note that it can be generalized to $r>2$ and $K>2$. Without loss of generality, we choose $k^*=1$. Then,
    \[
    \mathbb{P}_{N T}(\theta)=  b\left[\frac{1}{2 N}\left(\sum_{i=1}^{N} \lambda_{1,i 1} \lambda_{1,i 2}\right)^{2}+\frac{1}{2 T}\left(\sum_{t=1}^{T} f_{t 1} f_{t 2}\right)^{2} +\frac{1}{8 T}\left(\sum_{t=1}^{T} f_{t 1}^{2}-T\right)^{2}+\frac{1}{8 T}\left(\sum_{t=1}^{T} f_{t 2}^{2}-T\right)^{2}\right].
    \]
    The second derivative of $\mathbb{P}_{NT}(\th)$ on $\th=\th_0$ can be written as
    \[
    \frac{\partial^{2} \mathbb{P}_{N T}\left(\theta_{0}\right)}{\partial \theta \partial \theta^{\prime}}=b\cdot\sum\limits_{j=1}^{4} \gamma_{j} \gamma_{j}^{\prime},
    \]
    where 
    \begin{align*}
   \gamma_{1} & = \left[\mathbf{0}_{1 \times 2 N}, \mathbf{0}_{1 \times 2 N},\left(f_{11}^0, 0\right),\ldots,\left(f_{T1}^0, 0\right)\right]^{\prime} / \sqrt{T}, \\
   \gamma_{2} & =\left[\mathbf{0}_{1 \times 2 N}, \mathbf{0}_{1 \times 2 N},\left(0, f_{12}^0\right), \ldots,\left(0, f_{T2}^0\right)\right]^{\prime} / \sqrt{T}, \\
\gamma_{3} & =\left[\mathbf{0}_{1 \times 2 N}, \mathbf{0}_{1 \times 2 N},\left(f_{12}^0, f_{11}^0\right), \ldots,\left(f_{T2}^0, f_{T1}^0\right)\right]^{\prime} / \sqrt{T}, \\
\gamma_{4} & =\left[\left(\lambda_{1,12}^0, \lambda_{1,11}^0\right), \ldots,\left(\lambda_{1,N2}^0, \lambda_{1,N1}^0\right), 
\mathbf{0}_{1 \times 2 N},
\mathbf{0}_{1 \times 2 T}\right]^{\prime} / \sqrt{N},
   \end{align*}
are vectors of length $2 \cdot (2N+T)$. Next, define vectors \\
\resizebox{1.05\linewidth}{!}{
\begin{minipage}{\linewidth}
\begin{align*}
\delta_{1} & =\left[ 
\frac{1}{(NT)^{1/4}}\left( \lam{1}{11}^0, 0 \right), \ldots, \frac{1}{(NT)^{1/4}}\left( \lam{1}{N1}^0, 0 \right), 
\frac{1}{(NT)^{1/4}}\left( \lam{2}{11}^0, 0 \right), \ldots, \frac{1}{(NT)^{1/4}}\left( \lam{2}{N1}^0, 0 \right),
\frac{1}{(NT)^{1/4}}\left( -f_{11}^0, 0 \right), \ldots, \frac{1}{(NT)^{1/4}}\left( -f_{T1}^0, 0 \right)
\right]^{\prime} \\
\delta_{2} & =\left[ 
\frac{1}{(NT)^{1/4}}\left( \lam{1}{12}^0, 0 \right), \ldots, \frac{1}{(NT)^{1/4}}\left( \lam{1}{N2}^0, 0 \right), 
\frac{1}{(NT)^{1/4}}\left( \lam{2}{12}^0, 0 \right), \ldots, \frac{1}{(NT)^{1/4}}\left( \lam{2}{N2}^0, 0 \right),
\frac{1}{(NT)^{1/4}}\left( 0, -f_{11}^0 \right), \ldots, \frac{1}{(NT)^{1/4}}\left( 0, -f_{T1}^0\right)
\right]^{\prime} \\
\delta_{3} & =\left[ 
\frac{1}{(NT)^{1/4}}\left( 0,\lam{1}{12}^0 \right), \ldots, \frac{1}{(NT)^{1/4}}\left(0,\lam{1}{N2}^0 \right), 
\frac{1}{(NT)^{1/4}}\left( 0, \lam{2}{12}^0 \right), \ldots, \frac{1}{(NT)^{1/4}}\left(0,\lam{2}{N2}^0\right),
\frac{1}{(NT)^{1/4}}\left( 0, -f_{12}^0 \right), \ldots, \frac{1}{(NT)^{1/4}}\left( 0, -f_{T2}^0\right)
\right]^{\prime} \\
\delta_{4} & =\left[ 
\frac{1}{(NT)^{1/4}}\left( 0,\lam{1}{11}^0 \right), \ldots, \frac{1}{(NT)^{1/4}}\left(0,\lam{1}{N1}^0 \right), 
\frac{1}{(NT)^{1/4}}\left( 0, \lam{2}{11}^0 \right), \ldots, \frac{1}{(NT)^{1/4}}\left(0,\lam{2}{N1}^0\right),
\frac{1}{(NT)^{1/4}}\left( -f_{12}^0, 0 \right), \ldots, \frac{1}{(NT)^{1/4}}\left( -f_{T2}^0, 0 \right)
\right]^{\prime} \\
\end{align*}
\end{minipage}
}
of length $2 \cdot (2N+T)$ and a matrix
$\delta = [\delta_1, \, \delta_2, \, \delta_3, \, \delta_4] \in \mathbb{R}^{2(2N+T)\times 4}.$
We write $\frac{\partial \mathcal{S}^{*}(\theta_0)}{\partial \theta^{\prime}}$ as \\
\resizebox{1.07\linewidth}{!}{
\begin{minipage}{\linewidth}
\begin{align*}
    & \frac{\partial \mathcal{S}^{*}\left(\theta_{0}\right)}{\partial \theta^{\prime}} \\[10pt]
    &  = \begin{bmatrix}
\frac{w_1}{\sqrt{NT}} \operatorname{diag}\left[\left\{\sum\limits_{t=1}^{T} \bar{\varrho}_{1,i t}^{(2)} f_t^0 f_t^{0\p}\right\}_{i \leq N}\right] &  \mathbf{0}_{2N\times 2N} & \frac{w_1}{\sqrt{NT}}\left\{\bar{\varrho}_{1,i t}^{(2)} f_t^0 \lambda_{1, i}^{0\p}\right\}_{i \leq N, t \leq T} \\
\mathbf{0}_{2N\times 2N} & \frac{w_2}{\sqrt{NT}} \operatorname{diag}\left[\left\{\sum\limits_{t=1}^{T} \bar{\varrho}_{2,i t}^{(2)} f_t^0 f_t^{0\p}\right\}_{i \leq N}\right] & \frac{w_2}{\sqrt{NT}}\left\{\bar{\varrho}_{2,i t}^{(2)} f_t^0 \lambda_{2,i}^{0\p}\right\}_{i \leq N, t \leq T}\\
\frac{w_1}{\sqrt{NT}}\left\{\bar{\varrho}_{1,i t}^{(2)} f_t^0 \lambda_{1, i}^{0\p}\right\}_{i \leq N, t \leq T} & \frac{w_2}{\sqrt{NT}}\left\{\bar{\varrho}_{2,i t}^{(2)} f_t^0 \lambda_{2,i}^{0\p}\right\}_{i \leq N, t \leq T} & \frac{1}{\sqrt{NT}} \operatorname{diag}\left[\left\{\sum\limits_{k=1}^2 w_k\sum\limits_{i=1}^{N} \bar{\varrho}_{k,i t}^{(2)} \lambda_{k,i}^0 \lambda_{k,i}^{0\p} \right\}_{t \leq T}\right] 
\end{bmatrix} \\[10pt]
&= \underline{b} \begin{bmatrix}
\frac{w_1}{\sqrt{NT}} \operatorname{diag}\left[\left\{\sum\limits_{t=1}^{T}  f_t^0 f_t^{0\p}\right\}_{i \leq N}\right] &  \mathbf{0}_{2N\times 2N} & \mathbf{0}_{2N\times 2N} \\
\mathbf{0}_{2N\times 2N} & \frac{w_2}{\sqrt{NT}} \operatorname{diag}\left[\left\{\sum\limits_{t=1}^{T}  f_t^0 f_t^{0\p}\right\}_{i \leq N}\right] & \mathbf{0}_{2N\times 2N}\\
\mathbf{0}_{2N\times 2N} & \mathbf{0}_{2N\times 2N} & \frac{1}{\sqrt{NT}} \operatorname{diag}\left[\left\{\sum\limits_{k=1}^2 w_k\sum\limits_{i=1}^{N}  \lambda_{k,i}^0 \lambda_{k,i}^{0\p} \right\}_{t \leq T}\right] 
\end{bmatrix} \\[10pt]
& \quad + \underline{b} \begin{bmatrix}
\mathbf{0}_{2N\times 2N} &  \mathbf{0}_{2N\times 2N} & \quad \quad \frac{w_1}{\sqrt{NT}}\left\{ f_t^0 \lambda_{1, i}^{0\p}\right\}_{i \leq N, t \leq T} \\
\mathbf{0}_{2N\times 2N} & \mathbf{0}_{2N\times 2N} & \quad \quad \frac{w_2}{\sqrt{NT}}\left\{ f_t^0 \lambda_{2,i}^{0\p}\right\}_{i \leq N, t \leq T}\\
\frac{w_1}{\sqrt{NT}}\left\{ f_t^0 \lambda_{1, i}^{0\p}\right\}_{i \leq N, t \leq T} \quad \quad & \frac{w_2}{\sqrt{NT}}\left\{ f_t^0 \lambda_{2,i}^{0\p}\right\}_{i \leq N, t \leq T} \, & \quad \quad \mathbf{0}_{2N\times 2N}
\end{bmatrix} \\[10pt]
& \quad + \begin{bmatrix}
\frac{w_1}{\sqrt{NT}} \operatorname{diag}\left[\left\{\sum\limits_{t=1}^{T} (\bar{\varrho}_{1,i t}^{(2)}-\underline{b}) f_t^0 f_t^{0\p}\right\}_{i \leq N}\right] &  \mathbf{0}_{2N\times 2N} & \frac{w_1}{\sqrt{NT}}\left\{(\bar{\varrho}_{1,i t}^{(2)}-\underline{b}) f_t^0 \lambda_{1, i}^{0\p}\right\}_{i \leq N, t \leq T} \\
\mathbf{0}_{2N\times 2N} & \frac{w_2}{\sqrt{NT}} \operatorname{diag}\left[\left\{\sum\limits_{t=1}^{T} (\bar{\varrho}_{2,i t}^{(2)}-\underline{b}) f_t^0 f_t^{0\p}\right\}_{i \leq N}\right] & \frac{w_2}{\sqrt{NT}}\left\{(\bar{\varrho}_{2,i t}^{(2)}-\underline{b}) f_t^0 \lambda_{2,i}^{0\p}\right\}_{i \leq N, t \leq T}\\
\frac{w_1}{\sqrt{NT}}\left\{(\bar{\varrho}_{1,i t}^{(2)}-\underline{b}) f_t^0 \lambda_{1, i}^{0\p}\right\}_{i \leq N, t \leq T} & \frac{w_2}{\sqrt{NT}}\left\{(\bar{\varrho}_{2,i t}^{(2)}-\underline{b}) f_t^0 \lambda_{2,i}^{0\p}\right\}_{i \leq N, t \leq T} & \frac{1}{\sqrt{NT}} \operatorname{diag}\left[\left\{\sum\limits_{k=1}^2 w_k\sum\limits_{i=1}^{N} (\bar{\varrho}_{k,i t}^{(2)}-\underline{b}) \lambda_{k,i}^0 \lambda_{k,i}^{0\p} \right\}_{t \leq T}\right] 
\end{bmatrix}. \\[10pt]
\end{align*}
\end{minipage}
}

We denote the three terms of RHS of the above equation as $\rm{I}, \rm{II}$, and $\rm{III}$. We first consider the first term, $\rm{I}$. Note that for $k=1,2$,
\[
\frac{w_k}{\sqrt{NT}} \sum\limits_{t=1}^{T}  f_t^0 f_t^{0\p} = w_k \sqrt{\frac{T}{N}}\cdot \frac{F^{0\p}F^0}{T} = w_k \sqrt{\frac{T}{N}}\cdot \mathbb{I}_2.
\]
Also, using that $\frac{\Lam_2^0 \Lam_2^{0\p}}{N}$ is positive semidefinite, for fixed $t$, 
\[
\frac{1}{\sqrt{NT}} \sum\limits_{k=1}^2w_k\sum\limits_{i=1}^{N}  \lambda_{k,i}^0 \lambda_{k,i}^{0\p} = \sqrt{\frac{N}{T}}\sum\limits_{k=1}^2 w_k \cdot \frac{\Lam_k^0 \Lam_k^{0\p}}{N}  \succeq \sqrt{\frac{N}{T}} \cdot  w_1 \cdot \begin{bmatrix}
    \nu_1(N) & 0 \\
    0 & \nu_2(N) 
\end{bmatrix}.
\]
By Assumption 2, and since $\nu_1(N) \geq \nu_2(N)$ and $\nu_2(N) \rightarrow \nu_2 >0$, there exist a constant $c>0$ such that 
\begin{equation} \label{lem11_I}
    \rm{I} \succeq c \cdot \mathbb{I}_{2(2N+T)}.
\end{equation}
Next, we consider $\rm{II}+\underline{b} \cdot D \delta \delta\p D$ for $D= \operatorname{diag}(w_1\cdot\mathbb{I}_{2N}, w_2 \cdot\mathbb{I}_{2N}, \mathbb{I}_{2T})$. For $j=1,\ldots,4$, we denote the first $2\cdot 2N$ elements of $\delta_j$ by $\delta_{j\Lam}$, which correspond to the portions associated with the loading matrices, and the remaining $2T$ elements by $\delta_{jF}$. Then, we have \\
\resizebox{0.95\linewidth}{!}{
\begin{minipage}{\linewidth}
\[
D \delta \delta^{\prime} D = \left[\begin{array}{ccc}
\multicolumn{2}{c}{\multirow{2}{*}{\Large D( $\sum_{j=1}^{4} \delta_{j\Lam} \delta_{j\Lam}^{\prime}$) D}} & -w_1 (N T)^{-1 / 2}\left\{f_{0 t} \lambda_{0 i}^{\prime}\right\}_{i \leq N, t \leq T} \\
 & & -w_2(N T)^{-1 / 2}\left\{f_{0 t} \lambda_{0 i}^{\prime}\right\}_{i \leq N, t \leq T} \\
-w_1 (N T)^{-1 / 2}\left\{f_{0 t} \lambda_{0 i}^{\prime}\right\}_{i \leq N, t \leq T} \quad \quad & -w_2 (N T)^{-1 / 2}\left\{f_{0 t} \lambda_{0 i}^{\prime}\right\}_{i \leq N, t \leq T} & D (\sum_{k=1}^{4} \delta_{j F} \delta_{j F}^{\prime}) D
\end{array}\right]
\]
\end{minipage}
}
\\[10pt]
and that
\begin{equation} \label{lem11_II}
\rm{II}+ \underline{b} \cdot D \delta \delta^{\prime} D = \left[\begin{array}{ccc}
\multicolumn{2}{c}{\multirow{2}{*}{\large D ($\sum_{j=1}^{4} \delta_{j\Lam} \delta_{j\Lam}^{\prime}$ )D}} & \mathbf{0}_{2N \times 1} \\
 & & \mathbf{0}_{2N \times 1} \\
\mathbf{0}_{1 \times 2N} & \quad \mathbf{0}_{1 \times 2N} & \quad D (\sum_{k=1}^{4} \delta_{j F} \delta_{j F}^{\prime}) D
\end{array}\right] \succeq 0.
\end{equation}
Finally, the last term $\rm{III}$ can be written as 
\[
\text{$\rm{III}$} = \frac{1}{\sqrt{N T}} \sum_{i=1}^{N} \sum_{t=1}^{T} w_1 \left(\bar{\varrho}_{1,i t}^{(2)}-\underline{b}\right) \mu_{1,i t} \mu_{1,i t}^{\prime} + w_2 \left(\bar{\varrho}_{2,i t}^{(2)}-\underline{b}\right) \mu_{2,i t} \mu_{2,i t}^{\prime},
\]
where
\begin{align*}
    &\mu_{1,it} = \underbrace{[\, \mathbf{0}_{1 \times 2}, \ldots, f_{t}^{0\p}, \ldots, \mathbf{0}_{1 \times 2}}_{1 \times 2 N}, \, \mathbf{0}_{1 \times 2N}, \,  \underbrace{ \mathbf{0}_{1 \times 2}, \ldots, \lambda_{1,i}^{0\p}, \ldots, \mathbf{0}_{1 \times 2} \, ]^{\prime}}_{1 \times 2 T}, \\
    &\mu_{2,it} = [\,\mathbf{0}_{1 \times 2N},  \underbrace{ \mathbf{0}_{1 \times 2}, \ldots, f_{t}^{0\p}, \ldots, \mathbf{0}_{1 \times 2}  }_{1 \times 2 N}, \,  \,  \underbrace{ \mathbf{0}_{1 \times 2}, \ldots, \lambda_{2,i}^{0\p}, \ldots, \mathbf{0}_{1 \times 2} \, ]^{\prime}}_{1 \times 2 T}. 
\end{align*}
By Lemma \ref{lem:normality_1} and Assumption 2, $\bar{\varrho}_{1,it}^{(2)}-\underline{b} \geq 0$ and $\bar{\varrho}_{2,it}^{(2)}-\underline{b} \geq 0$. Therefore, $\rm{III}\succeq 0$. By combining this result with (\ref{lem11_I}), and (\ref{lem11_II}), we have
\begin{equation} \label{eq:p.d.f}
\frac{\partial \mathcal{S}^{*}\left(\theta_{0}\right)}{\partial \theta^{\prime}} +\underline{b} \cdot D \delta \delta^\prime D = \rm{I}+(\rm{II}+ \underline{b} \cdot D \delta \delta\p D)+\rm{III} \succeq c \cdot \mathbb{I}_{2(2N+T)}.
\end{equation}
We let $\mathcal{H}^*=\frac{\partial \mathcal{S}^{*}\left(\theta_{0}\right)}{\partial \theta^{\prime}}$. Then, 
$ \mathcal{H}^*= \frac{1}{\sqrt{N T}} \sum_{i=1}^{N} \sum_{t=1}^{T} (w_1 \bar{\varrho}_{1,i t}^{(2)} \mu_{1,i t} \mu_{1,i t}^{\prime} + w_2 \bar{\varrho}_{2,i t}^{(2)} \mu_{2,i t} \mu_{2,i t}^{\prime} )$,
and from $\delta_j^\prime \mu_{1,it}=\delta_j^\prime \mu_{2,it}=0$, it holds that 
$\delta_j^\prime \mathcal{H}^* = 0 $ for $j=1,\ldots,4$. 
Note that
$
\frac{\Lam_1^{0\p} \Lam_1^0} {N}=\begin{bmatrix}
    \nu_{1}(N) & 0 \\
    0 & \nu_2(N)
\end{bmatrix}, \text{\, and we denote \,}
\frac{\Lam_2^{0\p} \Lam_2^0} {N}=\begin{bmatrix}
    e_{11}(N) & e_{12}(N) \\
    e_{12}(N) & e_{22}(N)
\end{bmatrix}.$ Then, we can write

\noindent
\resizebox{0.75\linewidth}{!}{
\begin{minipage}{\linewidth}
$$\delta\p \delta = \begin{bmatrix}
    \sqrt{\frac{N}{T}} \nu_{1}(N)+  \sqrt{\frac{N}{T}} e_{11}(N)+\sqrt{\frac{T}{N}} & \quad  \sqrt{\frac{N}{T}} e_{12}(N) & 0 & 0 \\
    \sqrt{\frac{N}{T}} e_{12}(N) &  \sqrt{\frac{N}{T}} \nu_{2}(N)+  \sqrt{\frac{N}{T}} e_{22}(N)+\sqrt{\frac{T}{N}} & 0 & 0 \\
    0 & 0 & \sqrt{\frac{N}{T}} \nu_{2}(N)+ \sqrt{\frac{N}{T}} e_{22}(N)+\sqrt{\frac{T}{N}} & \sqrt{\frac{N}{T}} e_{12}(N) \\
    0 & 0 & \sqrt{\frac{N}{T}} e_{12}(N) &   \sqrt{\frac{N}{T}}\nu_{1}(N)+ \sqrt{\frac{N}{T}} e_{11}(N)+\sqrt{\frac{T}{N}} 
\end{bmatrix} .
$$
\end{minipage}
}\\
For simplicity, we denote 
{\footnotesize
\[
a_{NT} \defeq  \sqrt{\frac{N}{T}} \nu_{1}(N)+  \sqrt{\frac{N}{T}} e_{11}(N)+\sqrt{\frac{T}{N}}, \quad b_{NT}\defeq   \sqrt{\frac{N}{T}} e_{12}(N), \quad c_{NT}\defeq \sqrt{\frac{N}{T}}\nu_{2}(N)+ \sqrt{\frac{N}{T}} e_{22}(N)+\sqrt{\frac{T}{N}} .
\]
}
Let $D_{NT} \defeq  a_{NT} \cdot c_{NT} - b_{NT}^2$, which is the determinant of both two $2 \times 2$ diagonal block matrices of $\delta\p \delta$. 
Since $e_{11}(N) \cdot e_{22}(N) - e_{12}(N)^2 \geq 0$, we have $D_{NT} \geq \nu_1(N)+\nu_2(N)$, which implies that for sufficiently large $N$, $D_{NT} > \nu_2 > 0$. Hence, we get $\operatorname{rank}(\delta)$ = 4. Moreover, 
using that $\operatorname{rank}(D\delta)=\operatorname{rank}(\delta)$, the orthogonality of $\delta_j$ to $\mathcal{H}^*$, and equation (\ref{eq:p.d.f}), we obtain $\operatorname{rank}(\mathcal{H}^*)=2(2N+T)-4$.

Next, we show that for some arbitrary $v \in \mathbb{R}^{2(2N+T)}$ with $\norm{v}=1 $, $v^\prime \mathcal{H}(\th_0) v > 0$. Since $\delta_j^\prime \mathcal{H}^* = 0 $ for all $j$, $v$ can be decomposed as $v = v_{\mathcal{H}^*}+v_{\delta}$ where $v_{\mathcal{H}^*}$ and $v_{\delta}$ are orthogonal, and are in the column space of $\mathcal{H}^*$ and $\delta$, respectively.  
Then, using that $\mathcal{H}(\th_0)=\mathcal{H}^*+b
\cdot\sum\limits_{j=1}^4 \gamma_j \gamma_j^\prime$, we can write
\[
v^\prime \mathcal{H}(\th_0) v = v^\prime\mathcal{H}^*v + bv^\prime \left( \sum\limits_{j=1}^4 \gamma_j \gamma_j^\prime \right) v
= v_\mathcal{H^*}^\prime \mathcal{H^*} v_\mathcal{H^*} + bv^\prime \left( \sum\limits_{j=1}^4 \gamma_j \gamma_j^\prime \right) v.
\]
In case where $v_{\mathcal{H}^*} \neq 0$, since $v_\mathcal{H^*}^\prime \mathcal{H^*} v_\mathcal{H^*} >0$ and $bv^\prime \left( \sum\limits_{j=1}^4 \gamma_j \gamma_j^\prime \right) v \geq 0$, it holds that $v^\prime \mathcal{H}(\th_0)v >0$.

Now we consider the case where $v_{\mathcal{H}^*} = 0$. We get $v^\prime \mathcal{H}(\th_0) v = bv_\delta^\prime \left( \sum\limits_{j=1}^4 \gamma_j \gamma_j^\prime \right) v_\delta$. We project $\gamma_j$ onto $\delta$, by setting $\gamma_j = \delta \beta_j + \zeta_j$ for $j=1,2,3,4$. Here, $\beta_j = (\delta\p \delta)^{-1}\delta\p \gamma_j $.
For $j=1,\ldots,4$, $\beta_j = (\delta\p \delta)^{-1}\delta\p \gamma_j$ can be computed as \\
\resizebox{0.91\linewidth}{!}{
\begin{minipage}{\linewidth}
$$
\begin{gathered}
\beta_{1}= \frac{(T/N)^{1/4}}{D_{NT}}\begin{bmatrix}
-c_{NT}\\
b_{NT}\\
0 \\
0
\end{bmatrix}, \, \beta_{2}=\frac{(T/N)^{1/4}}{D_{NT}} \begin{bmatrix}
0 \\
0 \\
-a_{NT}\\
b_{NT} 
\end{bmatrix},\, \beta_{3}=\frac{(T/N)^{1/4}}{D_{NT}} \begin{bmatrix}
b_{NT} \\
-a_{NT} \\
b_{NT}\\
-c_{NT} 
\end{bmatrix},\, \beta_{4}=\frac{w_1 (T/N)^{1/4}}{D_{NT}} \begin{bmatrix}
-\nu_{2}(N)\cdot b_{NT}\\
\nu_{2}(N)\cdot a_{NT} \\
-\nu_{1}(N)\cdot b_{NT}\\
\nu_{1}(N)\cdot c_{NT} 
\end{bmatrix}.
\end{gathered}
$$
\end{minipage}
}
\\[10pt]
We have $\nu_1(N) \rightarrow \nu_1 $ and $ \nu_2(N) \rightarrow \nu_2$ as $N \rightarrow \infty$ and the difference $\nu_1-\nu_2$ is bounded away from zero by some positive constant for sufficiently large $N$. 
Moreover, Assumption 3(f) ensures a uniform bound for both $D_{NT}$ and $T/N$ for large $N,T$. By these conditions, it follows that $\{\beta_j: j=1,\ldots,4\}$ are linearly independent, and thus for $B=\sum\limits_{j=1}^4 \beta_k \beta\p_k$, there exists a constant $\underline{\rho}>0$ such that $\rho_{min}(B)>\underline{\rho}$, implying that $B - \underline{\rho}\mathbb{I}_r \succ 0$. 
Next, write $v_\delta = \delta w$ for some $w \neq0 \in \mathbb{R}^4$. Then, 
\[ 
v\p_\delta \gamma_j \gamma_j \p v_\delta = w \p \delta\p (\delta \beta_j +  \zeta_j) ( \beta\p_j \delta\p +  \zeta\p_j)\delta w = w\p \delta\p \delta \beta_j \beta\p_j \delta\p \delta w.
\]
Therefore, $v^\prime \mathcal{H}(\th_0) v = bv_\delta^\prime ( \sum\limits_{j=1}^4 \gamma_j \gamma_j^\prime ) v_\delta
= b w\p (\delta\p \delta B \delta\p \delta) w $. Since $\delta\p \delta B \delta\p \delta \succ 0 $, it follows that $\mathcal{H}(\th_0)$ is positive definite (hence invertible), and $\mathcal{H}(\th_0) \succeq c \cdot \mathbb{I}_{2(2N+T)}$ for some $c>0$.
Next, define
$ \mathcal{C}=\mathcal{H}(\th_0)-\mathcal{H}_d. $ We can write
\[
\mathcal{C} = \begin{bmatrix}
\mathbf{0}_{2N\times 2N} &  \mathbf{0}_{2N\times 2N} & \frac{w_1}{\sqrt{NT}}\left\{\bar{\varrho}_{1,i t}^{(2)} f_t^0 \lambda_{1, i}^{0\p}\right\}_{i \leq N, t \leq T} \\
\mathbf{0}_{2N\times 2N} & \mathbf{0}_{2N\times 2N} & \frac{w_2}{\sqrt{NT}}\left\{\bar{\varrho}_{2,i t}^{(2)} f_t^0 \lambda_{2,i}^{0\p}\right\}_{i \leq N, t \leq T}\\
\frac{w_1}{\sqrt{NT}}\left\{\bar{\varrho}_{1,i t}^{(2)} f_t^0 \lambda_{1, i}^{0\p}\right\}_{i \leq N, t \leq T} & \frac{w_2}{\sqrt{NT}}\left\{\bar{\varrho}_{2,i t}^{(2)} f_t^0 \lambda_{2,i}^{0\p}\right\}_{i \leq N, t \leq T} & \mathbf{0}_{2N\times 2N}
\end{bmatrix} + b \cdot \sum\limits_{j=1}^4 \gamma_j \gamma_j^{\p}.
\] 
Now, we show $\|\mathcal{H}(\th_0)^{-1}-\mathcal{H}_{d}^{-1}\|_{\max}=O(T^{-1})$. We can write
\begin{align*}
\|\mathcal{H}(\th_0)^{-1}-\mathcal{H}_{d}^{-1}\|_{\max} &= \| -\mathcal{H}_{d}^{-1} \mathcal{C H}_{d}^{-1}+\mathcal{H}_{d}^{-1} \mathcal{C H}^{-1} \mathcal{C H}_{d}^{-1} \|_{\max} \\
& \leq \| \mathcal{H}_{d}^{-1} \mathcal{C H}_{d}^{-1} \|_{\max} + \| \mathcal{H}_{d}^{-1} \mathcal{C H}^{-1} \mathcal{C H}_{d}^{-1} \|_{\max}.
\end{align*}
Using the prior result, we obtain
\[ \mathcal{H}_d^{-1}\mathcal{C H}^{-1}(\th_0) \mathcal{C H}_{d}^{-1} \preceq c^{-1}\mathcal{H}_d^{-1}\mathcal{C}^2 \mathcal{H}_{d}^{-1}.
\]
Since all diagonal elements are non-negative for a positive semidefinite matrix, each diagonal element of 
$c^{-1}\mathcal{H}_d^{-1}\mathcal{C}^2 \mathcal{H}_{d}^{-1}$ is greater than or equal to the corresponding diagonal element of $\mathcal{H}_d^{-1}\mathcal{C H}^{-1}(\th_0) \mathcal{C H}_{d}^{-1}$.
Furthermore, because the largest absolute value in a symmetric positive semidefinite matrix is always found in one of its diagonal elements, it follows that 
\[
\| \mathcal{H}_{d}^{-1} \mathcal{C H}^{-1} \mathcal{C H}_{d}^{-1} \|_{\max} \leq \|c^{-1}\mathcal{H}_d^{-1}\mathcal{C}^2 \mathcal{H}_{d}^{-1}\|_{\max}.
\]
Note that by Assumption 2, $H_d^{-1}$ is a block diagonal matrix with all elements that are $O(1)$, and it is easy to show that $\| C \|_{\max}=O(T^{-1})$ and $\| C^2 \|_{\max}=O(T^{-1})$. Therefore, we obtain
\[ 
\|\mathcal{H}(\th_0)^{-1}-\mathcal{H}_{d}^{-1}\|_{\max} 
\leq 
\| \mathcal{H}_{d}^{-1} \mathcal{C H}_{d}^{-1} \|_{\max} + c^{-1} \| \mathcal{H}_d^{-1}\mathcal{C}^2 \mathcal{H}_{d}^{-1}\|_{\max} =O(T^{-1}).
\]
Now, we use (\ref{eq:expandS}) to obtain stochastic expansions of $\tilde{f}_t-f_t^0$, and $\tilde{\lambda}_{k,i}-\lambda_{k,i}^0$. Note that since 
$\frac{1}{N} \sum\limits_{i=1}^N \tilde{\lambda}_{k^*,i1} \tilde{\lambda}_{k^*,i2} = 0$ and  $\frac{1}{N} \sum\limits_{i=1}^N \lambda^0_{k^*,i1} \lambda^0_{k^*,i2} = 0$, we get $\frac{\partial \mathbb{P}_{N T}(\tilde{\theta})}{\partial \theta}=\frac{\partial \mathbb{P}_{N T}\left(\theta_{0}\right)}{\partial \theta}=0$. Therefore, 
\begin{equation} \label{eq:th-th_0}
\tilde{\theta}-\theta_{0}=\mathcal{H}(\th_0)^{-1} \mathcal{S}^{*}(\tilde{\theta})-\mathcal{H}(\th_0)^{-1} \mathcal{S}^{*}\left(\theta_{0}\right)-0.5 \mathcal{H}(\th_0)^{-1} \mathcal{R}(\tilde{\theta}).
\end{equation}
We define $\tilde{\mathcal{S}}^{*}(\theta)=\mathcal{S}_{N T}^{*}(\theta)-\mathcal{S}^{*}(\theta)$, where\\
\resizebox{0.93\linewidth}{!}{
\begin{minipage}{\linewidth}
$$
\begin{aligned}
    \mathbb{S}_{NT}^*(\th)= & \underbrace{ \left[
    -\frac{1}{\sqrt{N T}}  \sum\limits_{t=1}^{T} w_1 \varrho_1^{(1)}\left(X_{1 t}-\lambda_{1,1}^{\prime} f_{t}\right) f_{t}^{\prime}, \ldots, -\frac{1}{\sqrt{N T}}  \sum\limits_{t=1}^{T} w_1 \varrho_1^{(1)}\left(X_{N t}-\lambda_{1,N}^{\prime} f_{t}\right) f_{t}^{\prime} \right.}_{1 \times Nr},   \\
     & \ldots \ldots , \underbrace{-\frac{1}{\sqrt{N T}}  \sum\limits_{t=1}^{T} w_K \varrho_K^{(1)}\left(X_{1 t}-\lambda_{K,1}^{\prime} f_{t}\right) f_{t}^{\prime}, \ldots, -\frac{1}{\sqrt{N T}}  \sum\limits_{t=1}^{T} w_K \varrho_K^{(1)}\left(X_{N t}-\lambda_{K,N}^{\prime} f_{t}\right) f_{t}^{\prime}, }_{1 \times Nr} \\
     &  \underbrace{\left.-\frac{1}{\sqrt{N T}}  \sumk \sum\limits_{i=1}^{N} w_k\varrho_k^{(1)}\left(X_{i 1}-\lambda_{k,i}^{\prime} f_t\right) \lambda_{k,i}^{\prime}, \ldots, -\frac{1}{\sqrt{N T}}  \sumk \sum\limits_{t=1}^{T} w_k \varrho_k^{(1)}\left(X_{i T}-\lambda_{k,i}^{\prime} f_{T}\right) \lambda_{k,i}^{\prime}
     \right]\p }_{1 \times Tr} \in \mathbb{R}^{(KN+T)r}.
\end{aligned}
$$
\end{minipage}
}
\\
Since $\tilde{\th}$ minimizes $\mathbb{S}_{NT}(\th)$, we have $\frac{\partial}{\partial \lambda_{k,i}} S_{NT}(\tilde{\th}) = \frac{1}{N}\sumk \sum\limits_{i=1}^N w_k \cdot \varrho^{(1)}_k(X_{it}-\tilde{\lambda}_{k,i}^{\prime} \tilde{f}_{t}) \tilde{f}_{t} = 0$ and 
$\frac{\partial}{\partial f_t} S_{NT}(\tilde{\th}) = \frac{1}{N}\sumk \sum\limits_{i=1}^N w_k \cdot \varrho^{(1)}_k(X_{it}-\tilde{\lambda}_{k,i}^{\prime} \tilde{f}_{t}) \tilde{\lambda}_{k,i} = 0$. Therefore, $S^*_{NT}(\tilde{\th})=0$ and $\tilde{\mathcal{S}}^{*}(\theta)=-\mathcal{S}^{*}(\theta)$.
Also, let $\mathcal{D} = \mathcal{H}(\th_0)^{-1}-\mathcal{H}_d^{-1}$. Then, we can write
\begin{align} \label{eq:H-1S}
\begin{split}
 \mathcal{H}(\th_0)^{-1} \mathcal{S}^{*}(\tilde{\theta}) & = -\mathcal{H}(\th_0)^{-1} \tilde{\mathcal{S}}^{*}(\tilde{\theta}) =-\mathcal{H}_{d}^{-1} \tilde{\mathcal{S}}^{*}(\tilde{\theta}) - \mathcal{D} \mathcal{S}^{*}(\tilde{\theta}) \\
& = -\mathcal{H}_{d}^{-1} \tilde{\mathcal{S}}^{*}\left(\theta_{0}\right)-\mathcal{H}_{d}^{-1}\left(\tilde{\mathcal{S}}^{*}(\tilde{\theta})-\tilde{\mathcal{S}}^{*}\left(\theta_{0}\right)\right) -\mathcal{D} \tilde{\mathcal{S}}^{*}\left(\theta_{0}\right)-\mathcal{D}\left(\tilde{\mathcal{S}}^{*}(\tilde{\theta})-\tilde{\mathcal{S}}^{*}\left(\theta_{0}\right)\right).
\end{split}
\end{align}
Next, we write
\[
\mathcal{R}(\tilde{\th}) = \left[ \mathcal{R}(\tilde{\theta})_{1}, \ldots, \mathcal{R}(\tilde{\theta})_{N}, \ldots, \mathcal{R}(\tilde{\theta})_{(K-1)N+1}, \ldots,\mathcal{R}(\tilde{\theta})_{KN}, \mathcal{R}(\tilde{\theta})_{KN+1}, \ldots, \mathcal{R}(\tilde{\theta})_{KN+T} \right] \in \mathbb{R}^{(KN+T)r}, 
\]
where $\mathcal{R}(\tilde{\theta})_{j} \in \mathbb{R}^r$ for all $j=1,\ldots, KN+T$,
and
\[
\mathcal{D} = \{ \mathcal{D}_{j,s} \}_{j \leq KN+T\, , \, s \leq KN+T} \in \mathbb{R}^{(KN+T)r \times (KN+T)r}, 
\]
where each $\mathcal{D}_{j,s}$ is an $r \times r$ matrix representing the $(j,s)$-th block of $\mathcal{D}$.
Note that by Assumption 2 and Lemma \ref{lem:normality_2}, for sufficiently large $N, T$, 
\[
\frac{1}{\sqrt{NT}}\sum\limits_{t=1}^T | \tilde{f}_{tp} -f_{tp}^0 | \leq \frac{1}{\sqrt{NT}}\sum\limits_{t=1}^T \| \tilde{f}_{t} -f_{t}^0 \| \lesssim \frac{1}{\sqrt{T}} \| \tilde{F}-F^0 \| = O_p(L_{NT}^{-1})  \text{ for all } p=1,\ldots r,
\] 
and similarly, for fixed $k$, $\frac{1}{\sqrt{NT}}\sum\limits_{i=1}^N | \tilde{\lambda}_{k,iq} -\lambda_{k,iq}^0 | = O_p(L_{NT}^{-1}) $ for all $q=1,\ldots,r$. Using these facts, we can show that for fixed $k$,
\begin{equation} \label{eq:Rlam}
\mathcal{R}(\tilde{\theta})_{(k-1)N+i}=\bar{O}_{P}(1) \cdot \|\tilde{\lambda}_{k,i}-\lambda_{k,i}^0 \|^{2}+\bar{O}_{P}(L_{NT}^{-1}) \cdot \| \tilde{\lambda}_{k,i}-\lambda_{k, i}^0 \|+\bar{O}_{P}(L_{NT}^{-2}) , \quad i=1,\ldots, N,
\end{equation}
and 
\begin{equation}\label{eq:Rfac}
\mathcal{R}(\tilde{\theta})_{KN+t}=\bar{O}_{P}(1) \cdot \|\tilde{f_{t}}-f_{0 t}\|^{2}+\bar{O}_{P}(L_{NT}^{-1}) \cdot \|\tilde{f}_{t}-f_{0 t}\|+\bar{O}_{P}(L_{NT}^{-2}), \quad t=1,\ldots, T.
\end{equation}
Also, by Lemmas \ref{lem:normality_1} and \ref{lem:normality_5}, we get 
\[
\left\|\mathcal{H}(\th_0)^{-1} \mathcal{S}^{*}\left(\theta_{0}\right)\right\|_{\max } \leq \left\|\mathcal{H}_d^{-1} \mathcal{S}^{*}\left(\theta_{0}\right)\right\|_{\max } +\left\|D \mathcal{S}^{*}\left(\theta_{0}\right)\right\|_{\max } = \bar{O}\left(h^{m}\right).
\]
Now we can obtain the stochastic expansion of $\tilde{\lambda}_{k,i}$, $\tilde{f_t}$ from (\ref{eq:th-th_0}) and (\ref{eq:H-1S}) as for fixed $k$,
\begin{align} \label{eq:lamexpand}
\begin{split}
\tilde{\lambda}_{k,i}&-\lambda_{k,i}^0 \\
= & \left(\Phi_{T, k,i}^*\right)^{-1} \frac{w_k}{T} \sum_{s=1}^{T} \tilde{\varrho}_{k,is}^{(1)} f_s^0 
+
\left(\Phi_{T, k,i}^*\right)^{-1} \frac{w_k}{T} \sum_{s=1}^{T}\left\{\tilde{\varrho}_k^{(1)}\left(X_{j t}-\tilde{\lambda}_{k,i}^{\prime} \tilde{f}_{s}\right) \tilde{f}_{s}-\tilde{\varrho}_{k,is}^{(1)} f_s^0\right\}  \\
& + \frac{1}{\sqrt{N T}} \sum\limits_{\ell=1}^K w_\ell  \sum_{j=1}^{N}\sum_{s=1}^{T} \mathcal{D}_{(k-1)N+i,(\ell-1)N+j}  \tilde{\varrho}_{\ell,j s}^{(1)}  f_{s}^0  +
\frac{1}{\sqrt{N T}} \sum\limits_{\ell=1}^K \sum_{j=1}^{N} \sum_{s=1}^{T} w_\ell \mathcal{D}_{(k-1)N+i, KN+s} \tilde{\varrho}_{\ell,j s}^{(1)}  \lambda_{\ell,j}^0 \\
& +\frac{1}{\sqrt{N T}} \sum\limits_{\ell=1}^K \sum_{j=1}^{N} \sum_{s=1}^{T} w_\ell \mathcal{D}_{(k-1)N+i, (\ell-1)N+j} \left\{\tilde{\varrho}_\ell^{(1)}\left(X_{j s}-\tilde{\lambda}_{\ell,j}^{\prime} \tilde{f}_{s}\right) \tilde{f}_{s}-\tilde{\varrho}_{\ell, j s}^{(1)} f_{s}^0 \right\} \\
& +\frac{1}{\sqrt{N T}}  \sum\limits_{\ell=1}^K \sum_{j=1}^{N} \sum_{s=1}^{T} w_\ell \mathcal{D}_{(k-1)N+i,KN+s}\left\{\tilde{\varrho}_\ell^{(1)}\left(X_{j s}-\tilde{\lambda}_{\ell,j}^{\prime} \tilde{f}_{s}\right) \tilde{\lambda}_{\ell,j}-\tilde{\varrho}_{\ell,j s}^{(1)} \lambda_{\ell,j}^0\right\} \\
& -0.5\left(\Phi_{T, k,i}^*\right)^{-1} \mathcal{R}(\tilde{\theta})_{(k-1)N+i}
-0.5 \sum\limits_{\ell=1}^K \sum_{j=1}^{N} \mathcal{D}_{(k-1)N+i, (\ell-1)N+j} \mathcal{R}(\tilde{\theta})_{(\ell-1)N+j} \\
& -0.5 \sum_{s=1}^{T} \mathcal{D}_{(k-1)N+i, KN+s} \mathcal{R}(\tilde{\theta})_{KN+s}+ \bar{O}\left(h^{m}\right),
\end{split}
\end{align}
and 
\begin{align}\label{eq:facexpand}
\begin{split}
 \tilde{f}_{t} & -f_{0 t} \\
& =  \left(\Psi_{N, t}\right)^{-1} \frac{1}{N} \sum_{\ell=1}^K \sum_{j=1}^{N} w_\ell\tilde{\varrho}_{\ell,j t}^{(1)} \lambda_{\ell,j}^0 
+
\left(\Psi_{N, t}\right)^{-1} \frac{1}{N} \sum\limits_{\ell=1}^K w_\ell \sum_{j=1}^{N}\left\{\tilde{\varrho}_\ell^{(1)}\left(X_{j t}-\tilde{\lambda}_{\ell,j}^{\prime} \tilde{f}_{t}\right) \tilde{\lambda}_{\ell,j}-\tilde{\varrho}_{\ell,j t}^{(1)} \lambda_{\ell,j}^0\right\} \\
& + \frac{1}{\sqrt{N T}} \sum\limits_{\ell=1}^K w_\ell \sum_{j=1}^{N} \sum_{s=1}^{T} \mathcal{D}_{KN+t, (\ell-1)N+j} \cdot \tilde{\varrho}_{\ell,j s}^{(1)} \cdot f_s^0 +  \frac{1}{\sqrt{N T}} \sum\limits_{\ell=1}^K  \sum_{j=1}^{N} \sum_{s=1}^{T} w_\ell \mathcal{D}_{KN+t, KN+s} \cdot \tilde{\varrho}_{\ell,j s}^{(1)} \cdot \lambda_{\ell,j}^0 \\
& + \frac{1}{\sqrt{N T}}\sum\limits_{\ell=1}^K  \sum_{j=1}^{N} \sum_{s=1}^{T} w_\ell \mathcal{D}_{KN+t, (\ell-1)N+j}\left\{\tilde{\varrho}_\ell^{(1)}\left(X_{j s}-\tilde{\lambda}_{\ell,j}^{\prime} \tilde{f}_{s}\right) \tilde{f}_{s}-\tilde{\varrho}_{\ell,j s}^{(1)} f_{s}^0\right\} \\
& +\frac{1}{\sqrt{N T}} \sum\limits_{\ell=1}^K  \sum_{j=1}^{N} \sum_{s=1}^{T} w_\ell \mathcal{D}_{KN+t, KN+s}\left\{\tilde{\varrho}_\ell^{(1)}\left(X_{j s}-\tilde{\lambda}_{\ell,j}^{\prime} \tilde{f}_{s}\right) \tilde{\lambda}_{\ell,j}-\tilde{\varrho}_{\ell,j s}^{(1)} \lambda_{\ell, j}^0 \right\} \\
& -0.5 \left(\Psi_{N, t}\right)^{-1} \mathcal{R}(\tilde{\theta})_{KN+t} - 0.5 \sum\limits_{j=1}^{KN} \mathcal{D}_{KN+t,j} \mathcal{R}(\tilde{\th})_j -0.5 \sum\limits_{s=1}^{T} \mathcal{D}_{KN+t,KN+s} \mathcal{R}(\tilde{\th})_{KN+s} + \bar{O}(h^m),
\end{split}
\end{align}

\end{proof} 

\begin{lemma} \label{lem:normality_6} Let $\{a_i\}_{i=1}^N$ and $\{b_t\}_{t=1}^T$  be sequences of uniformly bounded constants. It holds that 
\[
\frac{1}{N} \sum_{i=1}^{N} a_{i}\left(\tilde{\lambda}_{k,i}-\lambda_{k,i}^0\right)=O_{P}\left(\frac{1}{N h}\right) \text{\, and \,} ~\frac{1}{T} \sum_{t=1}^{T} b_{t}\left(\tilde{f}_{t}-f_t^0\right)=O_{P}\left(\frac{1}{T h}\right). 
\]
\end{lemma}
\begin{proof}
We first show the first result. For $\ell=1,\ldots,K$ and $ j=1, \ldots, T$, let 
\[
d_{\ell,j}=\sqrt{N T} \cdot \frac{1}{N} \sum\limits_{i=1}^{N} a_{i} \cdot \mathcal{D}_{(k-1)N+i,(\ell-1)N+j},
\]
and for $s=1,\ldots,T$,
\[
d_s = \sqrt{N T} \cdot \frac{1}{N} \sum\limits_{i=1}^{N} a_{i} \cdot \mathcal{D}_{(k-1)N+i,KN+s}.
\]
Then, by (\ref{eq:lamexpand}), we get 
\begin{align*} 
\begin{split}
\frac{1}{N}&\sum\limits_{i=1}^N a_i (\tilde{\lambda}_{k,i}-\lambda_{k,i}^0) \\
& =  \frac{1}{N}\sum\limits_{i=1}^N a_i  \left(\Phi_{T,k,i}^*\right)^{-1} \frac{w_k}{T} \sum_{s=1}^{T} \tilde{\varrho}_{k,js}^{(1)} f_s^0 
+
\frac{1}{N}\sum\limits_{i=1}^N a_i \left(\Phi_{T,k,i}^*\right)^{-1} \frac{w_k}{T} \sum_{s=1}^{T}\left\{\tilde{\varrho}_k^{(1)}\left(X_{j t}-\tilde{\lambda}_{k,j}^{\prime} \tilde{f}_{s}\right) \tilde{f}_{s}-\tilde{\varrho}_{k,js}^{(1)} f_s^0\right\}  \\
& + \frac{1}{NT}\sum\limits_{\ell=1}^K w_\ell  \sum_{j=1}^{N}\sum_{s=1}^{T} d_{\ell,i} \cdot \tilde{\varrho}_{\ell,j s}^{(1)} \cdot f_{s}^0  
+
\frac{1}{NT}\sum\limits_{\ell=1}^K  w_\ell \sum_{j=1}^{N} \sum_{s=1}^{T} d_s \cdot \tilde{\varrho}_{\ell,j s}^{(1)} \cdot \lambda_{\ell,j}^0 \\
& +\frac{1}{NT} \sum\limits_{\ell=1}^K w_\ell \sum_{j=1}^{N} \sum_{s=1}^{T} d_{\ell,j} \left\{\tilde{\varrho}_\ell^{(1)}\left(X_{j s}-\tilde{\lambda}_{\ell,j}^{\prime} \tilde{f}_{s}\right) \tilde{f}_{s}-\tilde{\varrho}_{\ell, j s}^{(1)} f_{s}^0 \right\} \\
& +\frac{1}{NT}  \sum\limits_{\ell=1}^K w_\ell \sum_{j=1}^{N} \sum_{s=1}^{T}  d_s\left\{\tilde{\varrho}_\ell^{(1)}\left(X_{j s}-\tilde{\lambda}_{\ell,j}^{\prime} \tilde{f}_{s}\right) \tilde{\lambda}_{\ell,j}-\tilde{\varrho}_{\ell,j s}^{(1)} \lambda_{\ell,j}^0\right\} \\
& -0.5 \frac{1}{N} \sum\limits_{i=1}^N a_i \left(\Phi_{T,k,i}^*\right)^{-1} \mathcal{R}(\tilde{\theta})_{(k-1)N+i}
-0.5 \frac{1}{\sqrt{NT}} \sum\limits_{\ell=1}^K \sum_{j=1}^{N} d_{\ell,j} \mathcal{R}(\tilde{\theta})_{(\ell-1)N+j} \\
& -0.5  \frac{1}{\sqrt{NT}} \sum_{s=1}^{T} d_s \mathcal{R}(\tilde{\theta})_{KN+s}+ \bar{O}\left(h^{m}\right).
\end{split}
\end{align*}
Note that by Lemma \ref{lem:normality_4}, for all $\ell=1,\ldots,K$, $\underset{1 \leq j \leq N}{\max}\left\|d_{\ell,j}\right\|$  and $\underset{1 \leq s \leq T}{\max} \left\|d_s \right\|$ are bounded. The first, third, and fourth terms are $O_p((NT)^{-1/2})$ by Lyapunov's CLT. Also, for the last four terms, $\bar{O}(h^m)$ is $O_p(L_{NT}^{-2})$ by Assumption 3, and the other three terms can be shown to be $O_p(L_{NT}^{-2})$ by using 
\[
\frac{1}{N}\sum_{i=1}^N \mathcal{R}(\tilde{\theta})_{(k-1)N+i} = \bar{O}_p(L_{NT}^{-2}) \text{\, and \,} \frac{1}{T}\sum_{t=1}^T \mathcal{R}(\tilde{\theta})_{KN+s} = \bar{O}_p(L_{NT}^{-2}), 
\]
which can be derived from (\ref{eq:Rlam}) and (\ref{eq:Rfac}). Now, we show that the remaining terms, the second, fifth, and sixth terms, are $O_p((Nh)^{-1})$. We focus on showing that 
\[
\frac{1}{NT}  \sum\limits_{\ell=1}^K w_\ell \sum_{j=1}^{N} \sum_{s=1}^{T}  d_{\ell,j}\left\{\tilde{\varrho}_\ell^{(1)}\left(X_{j s}-\tilde{\lambda}_{\ell,j}^{\prime} \tilde{f}_{s}\right) \tilde{f}_s-\tilde{\varrho}_{\ell,j s}^{(1)}f_s^0\right\} = O_p((Nh)^{-1}),
\]
since the results of the other two terms can be shown in the same way. Let 
\[
\mathbb{V}_{N T}(\theta)=\frac{1}{N T} \sum\limits_{\ell=1}^K w_\ell \sum_{j=1}^{N} \sum_{s=1}^{T} d_{\ell,j}\left\{\tilde{\varrho}_\ell^{(1)}\left(X_{j s}-\lambda_{\ell,j}^{\prime} f_{s}\right) f_{s}-\tilde{\varrho}_{\ell,j s}^{(1)} f_{ s}^0\right\}
\]
and define
\[
\Delta_{N T}\left(\theta_{a}, \theta_{b}\right)=\sqrt{N T} h\left[\mathbb{V}_{N T}\left(\theta_{a}\right)-\mathbb{V}_{N T}\left(\theta_{b}\right)\right].
\]
Then, we can write
\begin{align*}
    \Delta_{N T}\left(\theta_{a}, \theta_{b}\right)&=\frac{h}{\sqrt{N T}} \sum\limits_{\ell=1}^K w_\ell \sum_{j=1}^{N} \sum_{s=1}^{T} d_{j} \cdot \tilde{\varrho}_{\ell,j}^{(1)}\left(X_{j s}-\lambda_{\ell, j}^{a\p} f_{ s}^a\right) \cdot\left(f_{s}^a-f_{s}^b\right)
    \\
    & \qquad~~~+
    \frac{h}{\sqrt{N T}} \sum\limits_{\ell=1}^K w_\ell \sum_{j=1}^{N} \sum_{s=1}^{T} d_{j} \cdot\left[\tilde{\varrho}_{\ell,j}^{(1)}\left(X_{j s}-\lambda_{\ell,j}^{a\p} f_{s}^a\right)-\tilde{\varrho}_{\ell,j}^{(1)}\left(X_{j s}-\lambda_{\ell,j}^{b\p} f_{s}^b\right)\right] \cdot f_{s}^b \\
    & = \Delta_{1, N T}\left(\theta_{a}, \theta_{b}\right) 
    + 
    \Delta_{2, N T}\left(\theta_{a}, \theta_{b}\right). 
\end{align*}
Following a similar way as in the proof of Lemma S.7 in \cite{chen2021quantile}, we can show  
$\| \Delta_{NT}(\th^a, \th^b) \|_{\psi_2} \lesssim d(\th^a, \th^b)$ for sufficiently small $d(\th^a, \th^b)$, and furthermore, it holds that\\ $\mathbb{E}\left[\sup _{\theta \in \Theta^{r}(\delta)}\left|\mathbb{V}_{N T}(\theta)\right|\right] \lesssim \frac{\delta}{L_{N T} h}$. Finally, by Lemma \ref{lem:normality_2}, we get $\mathbb{V}_{NT}(\tilde{\th})=O_p((L_{NT}^2h)^{-1})$. By combining the results, we have
\[
\frac{1}{N} \sum_{i=1}^{N} a_{i}\left(\tilde{\lambda}_{k,i}-\lambda_{k,i}^0\right) = O_p((NT)^{-1/2})+O_p((L_{NT}^{-2})) + O_p((L_{NT}^2h)^{-1}),
\]
which is $O_p((Nh)^{-1})$ by Assumption 2. 

The second result can be shown in a similar way by using (\ref{eq:facexpand}) and similarly defining 
\[
d^*_{\ell,j}=\sqrt{N T} \cdot \frac{1}{T} \sum\limits_{t=1}^{T} b_{t} \cdot \mathcal{D}_{KN+t,(\ell-1)N+j},
\]
for $\ell=1,\ldots,K$, $ j=1, \ldots, T$, and
\[
d^*_s = \sqrt{N T} \cdot \frac{1}{T} \sum\limits_{t=1}^{T} b_{t} \cdot \mathcal{D}_{KN+t,KN+s}
\]
for $s=1,\ldots,T$.
\end{proof}

\begin{lemma} \label{lem:normality_7}
For each $k$ and $t$, it holds that
\[
\frac{1}{N} \sum_{i=1}^{N} \tilde{\varrho}_{k,i t}^{(1)}\left(\tilde{\lambda}_{k,i}-\lambda_{k,i}^0\right)=O_{P}\left(\frac{1}{N h}\right) \quad \text { and } \quad \frac{1}{N} \sum_{i=1}^{N} \tilde{\varrho}_{k,i t}^{(2)} \lambda_{k,i}^0\left(\tilde{\lambda}_{k,i}-\lambda_{k,i}^0\right)^{\prime}=O_{P}\left(\frac{1}{N h^{2}}\right).
\]
Also, for each $k$ and $i$,
\[
\frac{1}{T} \sum_{t=1}^{T} \tilde{\varrho}_{k,i t}^{(1)}\left(\tilde{f}_{t}-f_t^0\right)=O_{P}\left(\frac{1}{T h}\right) \quad \text { and } \quad \frac{1}{T} \sum_{t=1}^{T} \tilde{\varrho}_{k,i t}^{(2)} f_{t}^0\left(\tilde{f}_{t}-f_{t}^0\right)^{\prime}=O_{P}\left(\frac{1}{T h^{2}}\right).
\]
\end{lemma}
\begin{proof}
    The four results can be shown in a similar way, so we focus on proving the second result. Note that 
    \begin{align*}
    \mathcal{H}(\th_0)^{-1} \mathcal{S}^{*}(\tilde{\theta})
    &=
    \mathcal{H}_{d}^{-1} \mathcal{S}^{*}(\tilde{\theta})+\mathcal{D} \mathcal{S}^{*}(\tilde{\theta})\\ 
    &= 
    -\mathcal{H}_{d}^{-1} \mathcal{S}^{*}(\tilde{\theta})+\mathcal{D} \tilde{\mathcal{S}}^{*}(\tilde{\theta})\\
    &= -\mathcal{H}_{d}^{-1} \tilde{\mathcal{S}}^{*}\left(\theta_{0}\right)-\mathcal{H}_{d}^{-1}\left(\tilde{\mathcal{S}}^{*}(\tilde{\theta})-\tilde{\mathcal{S}}^{*}\left(\theta_{0}\right)\right)+\mathcal{D} \mathcal{S}^{*}(\tilde{\theta}). 
    \end{align*}
    Then, using (\ref{eq:th-th_0}), we get 
\begin{align} \label{eq:lem13_1}
\begin{split}
& \frac{1}{N} \sum_{i=1}^N \tilde{\varrho}_{k,it}^{(2)} \lambda_{k,i}^0\left(\tilde{\lambda}_{k,i}-\lambda_{k,i}^0\right)^{\prime} \\
& =\frac{w_k}{N T} \sum_{i=1}^N \sum_{s=1}^T \tilde{\varrho}_{k, i t}^{(2)} \tilde{\varrho}_{k, i s}^{(1)} \lambda_{k, i}^0 f_s^0\left(\Phi_{T,k,i}^*\right)^{-1} \\
& +\frac{w_k}{N T} \sum_{i=1}^N \sum_{s=1}^T \tilde{\varrho}_{k, i t}^{(2)} \lambda_{k,i}^0 \cdot \left\{ \tilde{\varrho}_k^{(1)}\left(X_{i s}-\lambda_{k,i}^{\prime} \tilde{f}_s\right) \tilde{f}_s^{\prime}-\tilde{\varrho}_{k, is}^{(1)} f_s^0 \right\} \left(\Phi_{T,k,i}^*\right)^{-1} \\
& +\frac{1}{\sqrt{N T}} \sum_{l=1}^k \sum_{j=1}^N \sum_{s=1}^T\left(\frac{1}{N} \sum_{i=1}^N w_l \tilde{\varrho}_{k,it}^{(2)} \lambda_{k,i}^0 \tilde{f}_s^{\prime} \mathcal{D}\p_{(k-1) N+i,(\ell-1) N+j} \cdot \tilde{\varrho}_\ell^{(1)}\left(X_{j s}-\tilde{\lambda}_{\ell, j}^{\prime} \tilde{f}_s\right)\right) \\
& +\frac{1}{\sqrt{N T}} \sum_{l=1}^k \sum_{j=1}^N \sum_{s=1}^T\left(\frac{1}{N} \sum_{i=1}^N w_l \tilde{\varrho}_{k,it}^{(2)} \lambda_{k,i}^0 \tilde{\lambda}_{\ell,j}^{\prime} \mathcal{D}\p_{(k-1) N+i,KN+s} \cdot \tilde{\varrho}_\ell^{(1)}\left(X_{j s}-\tilde{\lambda}_{\ell, j}^{\prime} \tilde{f}_s\right)\right) \\
& -0.5 \frac{1}{N} \sum_{i=1}^N \tilde{\varrho}_{k,i t}^{(2)} \lambda_{k,i}^0 \mathcal{R}(\tilde{\theta})\p_{(k-1) N+i}\left(\Phi_{T,k,i}^*\right)^{-1}
-
0.5  \frac{1}{N} \sum_{i=1}^N \sum_{l=1}^k \sum_{j=1}^N  \tilde{\varrho}_{k,i t}^{(2)} \lambda_{k,i}^0  \mathcal{R}(\tilde{\theta})\p_{(\ell-1)N+j} \mathcal{D}_{(k-1) N+i,(\ell-1)N+j}\p\\
& -0.5 \frac{1}{N} \sum_{i=1}^N \sum_{s=1}^T\tilde{\varrho}_{k,i t}^{(2)} \lambda_{k,i}^0  \mathcal{R}(\tilde{\theta})\p_{KN+s}\mathcal{D}_{(k-1) N+i, KN+s}\p 
+{O}\left(h^{m-1}\right).
\end{split}
\end{align}
We now determine the stochastic order of each term of the RHS of the above equation.
The first term can be written as
\begin{align} \label{eq:lem11_1st}
\begin{split}
    \frac{w_k}{N T} \sum_{i=1}^N \sum_{s=1}^T & \tilde{\varrho}_{k, i t}^{(2)} \tilde{\varrho}_{k, i s}^{(1)} \lambda_{k, i}^0 f_s^0\left(\Psi_{N,k,i}\right)^{-1} \\
    & = \frac{w_k}{N T} \sum_{i=1}^N \tilde{\varrho}_{k, i t}^{(2)} \tilde{\varrho}_{k, i t}^{(1)} \lambda_{k, i}^0 f_t^0\left(\Phi_{T,k,i}^*\right)^{-1}
    +
    \frac{w_k}{N T} \sum_{i=1}^N \sum_{s\neq t}^T \tilde{\varrho}_{k, i t}^{(2)} \tilde{\varrho}_{k, i s}^{(1)} \lambda_{k, i}^0 f_s^0\left(\Phi_{T,k,i}^*\right)^{-1}.
    \end{split}
\end{align}
By Lemma \ref{lem:normality_1} and Assumption 3, the first term of (\ref{eq:lem11_1st}) is $O_p(\frac{1}{Th})$. For the second term of (\ref{eq:lem11_1st}), using that $\tilde{\varrho}_{k, i t}^{(2)}$ and $\tilde{\varrho}_{k, i s}^{(1)}$ are independent, by Lyapunov's CLT and Lemma \ref{lem:normality_1}, we can show that it is $ O_p(\frac{1}{\sqrt{NTh}})$. Therefore, the first term is $O_p(\frac{1}{Th})$.

Next, the second term of (\ref{eq:lem11_1st}) can be written as
\begin{align} \label{eq:lem11_2nd}
\begin{split}
\frac{w_k}{N T} \sum_{i=1}^N \sum_{s=1}^T & \tilde{\varrho}_{k, i t}^{(2)} \lambda_{k,i}^0 \cdot \left\{ \tilde{\varrho}_k^{(1)}\left(X_{i s}-\lambda_{k,i}^{\prime} \tilde{f}_s\right) \tilde{f}_s^{\prime}-\tilde{\varrho}_{k, is}^{(1)} f_s^0 \right\} \left(\Psi_{N, k,i}\right)^{-1}  \\
& = \frac{w_k}{N T} \sum_{i=1}^N \tilde{\varrho}_{k, i t}^{(2)} \lambda_{k,i}^0 \cdot \left\{ \tilde{\varrho}_k^{(1)}\left(X_{i t}-\lambda_{k,i}^{\prime} \tilde{f}_t\right) \tilde{f}_t^{\prime}-\tilde{\varrho}_{k, it}^{(1)} f_t^0 \right\} \left(\Psi_{N, k,i}\right)^{-1} \\
& ~~~ +  \frac{w_k}{N T} \sum_{i=1}^N \sum_{s \neq t}^T \tilde{\varrho}_{k, i t}^{(2)} \lambda_{k,i}^0 \cdot \left\{ \tilde{\varrho}_k^{(1)}\left(X_{i s}-\lambda_{k,i}^{\prime} \tilde{f}_s\right) \tilde{f}_s^{\prime}-\tilde{\varrho}_{k, is}^{(1)} f_s^0 \right\} \left(\Psi_{N, k,i}\right)^{-1}.  
\end{split}
\end{align}
By Lemma \ref{lem:normality_1} and Assumption 2, the first term of (\ref{eq:lem11_2nd}) is $O_p(\frac{1}{Th})$. The second term has the a similar form with $\mathbb{V}_{NT}(\tilde{\th})$ in Lemma \ref{lem:normality_6}, just need to consider that $\underset{1 \leq i \leq N}{\max}\|\tilde{\varrho}_{k, i t}^{(2)} \lambda_{k,i}^0 \left(\Psi_{N, k,i}\right)^{-1} \| = O(h^{-1})$. Therefore, it can be shown that the second term is $O_p(\frac{1}{Nh^2})$ in a similar way in Lemma \ref{lem:normality_6}. 

The third and fourth terms have the same order and can be proved in the same way. Thus, we only show the order of the fourth term.
Using that 
\[
\frac{\partial}{\partial f_t} S_{NT}(\tilde{\th}) = \frac{1}{N}\sum\limits_{\ell=1}^K \sum\limits_{i=1}^N w_\ell \cdot \varrho^{(1)}_\ell(X_{it}-\tilde{\lambda}_{\ell,i}^{\prime} \tilde{f}_{t}) \tilde{\lambda}_{\ell,i} = 0,\] we get
\begin{align*}
    \frac{1}{\sqrt{N T}} & \sum_{l=1}^K \sum_{j=1}^N \sum_{s=1}^T\left(\frac{1}{N} \sum_{i=1}^N w_l \tilde{\varrho}_{k,it}^{(2)} \lambda_{k,i}^0 \tilde{\lambda}_{\ell,j}^{\prime} \mathcal{D}_{(k-1) N+i,KN+s} \cdot \tilde{\varrho}_\ell^{(1)}\left(X_{j s}-\tilde{\lambda}_{\ell, j}^{\prime} \tilde{f}_s\right)\right) \\
    & = \frac{1}{\sqrt{N T}} \sum_{s=1}^T\left\{\frac{1}{N} \sum_{i=1}^N \tilde{\varrho}_{k,it}^{(2)} \lambda_{k,i}^0 \left( \sum_{l=1}^K \sum_{j=1}^N w_l \tilde{\lambda}_{\ell,j}^{\prime} \cdot \tilde{\varrho}_\ell^{(1)}\left(X_{j s}-\tilde{\lambda}_{\ell, j}^{\prime} \tilde{f}_s\right) \right) \mathcal{D}\p_{(k-1) N+i,KN+s} \right\}\\
    & = \frac{1}{\sqrt{N T}} \sum_{s=1}^T\left\{\frac{1}{N} \sum_{i=1}^N \tilde{\varrho}_{k,it}^{(2)} \lambda_{k,i}^0 \left( \sum_{l=1}^K \sum_{j=1}^N w_l \tilde{\lambda}_{\ell,j}^{\prime} \cdot \bar{\varrho}_\ell^{(1)}\left(X_{j s}-\tilde{\lambda}_{\ell, j}^{\prime} \tilde{f}_s\right) \right) \mathcal{D}\p_{(k-1) N+i,KN+s} \right\}.
\end{align*}
Now we define
\[
\mathcal{Z}_{k,t,s} = \frac{1}{N}\sum\limits_{i=1}^N \tilde{\varrho}_{k,it}^{(2)}\lambda_{k,ip}^0 \cdot \sqrt{NT}\mathcal{D}_{(k-1)N+i,KN+s}[q],
\]
where $\mathcal{D}_{(\cdot,\cdot)}[q]$ denotes the $q$th row of $\mathcal{D}_{(\cdot,\cdot)}$.
Then, the $(p,q)$th element of the fourth term can be written as
\[
\frac{1}{NT}\sum_{\ell=1}^K \sum_{j=1}^N \sum_{s=1}^Tw_\ell \bar{\varrho}_\ell^{(1)}\left(X_{j s}-\tilde{\lambda}_{\ell, j}^{\prime} \tilde{f}_s\right) \tilde{\lambda}_{\ell, j}\p \mathcal{Z}_{k,t,s}\p. 
\]
By Cauchy-Schwarz inequality, its norm is bounded by
\begin{align*}
\frac{1}{NT} & \sqrt{KN \sum_{s=1}^{T}\left\|\mathcal{Z}_{k,t,s}\right\|^{2}} \cdot \sqrt{ \sum_{\ell=1}^K \sum_{j=1}^{N} \sum_{s=1}^{T}\left[w_\ell \bar{\varrho}_\ell^{(1)}\left(X_{j s}-\tilde{\lambda}_{\ell, j}^{\prime} \tilde{f}_s\right)\right]^{2} \left\| \tilde{\lambda}_{\ell, j} \right\|^{2} }\\
& = \sqrt{\frac{KN}{T}} \sqrt{\frac{1}{N} \sum_{s=1}^{T}\left\|\mathcal{Z}_{k,t,s}\right\|^{2}} \cdot \sqrt{ \frac{1}{NT} \sum_{\ell=1}^K \sum_{j=1}^{N} \sum_{s=1}^{T}\left[w_\ell \bar{\varrho}_\ell^{(1)}\left(X_{j s}-\tilde{\lambda}_{\ell, j}^{\prime} \tilde{f}_s\right)\right]^{2} \left\| \tilde{\lambda}_{\ell, j} \right\|^{2} }
.
\end{align*}
We let $v_{i,s}=\lambda_{k,ip}^0 \sqrt{NT}\mathcal{D}_{(k-1)N+i,KN+s}[q]$, so that $\mathcal{Z}_{k,t,s}=\frac{1}{N} \sum\limits_{i=1}^N \tilde{\varrho}_{k,it}^{(2)}v_{i,s}$. Then,
\[
N^2\|{\mathcal{Z}_{k,t,s}}\|^2 = \| \sum\limits_{i=1}^N \tilde{\varrho}_{k,it}^{(2)}v_{i,s} \|^2 = \sum\limits_{i=1}^N \| \tilde{\varrho}_{k,it}^{(2)} v_{i,s} \|^2 + \sum\limits_{i=1}^N \sum\limits_{j \neq i}^N v_{i,s}\p v_{j,s} \tilde{\varrho}_{k,it}^{(2)} \tilde{\varrho}_{k,jt}^{(2)}.
\]
By taking expectations on both sides, we get
\[
N^2 \E\|{\mathcal{Z}_{k,t,s}}\|^2 = \sum\limits_{i=1}^N  \|v_{i,s}\|^2\cdot \E\left[\left(\tilde{\varrho}_{it}^{(2)}\right)^2\right].
\]
Note that Lemma \ref{lem:normality_5} implies that $\|v_{i,s} \|$ is uniformly bounded, and by Lemma \ref{lem:normality_1}, $\E\left[\left(\tilde{\varrho}_{it}^{(2)}\right)^2\right]=\E\left[\left(\varrho_{it}^{(2)}\right)^2\right]-\left(\bar{\varrho}_{it}^{(2)}\right)^2=\bar{O}(h^{-1})$. Therefore, $\E\|{\mathcal{Z}_{k,t,s}}\|^2=O(\frac{1}{Nh})$. Also, by Lemma \ref{lem:normality_1}, for fixed $\ell$,
\[
\left[w_\ell \bar{\varrho}_\ell^{(1)}\left(X_{j s}-\tilde{\lambda}_{\ell, j}^{\prime} \tilde{f}_s\right)\right]^{2} 
\lesssim 
[w_\ell \bar{\varrho}_{\ell,js}^{(1)}+w_\ell \bar{\varrho}_\ell^{(2)}(X_{js}-C^*)\cdot (\lambda_{\ell,j}^{0\p}f_s^0-\tilde{\lambda}_{\ell,j}\p\tilde{f}_s)]^2
\lesssim
O(h^{2m})+(\lambda_{\ell,j}^{0\p}f_s^0-\tilde{\lambda}_{\ell,j}\p\tilde{f}_s)^2,
\]
where $C^*$ is a value between $\lambda_{\ell,j}^{0\p}f_s^0$ and $\tilde{\lambda}_{\ell,j}\p\tilde{f}_s$. By summing both sides across $\ell,j$, and $s$, we get
\[
\sqrt{ \frac{1}{NT}\sum_{\ell=1}^K \sum_{j=1}^{N} \sum_{s=1}^{T}\left[w_\ell \bar{\varrho}_\ell^{(1)}\left(X_{j s}-\tilde{\lambda}_{\ell, j}^{\prime} \tilde{f}_s\right)\right]^{2} \left\| \tilde{\lambda}_{\ell, j} \right\|^{2} } \lesssim O(h^m)+d(\tilde{\th},\th_0)=O_p(L_{NT}^{-1}).
\]
by Lemma \ref{lem:normality_2} and Assumption 2. Therefore, the fourth term of (\ref{eq:lem13_1}) is $O_p(\frac{1}{N\sqrt{h}})$.

Next, it is easily shown that the fifth term of (\ref{eq:lem13_1}) is $O_p(L_{NT}^{-2}h^{-1})$ by using (\ref{eq:Rlam}) and the result that $\|\tilde{\Lam}_k-\Lam_k^0\| / \sqrt{N} =O_p(L_{NT}^{-1})$. 

For the sixth and seventh terms, the order can be obtained in a similar way to the third and fourth terms, and we only focus on the seventh term. We can write the $p,q$th element of the seventh term as
\[
-\frac{0.5}{\sqrt{NT}} \sum\limits_{s=1}^T \mathcal{Z}_{k,s,t} \mathcal{R}(\tilde{\th})_{KN+s}.
\]
Then, by Cauchy-Schwarz inequality, its norm is bounded by
\begin{align*}
\frac{0.5}{\sqrt{NT}} \sqrt{\sum\limits_{s=1}^T \|\mathcal{Z}_{k,s,t} \|^2 }
\cdot 
\sqrt{\sum\limits_{s=1}^T \left\| \mathcal{R}(\tilde{\th})_{KN+s}\right\|^2} & \leq
0.5 \frac{T}{\sqrt{N}} \sqrt{\frac{1}{T}\sum\limits_{s=1}^T \|\mathcal{Z}_{k,s,t} \|^2 }
\cdot \frac{1}{T} \sum\limits_{s=1}^T  \left\| \mathcal{R}(\tilde{\th})_{KN+s}\right\|\\
& =  O_p(L_{NT}) \cdot O_p(\frac{1}{\sqrt{Nh}}) \cdot O_p(L_{NT}^{-2})=O_p(L_{NT}^{-2}h^{-1/2}).
\end{align*}
The seventh term has the same order of $O_p(L_{NT}^{-2}h^{-1/2})$, which can be shown in the same way. Finally, by combining the orders of all seven terms and noting that the second term has the largest order $O_p(\frac{1}{Nh^2})$, we obtain the desired result.

As mentioned, the remaining results can be proved in a similar way. In particular, just as we utilized the uniform boundedness of $h \cdot \varrho_{k,it}^{(2)}(\cdot)$ in the proof of the second result, the proofs of the first and third results rely on the uniform boundedness of $\varrho_{k,it}^{(1)}(\cdot)$. Additionally, we use the boundedness of  $\| f_t^0 \|$, $\| \lambda_{k,i}^0 \|$ for each $t, k$ and $i$. 
\end{proof}

\noindent \textbf{Proof of Theorem 2:}
    We first show the asymptotic distribution of $\sqrt{N}\left(\tilde{f}_t-f_t^0 \right)$.
    For $\lambda_{k,i}^* \in \mathcal{A}$ in (\ref{eq:Taylorexp}), we can write
    \begin{equation}\label{eq:thm4_expand1}
    \varrho_k^{(1)}\left(X_{i t}-\lambda_{k, i}^{*\p} f_t^0\right) = \varrho_{k,it}^{(1)} + \varrho_k^{(2)}(X_{it}-\lambda_{k, i}^{**\p} f_t^0)f_t^{0\p}(\lambda_{k, i}^*-\lambda_{k, i}^0),
    \end{equation}
    and
    \begin{align} \label{eq:thm4_expand2}
    \begin{split}
    \varrho^{(2)}_k\left(X_{i t}-\lambda_{k,i}^{*\p} f_{t}^{0}\right) \lambda_{k,i}^{*}
    & =
    \varrho^{(2)}_{k,it}\lambda_{k,i}^0 +
    \varrho^{(2)}_k\left(X_{i t}-\lambda_{k,i}^{**\p} f_{t}^{0}\right) \left(\lambda_{k,i}^*-\lambda_{k,i}^0\right) \\
    & \quad + \varrho^{(3)}_k\left(X_{i t}-\lambda_{k,i}^{**\p} f_{t}^{0}\right)  \lambda_{k,i}^{**}f_t^{0\p}
    \left(\lambda_{k,i}^*-\lambda_{k,i}^0\right),
    \end{split}
    \end{align}
    where $ \lambda_{k,i}^{**}$ lies between $ \lambda_{k,i}^{*}$ and $ \lambda_{k,i}^{0}$.
    We substitute these results in (\ref{eq:Taylorexp}), multiply the weight $w_k$ on both sides, and sum up both sides across $i=1,\ldots,N$ and $k=1,\ldots K$. Then, we take expectation and put $\lambda_{k,i}=\tilde{\lambda}_{k,i}$, $f_t =\tilde{f}_t$. Using that $\| \lambda_{k,i}^*-\lambda_{k,i}^0 \| \leq \| \tilde{\lambda}_{k,i}-\lambda_{k,i}^0 \| $, and by Lemma \ref{lem:normality_1}, we can easily obtain
    \begin{align*}
         &\frac{1}{N}\sumk\sum\limits_{i=1}^N w_k \cdot \bar{\varrho}_k^{(1)}\left(X_{i t}-\tilde{\lambda}_{k,i}^{\prime} \tilde{f}_{t}\right) \tilde{\lambda}_{k,i} =  \frac{1}{N} \sumk \sum\limits_{i=1}^N w_k \bar{\varrho}_{k,i t}^{(1)} \lambda_{k,i}^0 
         +
         \frac{1}{N} \sumk \sum\limits_{i=1}^N w_k \bar{\varrho}_{k,i t}^{(1)} \left( \tilde{\lambda}_{k,i}-\lambda_{k,i}^0  \right) \\
         &\quad \quad -\frac{1}{N} \sumk \sum\limits_{i=1}^N w_k\bar{\varrho}_{k,i t}^{(2)} \lambda_{k,i}^0 f_t^{0\p} 
         \left( \tilde{\lambda}_{k,i}-\lambda_{k,i}^0  \right) 
         -
         \frac{1}{N}\sum\limits_{i=1}^N w_k \bar{\varrho}_{k,i t}^{(2)}  \tilde{\lambda}_{k,i} \tilde{\lambda}\p_{k,i} \left( \tilde{f}_t - f_t^0 \right) \\ 
         & \quad \quad + \sumk  O_p(\|\tilde{f}_t - f_t^0\|)\cdot O_p\left(\frac{\| \tilde{\Lam}_k-\Lam_k^0 \|}{\sqrt{N} } \right) + O_p(\| \tilde{f}_t - f_t^0\|^2) + \sumk O_p\left(\frac{\| \tilde{\Lam}_k-\Lam_k^0 \|^2}{N}\right).
    \end{align*}
    Using (24) and by Lemmas \ref{lem:normality_1}, \ref{lem:normality_2}, \ref{lem:normality_4}, along with Assumption 2, we have
    \begin{align*}
     & \Psi_t(\tilde{f}_t -f_t^0)  +o_p(\|  \tilde{f}_t - f_t^0 \| ) 
     = 
     -\frac{1}{N}\sumk\sum\limits_{i=1}^N w_k \cdot \bar{\varrho}_k^{(1)}\left(X_{i t}-\tilde{\lambda}_{k,i}^{\prime} \tilde{f}_{t}\right) \tilde{\lambda}_{k,i} + \frac{1}{N} \sumk \sum\limits_{i=1}^N w_k \bar{\varrho}_{k,i t}^{(1)} \lambda_{k,i}^0 \\ 
     & \quad +
     \frac{1}{N} \sumk \sum\limits_{i=1}^N w_k \bar{\varrho}_{k,i t}^{(1)} \left( \tilde{\lambda}_{k,i}-\lambda_{k,i}^0  \right) -\frac{1}{N} \sumk \sum\limits_{i=1}^N w_k\bar{\varrho}_{k,i t}^{(2)} \lambda_{k,i}^0 f_t^{0\p} 
         \left( \tilde{\lambda}_{k,i}-\lambda_{k,i}^0  \right) + O_p(L_{NT}^{-2}).
    \end{align*}
    By Lemmas \ref{lem:normality_1} and \ref{lem:normality_7}, 
    \begin{equation} \label{eq:Thm4_total}
       \Psi_t(\tilde{f}_t -f_t^0)  +o_p(\|  \tilde{f}_t - f_t^0 \| ) 
     = 
     -\frac{1}{N}\sumk\sum\limits_{i=1}^N w_k \cdot \bar{\varrho}_k^{(1)}\left(X_{i t}-\tilde{\lambda}_{k,i}^{\prime} \tilde{f}_{t}\right) \tilde{\lambda}_{k,i} + O(h^m) +  O_p\left(\frac{1}{Nh^2}\right). 
    \end{equation}
    Since $\frac{\partial}{\partial f_t} S_{NT}(\tilde{\th}) = \frac{1}{N}\sumk \sum\limits_{i=1}^N w_k \cdot \varrho^{(1)}_k(X_{it}-\tilde{\lambda}_{k,i}^{\prime} \tilde{f}_{t}) \tilde{\lambda}_{k,i} = 0$, 
    \begin{align*}
    & -\frac{1}{N}\sumk\sum\limits_{i=1}^N w_k \cdot \bar{\varrho}_k^{(1)}\left(X_{i t}-\tilde{\lambda}_{k,i}^{\prime} \tilde{f}_{t}\right) \tilde{\lambda}_{k,i} = \frac{1}{N}\sumk\sum\limits_{i=1}^N w_k \cdot \tilde{\varrho}_k^{(1)}\left(X_{i t}-\tilde{\lambda}_{k,i}^{\prime} \tilde{f}_{t}\right) \tilde{\lambda}_{k,i}, 
    \end{align*}
    and by expanding the RHS around $\lambda_{k,i}^{0\p}f_t^0$, we get
    \begin{align*}
        -\frac{1}{N}\sumk\sum\limits_{i=1}^N w_k \bar{\varrho}_k^{(1)}\left(X_{i t}-\tilde{\lambda}_{k,i}^{\prime} \tilde{f}_{t}\right) \tilde{\lambda}_{k,i} 
        & = \frac{1}{N}\sumk\sum\limits_{i=1}^N w_k \tilde{\varrho}_{k,it}^{(1)}\tilde{\lambda}_{k,i} 
        -\frac{1}{N}\sumk\sum\limits_{i=1}^N w_k \tilde{\varrho}_{k,it}^{(2)} \left(\tilde{\lambda}_{k,i}^{\prime} \tilde{f}_{t} - \lambda_{k,i}^{0\p}f_t^0 \right) \tilde{\lambda}_{k,i} \\
        &  \quad + 0.5\frac{1}{N}\sumk\sum\limits_{i=1}^N w_k \tilde{\varrho}_{k,it}^{(3)}(X_{it}-c_{k,it}^*)\left(\tilde{\lambda}_{k,i}^{\prime} \tilde{f}_{t} - \lambda_{k,i}^{0\p}f_t^0 \right)^2 \tilde{\lambda}_{k,i},
    \end{align*}
    where $c_{k,it}^*$ is a value between $\tilde{\lambda}_{k,i}^{\p}\tilde{f}_t$ and $\lambda_{k,i}^{0\p}f_t^0$. We now compute the order of each term of RHS in the above equation.
First, by Lemma \ref{lem:normality_7}, we have
\begin{equation} \label{eq:Thm4_first}
\frac{1}{N}\sumk\sum\limits_{i=1}^N w_k \tilde{\varrho}_{k,it}^{(1)}\tilde{\lambda}_{k,i} = \frac{1}{N}\sumk\sum\limits_{i=1}^N w_k \tilde{\varrho}_{k,it}^{(1)}\lambda_{k,i}^0 + O_p(\frac{1}{Nh}).
\end{equation}
Next,
\begin{align*}
    & \frac{1}{N}\sumk\sum\limits_{i=1}^N w_k \tilde{\varrho}_{k,it}^{(2)} \left(\tilde{\lambda}_{k,i}^{\prime} \tilde{f}_{t} - \lambda_{k,i}^{0\p}f_t^0 \right) \tilde{\lambda}_{k,i} \\
    & \quad =  \frac{1}{N}\sumk\sum\limits_{i=1}^N w_k \tilde{\varrho}_{k,it}^{(2)} \tilde{\lambda}_{k,i}\left( \tilde{\lambda}_{k,i}-\lambda_{k,i}^0 \right)\p \tilde{f}_t 
    + 
    \frac{1}{N}\sumk\sum\limits_{i=1}^N w_k \tilde{\varrho}_{k,it}^{(2)} \tilde{\lambda}_{k,i} \lambda_{k,i}^{0\p} \left( \tilde{f}_t-f_t^0 \right) \\
    & \quad = \frac{1}{N}\sumk\sum\limits_{i=1}^N w_k \tilde{\varrho}_{k,it}^{(2)} \left( \tilde{\lambda}_{k,i} -\lambda_{k,i}^0 \right) \left( \tilde{\lambda}_{k,i}-\lambda_{k,i}^0 \right)\p \tilde{f}_t 
    +
    \frac{1}{N}\sumk\sum\limits_{i=1}^N w_k \tilde{\varrho}_{k,it}^{(2)} \lambda_{k,i}^0 \left( \tilde{\lambda}_{k,i}-\lambda_{k,i}^0 \right)\p \tilde{f}_t\\
    & \quad +
    \frac{1}{N}\sumk\sum\limits_{i=1}^N w_k \tilde{\varrho}_{k,it}^{(2)} \left(\tilde{\lambda}_{k,i} - \lambda_{k,i}^0\right) \lambda_{k,i}^{0\p}(\tilde{f}_t-f_t^0) 
    + 
    \frac{1}{N}\sumk\sum\limits_{i=1}^N w_k \tilde{\varrho}_{k,it}^{(2)} \lambda_{k,i}^0 \lambda_{k,i}^{0\p}(\tilde{f}_t-f_t^0) \\
    & \quad = O_p\left(\frac{1}{L_{NT}^2 h}\right) + O_p\left(\frac{1}{Nh^2}\right) + O_p\left( \frac{\| \tilde{\Lam}_k -\Lam_k^0 \|} {\sqrt{N}h} \right) O_p\left( \| \tilde{f}_t -f_t^0 \| \right) +O_p\left(\frac{1}{\sqrt{Nh}}\right) O_p(\|\tilde{f}_t -f_t^0  \|)
\end{align*}
by Lemmas \ref{lem:normality_1}, \ref{lem:normality_2}, \ref{lem:normality_7}, and Lyapunov's CLT. Since $\frac{\| \tilde{\Lam}_k -\Lam_k^0\|}{\sqrt{N}}=O_p(L_{NT}^{-1})$ and $\sqrt{N}h \rightarrow \infty$, we get
\begin{equation}\label{eq:Thm4_second}
\frac{1}{N}\sumk\sum\limits_{i=1}^N w_k \tilde{\varrho}_{k,it}^{(2)} \left(\tilde{\lambda}_{k,i}^{\prime} \tilde{f}_{t} - \lambda_{k,i}^{0\p}f_t^0 \right) \tilde{\lambda}_{k,i} 
= 
O_p\left(\frac{1}{Nh^2}\right) + o_p(\|\tilde{f}_t-f_t^0 \|).
\end{equation}
Finally, by Lemmas \ref{lem:normality_1} and \ref{lem:normality_2},
\begin{align*}
    & \left\| \frac{1}{N}\sumk\sum\limits_{i=1}^N w_k \tilde{\varrho}_{k,it}^{(3)}(X_{it}-c_{k,it}^*)\left(\tilde{\lambda}_{k,i}^{\prime} \tilde{f}_{t} - \lambda_{k,i}^{0\p}f_t^0 \right)^2 \tilde{\lambda}_{k,i} \right\| \\
    & \quad \lesssim \frac{1}{N} \sumk \sum\limits_{i=1}^N 
    w_k | \tilde{\varrho}_{k,it}^{(3)}(X_{it}-c_{k,it}^*) | \cdot \left\| \tilde{\lambda}_{k,i} - \lambda_{k,i}^0 \right\|^2 
    + \left\| \tilde{f}_t -f_t^0 \right\|^2 \frac{1}{N} \sumk \sum\limits_{i=1}^N 
    w_k| \tilde{\varrho}_{k,it}^{(3)}(X_{it}-c_{k,it}^*) | \\
    & \quad  = \sumk  O_p\left( h^{-2 }\frac{ \| \tilde{\Lam}_{k} - \Lam_{k}^0 \|^2}{N} \right)
    + O_p(\| \tilde{f}_{t} - f_t^0 \|) \cdot O_p\left(\frac{h^{-2}}{\sqrt{N}h}\right),
\end{align*}
which simplifies to 
\begin{equation} \label{eq:Thm4_third}
O_p\left(\frac{1}{Nh^2}\right)+o_p(\| \tilde{f}_{t} - f_t^0 \|)
\end{equation}
by Lemma \ref{lem:normality_2}, given that $\sqrt{N}h^3 \rightarrow \infty$.
Then, by (\ref{eq:Thm4_total}), (\ref{eq:Thm4_first}), (\ref{eq:Thm4_second}), and (\ref{eq:Thm4_third}), we have
\begin{equation} \label{eq:Thm4_final}
  \Psi_t \cdot (\tilde{f}_t -f_t^0)  +o_p(\|  \tilde{f}_t - f_t^0 \| ) 
     = 
    \frac{1}{N}\sumk\sum\limits_{i=1}^N w_k \tilde{\varrho}_{k,it}^{(1)}\lambda_{k,i}^0+ O(h^m) +  O_p\left(\frac{1}{Nh^2}\right). 
\end{equation}
Note that by Lemma \ref{lem:normality_1},  
\begin{align*}
\operatorname{Var}\left( \sumk w_k \tilde{\varrho}_{k,it}^{(1)}\lambda_{k,i}^0  \right) &= \sumk\sum\limits_{k\p=1}^K w_k w_{k\p} \operatorname{Cov}\left(\tilde{\varrho}_{k,it}^{(1)}, \tilde{\varrho}_{k\p,it}^{(1)}\right) \lambda_{k,i}^0 \lambda_{k\p,i}^{0\p} + \bar{O}(h) \\
& = \sumk\sum\limits_{k\p=1}^K w_k w_{k\p} \operatorname{min}(\tau_k, \tau_{k\p})\cdot (1-\operatorname{max}(\tau_k, \tau_{k\p})) \lambda_{k,i}^0 \lambda_{k\p,i}^{0\p} + \bar{O}(h),
\end{align*}
and thus, it can be shown by Lyapunov's CLT that
\[
 \frac{1}{\sqrt{N}}  \sum\limits_{i=1}^N \sumk  w_k \tilde{\varrho}_{k,it}^{(1)} \lambda_{k,i}^0 \xrightarrow{d} \mathcal{N}\left(0,  \sumk\sum\limits_{k\p=1}^K w_k w_{k\p} \operatorname{min}(\tau_k, \tau_{k\p}) (1-\operatorname{max}(\tau_k, \tau_{k\p})) \Sigma_{k,k\p} \right).
\]
Also, since $\sqrt{N}h^m \rightarrow 0$ and $\frac{1}{\sqrt{N}h^2} \rightarrow 0$, it follows from (\ref{eq:Thm4_final}) that
\[
\sqrt{N} \left( \tilde{f}_t-f_t^0 \right) \xrightarrow{d} \mathcal{N}\left(0, \Psi_t^{-1} \sumk\sum\limits_{k\p=1}^K \left\{ w_k w_{k\p} \operatorname{min}(\tau_k, \tau_{k\p}) (1-\operatorname{max}(\tau_k, \tau_{k\p})) \Sigma_{k,k\p} \right\} \Psi_t^{-1} \right).
\]

Next, we derive the asymptotic distribution of $\sqrt{T} \left( \tilde{\lambda}_{k,i}-\lambda_{k,i}^0 \right)$. As the proof follows a similar approach to the previous result, we reduce some details for brevity. For fixed $k$, by expanding $\varrho_k^{(1)}\left(X_{i t}-\lambda_{k,i}^{\prime} f_{t}\right) f_{t}$ on $\lambda_{k,i}^0$ and $f_t^0$, similarly to (\ref{eq:Taylorexp}), we get
\begin{align*}
& \varrho_k^{(1)}\left(X_{i t}-\lambda_{k,i}^{\prime} f_{t}\right) f_t  \\
    & \quad = \varrho_{k,i t}^{(1)} f_t^0 + \left\{ 
    \varrho_k^{(1)}\left(X_{i t}-\lambda_{k,i}^{0\p} f_{t}^{*}\right)-\varrho^{(2)}_k\left(X_{i t}-\lambda_{k,i}^{0\p} f_{t}^{*}\right)  f_{t}^{*}  \lambda_{k,i}^{0\p}
    \right\} \left(f_t-f_t^0\right) \\
& \quad\quad -\varrho_{k,i t}^{(2)} f_t f_t^{\prime}\left(\lambda_{k,i}-\lambda_{k,i}^0\right) 
+
\varrho_k^{(3)}\left(X_{i t}-\lambda_{k,i}^{0^\prime} f_{t}^{*}\right) f_t f_t^{\prime} \cdot\left(\lambda_{k,i}-\lambda_{k,i}^0\right) \lambda_{k,i}^{0^\prime}\left(f_t-f_t^0\right) \\
& \quad\quad + 0.5 \varrho_k^{(3)}\left(X_{i t}-\lambda_{k,i}^{\prime} f_{t}^{*}\right) \lambda_{k,i}\left[\left(f_t-f_t^0\right)^{\prime} \lambda_{k,i}\right]^{2}.
\end{align*}
Then, by expanding $ \varrho_k^{(1)}\left(X_{i t}-\lambda_{k,i}^{0\p} f_{t}^{*}\right)$ and $ \varrho^{(2)}_k\left(X_{i t}-\lambda_{k,i}^{0\p} f_{t}^{*}\right)  f_{t}^{*} $ around $f_t^0$ similarly to (\ref{eq:thm4_expand1}) and (\ref{eq:thm4_expand2}), we can get
\begin{align*}
    &\frac{1}{T}\sum\limits_{t=1}^T \bar{\varrho}_k^{(1)}\left(X_{i t}-\tilde{\lambda}_{k,i}^{\prime} \tilde{f}_{t}\right) \tilde{f}_t
    =  \frac{1}{T}  \sum\limits_{t=1}^T \bar{\varrho}_{k,i t}^{(1)} f_t^0 
    +
    \frac{1}{T}  \sum\limits_{t=1}^T \bar{\varrho}_{k,i t}^{(1)} \left( \tilde{f}_t-f_t^0  \right) \\
     &\quad \quad -\frac{1}{T}  \sum\limits_{t=1}^T \bar{\varrho}_{k,i t}^{(2)} f_t^0 \lambda_{k,i}^{0\p} 
    \left( \tilde{f}_t-f_t^0  \right) 
         -
    \frac{1}{T}\sum\limits_{t=1}^T \bar{\varrho}_{k,i t}^{(2)}  \tilde{f}_t \tilde{f}_t\p \left( \tilde{\lambda}_{k,i} - \lambda_{k,i}^0 \right) \\ 
    & \quad \quad +   O_p(\| \tilde{\lambda}_{k,i} - \lambda_{k,i}^0 \|)\cdot O_p\left(\frac{\| \tilde{F}-F^0 \|}{\sqrt{T} } \right) + O_p(\|  \tilde{\lambda}_{k,i} - \lambda_{k,i}^0 \|^2) +  O_p\left(\frac{\| \tilde{F}-F^0 \|^2}{T}\right),
\end{align*}
which then follows from Lemmas \ref{lem:normality_1} and \ref{lem:normality_7} that
\begin{equation}\label{eq:Thm4_final2}
       \Phi_{k,i}(\tilde{\lambda}_{k,i} - \lambda_{k,i}^0)  +o_p(\|  \tilde{\lambda}_{k,i} - \lambda_{k,i}^0 \| ) 
     = 
     -\frac{1}{T}\sum\limits_{t=1}^T \bar{\varrho}_k^{(1)}\left(X_{i t}-\tilde{\lambda}_{k,i}^{\prime} \tilde{f}_{t}\right) \tilde{f}_t + O(h^m) +  O_p\left(\frac{1}{Th^2}\right). 
\end{equation}
Since $\frac{\partial}{\partial \lambda_{k,i}} S_{NT}(\tilde{\th}) = \frac{1}{T} \sum\limits_{t=1}^T \varrho^{(1)}_k(X_{it}-\tilde{\lambda}_{k,i}^{\prime} \tilde{f}_{t}) \tilde{f}_t = 0$, 
    \begin{align*}
    & -\frac{1}{T}\sum\limits_{t=1}^T \bar{\varrho}_k^{(1)}\left(X_{i t}-\tilde{\lambda}_{k,i}^{\prime} \tilde{f}_{t}\right) \tilde{f}_t = \frac{1}{T} \sum\limits_{t=1}^T \tilde{\varrho}_k^{(1)}\left(X_{i t}-\tilde{\lambda}_{k,i}^{\prime} \tilde{f}_{t}\right) \tilde{f}_t, 
    \end{align*}
    and by expanding the RHS around $\lambda_{k,i}^{0\p}f_t^0$, we get
    \begin{align*}
        -\frac{1}{T}\sum\limits_{t=1}^T \bar{\varrho}_k^{(1)}\left(X_{i t}-\tilde{\lambda}_{k,i}^{\prime} \tilde{f}_{t}\right) \tilde{f}_t
        & = \frac{1}{T}\sum\limits_{t=1}^T \tilde{\varrho}_{k,it}^{(1)}\tilde{f}_t
        -\frac{1}{T}\sum\limits_{t=1}^T \tilde{\varrho}_{k,it}^{(2)} \left(\tilde{\lambda}_{k,i}^{\prime} \tilde{f}_{t} - \lambda_{k,i}^{0\p}f_t^0 \right) \tilde{f}_t \\
        &  \quad + 0.5\frac{1}{T}\sum\limits_{t=1}^T \tilde{\varrho}_{k,it}^{(3)}(X_{it}-d_{k,it}^*)\left(\tilde{\lambda}_{k,i}^{\prime} \tilde{f}_{t} - \lambda_{k,i}^{0\p}f_t^0 \right)^2 \tilde{f}_t,
    \end{align*}
    where $d_{k,it}^*$ is a value between $\tilde{\lambda}_{k,i}^{\p}\tilde{f}_t$ and $\lambda_{k,i}^{0\p}f_t^0$.
In a similar way to deriving the asymptotic distribution of factors, we can show that the first term of the above equation satisfies
\[
\frac{1}{T}\sum\limits_{t=1}^T \tilde{\varrho}_{k,it}^{(1)}\tilde{f}_t = \frac{1}{T}\sum\limits_{t=1}^T  \tilde{\varrho}_{k,it}^{(1)}f_t^0 + O_p(\frac{1}{Th}),
\]
while the second and third terms can be shown to be $O_p(\frac{1}{Th^2})+o_p(\| \tilde{\lambda}_{k,i} - \lambda_{k,i}^0 \|)$. Therefore, from (\ref{eq:Thm4_final2}), we have
\[
\Phi_{k,i}(\tilde{\lambda}_{k,i} - \lambda_{k,i}^0)  +o_p(\|  \tilde{\lambda}_{k,i} - \lambda_{k,i}^0 \| ) 
     = 
     \frac{1}{T}\sum\limits_{t=1}^T \tilde{\varrho}_{k,it}^{(1)}f_t^0 + O(h^m) +  O_p\left(\frac{1}{Th^2}\right). 
\]
Since it follows from Lemma \ref{lem:normality_1} that 
\[
\operatorname{Var}\left( \tilde{\varrho}_{k,it}^{(1)}f_t^0 \right) = \tau_k(1-\tau_k)f_t^0 f_t^{0\p} + \bar{O}(h),
\]
we can finally derive
\[
\sqrt{T}\left( \tilde{\lambda}_{k,i}-\lambda_{k,i}^0 \right) \xrightarrow{d} \mathcal{N}(0,  \tau_k(1-\tau_k) \Phi_{k,i}^{-2}).
\]

\section{Proof of Theorems \ref{thm:factornumber1} and \ref{thm:factornumber2}} 
Now, we prove Theorems \ref{thm:factornumber1} and \ref{thm:factornumber2}. Suppose that Assumptions 1 and 3 hold. Let $s>r$. For $\th^s \in \Theta^s$, let $d(\th^s,\th_0)=\sqrt{ \frac{1}{NT} \sumk \norm{ \Lambda_k F\p-\Lambda^0_k F^{0\p} }^2}$. Write $\Lambda_k^s = \left[ \Lambda_k^{s,r} \, \Lambda_k^{s,-r} \right]$, where $\Lambda_k^{s,r} \in \R^{N \times r}$ and $\Lambda_k^{s,-r} \in \R^{N \times (s-r)}$ for $k=1,\ldots,K$. Also, define $\Theta^s(\delta)=\{\th^s \in \Theta^s : d(\th^s, \th_0)\leq \delta\}$.

\begin{lemma} \label{lem:facnum_1} Let ~$r < s <\infty$. Then, for any $\th^s \in \Theta^s(\delta)$ and sufficiently small $\delta$, it holds that
\[
\frac{\norm{F^{s,r}-F^0S}}{\sqrt{T}} \lesssim \delta \text{\, and \,} \frac{\norm{\Lambda_k^{s,r}-\Lambda^0S}}{\sqrt{N}} \lesssim \delta, ~\, \frac{\norm{\Lambda_k^{s,-r}} }{\sqrt{N}} \lesssim \delta, \text{\, for \,} k=1,\ldots, K,
\]
where $S=\text{sgn}((F^{s,r})\p F^0/T)$. 
\end{lemma}
\begin{proof}
Define $\Lambda_{k}^{0*}=\left[ \Lambda_k^0 \, \, \mathbf{0} \right] \in \R^{N \times s}, ~F^{0*}=\left[ F^0 \, V \right] \in \R^{T \times s}$ for $V \in \R^{T \times (s-r)}$ s.t. $\frac{(F^{0*})\p F^{0*}}{T}=\I_s$. For $\th^s \in \Theta^s(\delta)$, define
$ 
U^*= \begin{bmatrix}
S & 0 \\[-5pt]
0 & 0
\end{bmatrix} \in \R^{s \times s}
$. Then, using $\Lambda_k^{0*}(F^{0*})\p=\Lambda_k^0 F^{0\p}$ and $\Lambda_k^{0*} U^*F\p=\Lambda_0S (F^{s,r})\p$, 
\[
\begin{aligned}
    \frac{\norm{\Lambda_k-\Lambda_k^{0*} U^*}}{\sqrt{N}} &= \frac{\norm{(\Lambda_k-\Lambda_k^{0*} U^*)F\p}}{\sqrt{NT}}\\ 
    &=\frac{\norm{\Lambda_k F\p -\Lambda_k^0 (F^{0})\p+\Lambda_k^{0} (F^{0})\p- \Lambda_k^{0*} U^* F\p}}{\sqrt{NT}} \\
    & \leq \frac{\norm{\Lambda_k F\p -\Lambda_k^0 F^{0\p}}}{\sqrt{NT}}+\frac{\norm{\Lambda_k^{0} F^{0\p}- \Lambda_k^{0} S (F^{s,r})\p}}{\sqrt{NT}}\\
    & \leq d(\th^s,\th_0)+ M_1 \cdot \frac{\norm{F^{s,r}-F^0S}}{\sqrt{T}}.
\end{aligned}
\]
Also, since $\Lambda_k-\Lambda_k^{0*}U^*=[\Lambda_k^{s,r} \, \Lambda_k^{s,-r}]-[\Lambda_k^0 S \,\, \mathbf{0}]=[\Lambda_k^{s,r}-\Lambda_k^0 S \, \quad \Lambda_k^{s,-r}]$, 
\[
\frac{\norm{\Lambda_k^{s,r}-\Lambda_k^0 S}}{\sqrt{N}}\leq  \frac{\norm{\Lambda_k-\Lambda_k^{0*} U^*}}{\sqrt{N}}~~ \mbox{and} ~~\frac{\norm{\Lambda_k^{s,-r}}}{\sqrt{N}} \leq  \frac{\norm{\Lambda_k-\Lambda_k^{0*} U^*}}{\sqrt{N}},
\]
which hold for all $k=1,\ldots,K$. Therefore, it suffices to show that $ \frac{\norm{F^{s,r}-F^0S}}{\sqrt{T}} \lesssim d(\th^s,\th_0)$. Let $R_T=\frac{F^{s\p}F^0}{T}$. Then,
\begin{equation} \label{eq:D_NT 1}
\begin{aligned}
    \norm{\frac{\Lamsp\Lams}{N}-R_T(\frac{\Lamzp\Lamz}{N})R_T\p} 
    &=\norm{\frac{\Lamsp\Lams}{N}-\frac{\Lamsp \Lamz R_T\p}{N}+\frac{\Lamsp \Lamz R_T\p}{N}-R_T(\frac{\Lamzp\Lamz}{N})R_T\p} \\
    & = \frac{1}{N}\norm{\Lamsp(\Lams-\Lamz R_T\p)+(\Lamsp-R_T \Lamzp)\Lamz R_T\p}\\
    & \leq \frac{1}{\sqrt{N}}\norm{\Lams-\Lamz R_T\p}\cdot\bigg(\frac{\norm{\Lamsp}}{\sqrt{N}}+\frac{\norm{\Lamz R_T\p}}{\sqrt{N}}\bigg).
\end{aligned}
\end{equation}
Note that from
\[
    \frac{1}{\sqrt{NT}}\norm{(\Lams F^{s\p}-\Lamz F^{0\p})\cdot \frac{F^{s}F^{s\p}}{T}} \leq \frac{1}{\sqrt{NT}}\norm{(\Lams F^{s\p}-\Lamz F^{0\p})}\cdot \norm{\frac{F^{s}F^{s\p}}{T}} \leq d(\th^s,\th_0)\cdot \sqrt{s},
\]
we have $\frac{1}{\sqrt{N}}\norm{\Lams-\Lamz R_T\p} \lesssim d(\th^s,\th_0)$. Therefore, by (\ref{eq:D_NT 1}), 
\begin{equation} \label{eq:D_NT2}
\norm{\frac{\Lamsp\Lams}{N}-R_T(\frac{\Lamzp\Lamz}{N})R_T\p} \lesssim d(\th,\th_0).
\end{equation}
Let $\rho_i(A)$ be the $i$-th largest eigenvalue of $A$ and $\sigma_{i}^k(N)=\rho_i(\frac{\Lam_k^{s\p}\Lam_k^s}{N})$. Then,
\begin{equation} \label{eq:eigenvalues}
\Big|\rho_{r+j}\Big(\frac{\Lamsp \Lams}{N}\Big)-\rho_{r+j}\Big(\, R_T\p\Big(\frac{\Lamzp \Lamz}{N}\Big)R_T\,\Big)\Big|=\big|\sigma^k_{r+j}(N)-0\big| \leq \norm{\frac{\Lamsp\Lams}{N}-R_T\Big(\frac{\Lamzp\Lamz}{N}\Big)R_T} \lesssim d(\th^s,\th_0).
\end{equation}
Now, define $R_T^r=\frac{(F^{k,r})\p F^0}{T}$ and $R_T^{-r}=\frac{(F^{k,-r})\p F^0}{T}$. Note that $R_T=[\,R_T^r \, \, R_T^{-r}\,]$.
From (\ref{eq:D_NT2}),
\[
\norm{
\begin{bmatrix}
\frac{(\Lamsr)\p\Lamsr}{N} & \frac{(\Lamsr)\p\Lamsnr}{N} \\
\frac{(\Lamsnr)\p\Lamsr}{N} & \frac{(\Lamsnr)\p\Lamsnr}{N}
\end{bmatrix}-
\begin{bmatrix}
R_T^r(\frac{\Lamzp \Lamz}{N}) R_T^{r\p} & R_T^r(\frac{\Lamzp \Lamz}{N}) (R_T^{-r})\p \\
R_T^{-r}(\frac{\Lamzp \Lamz}{N}) R_T^{r\p} & R_T^{-r}(\frac{\Lamzp \Lamz}{N}) (R_T^{-r})\p
\end{bmatrix}
} \lesssim d(\th^s,\th_0).
\]
Therefore, $\norm{\frac{(\Lamsnr)\p\Lamsnr}{N}-R_T^{-r}(\frac{\Lamzp \Lamz}{N})(R_T^{-r})\p}\lesssim d(\th^s,\th_0)$. Also, since
\[
\norm{\frac{(\Lamsnr)\p\Lamsnr}{N}}=\sqrt{\{\sigma_{r+1}^k(N)\}^2+\cdots+\{\sigma_{s}^k(N)\}^2} \lesssim d(\th^s,\th_0), 
\]
it holds that 
\[
\{\sigma_{r+1}^k(N)\}^2 \cdot \norm{R_T^{-r}}^2 \lesssim \norm{R_T^{-r}(\frac{\Lamzp \Lamz}{N})(R_T^{-r})\p} \leq d(\th,\th_0)+\norm{\frac{(\Lamsnr)\p\Lamsnr}{N}} \lesssim d(\th^s,\th_0).
\]
Therefore,  $\norm{R_T^{-r}}^2 \lesssim d(\th^s,\th_0)$.

Next, using that
\[
\begin{aligned}
I=\frac{F^{0\p}F^0}{T}&=R_T\p R_T+\frac{F^{0\p}F^0}{T}-R_T\p \frac{F^{s\p}F^s}{T} R_T\\
&=R_T^{r\p} R_T^r+(R_T^{-r})\p R_T^{-r}+\frac{F^{0\p}}{\sqrt{T}}\cdot \frac{F^0-F^s R_T}{\sqrt{T}},
\end{aligned}
\]
for $k^*$ used in the normalization condition, we have
{\small
\[
\begin{aligned}
    \frac{\Lambda_{k^*}^{0\p} \Lambda_{k^*}^0 }{N}
    &=R_T\p(\frac{\Lambda_{k^*}^{s\p} \Lambda_{k^*}^s}{N})R_T+\frac{\Lambda_{k^*}^{0\p} \Lambda_{k^*}^0}{N}-R_T\p(\frac{\Lambda_{k^*}^{s\p} \Lambda_{k^*}^s}{N})R_T
    \\ 
    &= R_T^{r^{\prime}}\cdot \operatorname{diag}(\sigma_{1}^{k^*}(N), \ldots, \sigma_{r}^{k^*}(N)) \cdot (R_T^{r\p})^{-1} + R_T^{r^{\prime}}\cdot \operatorname{diag}(\sigma_{1}^{k^*}(N), \ldots, \sigma_{r}^{k^*}(N))\cdot (R_T^{r\p})^{-1}(R_T^{r\p}R_T^r-\I) \\
    & \quad +R_T^{-r^{\prime}}\cdot \operatorname{diag}(\sigma_{r+1}^{k^*}(N), \ldots, \sigma_{s}^{k^*}(N))\cdot R_T^{-r} + \frac{\Lambda_{k^*}^0}{\sqrt{N}}\cdot\frac{\Lambda_{k^*}^0-\Lambda_{k^*}^s R_T}{\sqrt{N}} -\frac{\Lambda_{k^*}^0-\Lambda_{k^*}^s R_T}{\sqrt{N}}\cdot\frac{\Lambda_{k^*}^s R_T}{\sqrt{N}}.
\end{aligned}
\]
}
Then, 
\[
\left(\frac{\Lam_{k^*}^{0\p}\Lam_{k^*}^0}{N}+D_{NT}\right)R_T^{r\p}=R_T^{r\p}\cdot \operatorname{diag}(\sigma_{1}^{k^*}(N),\ldots,\sigma_{r}^{k^*}(N)),
\]
where
\begin{equation} \label{eq:D_NT3}
\begin{aligned}
D_{NT}&=R_T^{r^{\prime}}\cdot \operatorname{diag}(\sigma_{1}^{k^*}(N), \ldots, \sigma_{r}^{k^*}(N))\cdot (R_T^{r\p})^{-1}(\I-R_T^{r\p}R_T^r)-R_T^{-r^{\prime}}\cdot \operatorname{diag}(\sigma_{r+1}^{k^*}(N), \ldots, \sigma_{s}^{k^*}(N)) \cdot R_T^{-r}\\
& \quad \quad - \frac{\Lambda_{k^*}^{0\p}}{\sqrt{N}}\cdot\frac{\Lambda_{k^*}^0-\Lambda_{k^*}^s R_T}{\sqrt{N}} -\frac{(\Lambda_{k^*}^{0\p}-\Lambda_{k^*}^s R_T)\p}{\sqrt{N}}\cdot\frac{\Lambda_{k^*}^s R_T}{\sqrt{N}}.
\end{aligned}
\end{equation}
Note that 
\begin{equation} \label{eq:D_NT3(1)}
\begin{aligned}
    \frac{1}{\sqrt{N}}\norm{\Lam_{k^*}^0-\Lam_{k^*}^s R_T} &=\frac{1}{\sqrt{NT}}\norm{(\Lam_{k^*}^0F^{0\p}-\Lam_{k^*}^sF^{s\p})P_{F^0}} \leq d(\th^s,\th_0)\cdot \sqrt{r}.
\end{aligned}
\end{equation}
Also, for $M_{F^0}=I-P_{F^0}$ and using that $M_{F^0}^2=M_{F^0}$,
\[
\begin{aligned}
    \frac{1}{\sqrt{NT}}\norm{(\Lam_{k^*}^sF^{s\p}-\Lam_{k^*}^0F^{0\p})M_{F^0}}  &= \frac{1}{\sqrt{NT}}\norm{\Lambda_{k^*}^0 F^{0\p}M_{F^s}}=\sqrt{Tr\left(\frac{\Lambda_{k^*}^{0\p}\Lambda_{k^*}^0}{N}\cdot \frac{F^{0\p}M_{F^s}F^0}{T}\right)} \\
    &\geq \sqrt{\rho_{\min}\left(\frac{\Lambda_{k^*}^{0\p}\Lambda_{k^*}^0}{N}\right)\cdot \frac{\norm{M_{F^s} F^0}^2}{T}} = \nu_r(N)\cdot \frac{\norm{M_{F^s} F^0}}{\sqrt{T}},
\end{aligned}
\]
and
\[
\frac{1}{\sqrt{NT}}\norm{(\Lam_{k^*}^sF^{s\p}-\Lam_{k^*}^0F^{0\p})M_{F^0}} \leq d(\th^s,\th_0)\cdot 2\sqrt{r}.
\]
Then, we obtain
\[
    \frac{\norm{F^0-F^s R_T}}{\sqrt{T}} = \frac{\norm{M_{F^s} F^0}}{\sqrt{T}}  \lesssim d(\th^s,\th_0),
\]
and 
\begin{equation} \label{eq:D_NT3(2)}
    \norm{\I-R_T^{r\p}R_T^r} = \norm{(R_T^{-r})\p R_T^{-r}+\frac{F^0}{\sqrt{T}}\cdot \frac{F^0-F^s R_T}{\sqrt{T}}}\lesssim d(\th^s,\th_0).
\end{equation}
From (\ref{eq:D_NT3}), (\ref{eq:D_NT3(1)}), and (\ref{eq:D_NT3(2)}), $\norm{D_{NT}} \lesssim d(\th^s, \th_0)$. Then, by Weyl's inequality,
\begin{equation} \label{eq:eigen for <r}
    |\sigma_{j}^{k^*}(N)-\nu_j(N)| \leq \norm{D_{NT}} \lesssim d(\th^s,\th_0), \quad j=1,\ldots,r
\end{equation}
and by the perturbation theory for eigenvectors, for $V_T=\operatorname{diag}\big((R^r_{T, 1} R^{r\p}_{T, 1})^{-1 / 2}, \ldots,(R^r_{T, r} R^{r\p}_{T, r})^{-1 / 2}\big)$ where $R_{T, j}^{r\p}$ denotes the $j$th column of $R_T^{r\p}$, we have
\begin{equation} \label{eq:rvs-i}
\norm{R_T^{r\p} V_T \mathrm{~S}-\mathbb{I}_r}=\norm{R_T^{r\p} V_T-\mathrm{S}} \lesssim d(\theta^s, \theta_0).
\end{equation}
Also, it holds that 
\begin{equation}\label{eq:v-r}
\norm{V_T-\I_r} \lesssim \norm{R_T^r R_T^{r\p}-\I_r} = \frac{\norm{F^{s,r}M_{F^0}F^{s,r}}}{T} \leq \frac{\norm{M_{F^0}F^{s,r}}^2}{T} \lesssim d^2(\th^s,\th_0).
\end{equation}
The last inequality holds since 
\[
\sigma_r^{k^*}(N)\cdot \frac{\norm{M_{F^0}F^{s,r}}}{\sqrt{T}} \lesssim d(\th^s,\th_0), 
\]
and 
\[
\big|\sigma_r^{k^*}(N)-\nu_r(N)\big| \lesssim d(\th^s,\th_0) \text{ which implies } \sigma_r^{k^*}(N)>c \text{\, for some constant } c>0.
\]
By (\ref{eq:rvs-i}) and (\ref{eq:v-r}), 
\begin{equation} \label{eq:rt-s}
    \norm{R_T^{r\p}-S} \leq \norm{R_T^{r\p} V_T-S}+\norm{R_T^{r\p}V_T-R_T^{r}}\lesssim d(\th^s,\th_0). 
\end{equation}
On the other hand,
\[
\begin{aligned}
    \frac{\norm{F^{s,r}-F^0 S}}{\sqrt{T}} &= \frac{1}{\sqrt{T}}\norm{F^{s,r}S-F^0(\frac{F^{0\p}F^{s,r}S}{T})+F^0(\frac{F^{0\p}F^{s,r}S}{T})-F^0}\\
    & \leq \frac{1}{\sqrt{T}}\norm{F^{s,r}- \frac{F^0F^{0\p}}{T}F^{s,r}}+\norm{R_T\p-S} = \frac{1}{\sqrt{T}}\norm{M_{F^0}F^{s,r}}+\norm{R_T\p-S}.
\end{aligned}
\]
It follows from the last inequality of (\ref{eq:v-r}) and (\ref{eq:rt-s}) that $\frac{\norm{F^{s,r}-F^0 S}}{\sqrt{T}} \leq d(\th^s,\th_0)$ and this concludes the proof.
\end{proof}
\vspace{5mm}

\begin{lemma} \label{lem:facnum_2}
Let $r<s<\infty$. Then for sufficiently small $\delta$, it holds that
\[
\E[\underset{\th^s \in \Theta^s(\delta)}{\sup}|\W_{NT}(\th^s)|] \lesssim \frac{\delta}{L_{NT}}.
\]
\end{lemma}
\begin{proof}
As in Lemma \ref{lem:consist_3}, for $\th_a,~\th_b \in \Theta^s$, $\norm{\sqrt{NT|\W_{NT}(\th_a)-\W_{NT}(\th_b)|}} \lesssim d(\th_a, \th_b)$.
Define $\Theta^S(\delta)=\{\th^s \in \Theta^s:d(\th^s,\th_0)\leq d(\th^s,\th_0))\}$. Since $\W_{NT}(\th)$ is separable, 
\[
\begin{aligned}
    \E\bigg[ \underset{\th^s  \in \Theta^s(\delta)}{\sup}|\W_{NT}(\th^s)|\bigg] &\leq \norm{\underset{\th^s \in \Theta^s(\delta)}{\sup} \sqrt{NT}|\W_{NT}(\th)|}_{\psi_2}\\ 
    &\leq \norm{\underset{\th_a,\th_b \in \Theta^s(\delta)}{\sup} \sqrt{NT}|\W_{NT}(\th_a)-\W_{NT}(\th_b)|}_{\psi_2} \\
    & \lesssim \int_0^{2\delta}\psi^{-1}(D(\epsilon,d,\Theta^r(\delta))) d\epsilon \\
    & \lesssim \int_0^{2\delta} \sqrt{\log D(\frac{\epsilon}{2},d,\Theta^r(\delta)}) d\epsilon = 2 \int_0^{\delta} \sqrt{\log D(\epsilon,d,\Theta^r(\delta)}) d\epsilon. 
\end{aligned}
\]
Therefore, it suffices to show that $\int_0^{\delta} \sqrt{\log D(\epsilon,d,\Theta^r(\delta)}) d\epsilon = O(\sqrt{KN+T} \cdot \delta)$.\\
For $U \in \mathcal{S}=\{U: U=\operatorname{diag}(u_1,\ldots, u_r), u_i \in \{-1,1\}\}$,
define \[
\Theta^s(\delta;U)=\bigg\{\th \in \Theta^s:\frac{\norm{F^{s,r}-F^0U}}{\sqrt{T}}+\sumk\bigg\{\frac{\norm{\Lamsr-\Lamz U}}{\sqrt{N}}+ \frac{\norm{\Lamsnr}}{\sqrt{N}}
\bigg\} \leq K_3 \delta \bigg\}.
\]
By Lemma \ref{lem:facnum_1}, if $d(\th^s,\th_0) \leq \delta$, then $\th^s \in \Theta^s(\delta;U)$ for some $U\in\mathcal{S}$. Note that for $\th^s,\tilde{\th}^s \in \Theta^s(\delta)$,
\[
\begin{aligned}
    \frac{1}{\sqrt{NT}}\norm{\Lam_k^s F^{s\p}-\tilde{\Lam}^s_k\tilde{F}^{s\p}}&=\frac{1}{\sqrt{NT}} \norm{\Lam_k^{s,r}(F^{s,r})\p+\Lam_k^{s,-r}(F^{s,-r})\p-\tilde{\Lam}_k^{s,r}(\tilde{F}^{s,r})\p-\tilde{\Lam}_k^{s,-r}(\tilde{F}^{s,-r})\p} \\
    & \leq \frac{1}{\sqrt{NT}}\norm{\Lam_k^{s,r}(F^{s,r})\p-\tilde{\Lam}_k^{s,r}(\tilde{F}^{s,r})\p}+\frac{1}{\sqrt{NT}}\norm{\Lam_k^{s,-r}(F^{s,-r})\p-\tilde{\Lam}_k^{s,-r}(\tilde{F}^{s,-r})\p}\\
    & \leq \frac{\norm{F^{k,r}}}{\sqrt{T}}\cdot \frac{\norm{\Lam_k^{s,r}-\tilde{\Lam}_k^{s,r}}}{\sqrt{N}} + \frac{\norm{\tilde{\Lam}_k^{s,r}}}{\sqrt{N}}\cdot \frac{\norm{F_k^{s,r}-\tilde{F}_k^{s,r}}}{\sqrt{T}}  \\ 
    &  ~~~~~~~~+\frac{\norm{F^{k,-r}}}{\sqrt{T}}\cdot \frac{\norm{\Lam_k^{s,-r}-\tilde{\Lam}_k^{s,-r}}}{\sqrt{N}} + \frac{\norm{\tilde{\Lam}_k^{s,-r}}}{\sqrt{N}}\cdot \frac{\norm{F_k^{s,-r}-\tilde{F}_k^{s,-r}}}{\sqrt{T}}.
\end{aligned}
\]
Using that $\frac{\norm{F^{s,r}}}{\sqrt{T}}=\sqrt{r}, \frac{\norm{F^{s,-r}}}{\sqrt{T}}=\sqrt{s-r},  \frac{\norm{\tilde{\Lam}_k^{s,r}}}{\sqrt{N}}\leq M_1$, $ \frac{\norm{\tilde{\Lam}_k^{s,-r}}}{\sqrt{N}}\leq M_3 \delta$, and $d(\th^s,\th_0)\leq \frac{1}{\sqrt{NT}}\sumk\norm{\Lam_k^s F^{s\p}-\tilde{\Lam}^s_k\tilde{F}^{s\p}}$, we have
{\small
\[ 
d(\th^s,\th_0) \leq \sqrt{r}\sumk\frac{\norm{\Lamsr-\tilde{\Lam}_k^{s,r}}}{\sqrt{N}} + \sqrt{s-r}\sumk\frac{\norm{\Lamsnr-\tilde{\Lam_k^{s,-r}}}}{\sqrt{N}}+KM_1 \frac{\norm{F^{s,r}-\tilde{F}^{s,r}}}{\sqrt{T}}+K M_3  \delta \frac{\norm{F^{s,-r}-\tilde{F}^{s,-r}}}{\sqrt{T}}.
\]
}
Next, for $\th^s,~\tilde{\th}^s \in \underset{U\in \mathcal{S}}{\cup}\Theta^s(\delta;U)$, define $d^*$ as
{\small
\[
d^*(\th^s,\tilde{\th}^s)\defeq M_5\cdot \sqrt{\frac{\norm{F^{s,r}-\tilde{F}^{s,r}}^2}{T}+\sumk \frac{\norm{\Lamsr-\tilde{\Lam}_k^{s,r}}^2}{N}+\sumk \frac{\norm{\Lamsnr-\Lam_k^{s,-r}}^2}{N}}+M_6\cdot \delta \cdot \frac{\norm{F^{s,-r}-\tilde{F}^{s,-r}}}{\sqrt{T}},
\]
}
where $M_5=\sqrt{s+(KM_1)^2}$ and $M_6=KM_3$.
Then, by Cauchy-Schwarz inequality, $d(\th^s,\tilde{\th}^s) \leq d^*(\th^s,\tilde{\th}^s)$ and it holds that
\[
D\big(\epsilon,d,\Theta^s(\delta)\big) \leq D\big(\epsilon,d^*,\Theta^s(\delta) \big)\leq D\Big(\frac{\epsilon}{2},d^*,\underset{U \in \mathcal{S}}{\cup}\Theta^s(\delta;U)\Big) \lesssim D\Big(\frac{\epsilon}{2},d^*,\Theta^s(\delta;\I_r)\Big).
\]
The last inequality holds since $|S|=2^r$ and $D(\frac{\epsilon}{2},d^*,\Theta^s(\delta;U_1))=D(\frac{\epsilon}{2},d^*,\Theta^s(\delta;U_2)), \forall~ U_1, U_2 \in \mathcal{S}$. 
Define $\Theta^{s*}(\delta)$ as 
\[
\begin{aligned}
\Theta^{s*}(\delta)&\defeq \bigg\{\th^s \in  \Theta^{s}:M_5\cdot \sqrt{\frac{\norm{F^{s,r}-\tilde{F}^{s,r}}^2}{T}+\sumk \frac{\norm{\Lamsr-\tilde{\Lam}_k^{s,r}}^2}{N}+\sumk \frac{\norm{\Lamsnr}^2}{N}}\\
&~~~~~~~~~~~~~~~~~~~~~~+M_6\cdot \delta \cdot \frac{\norm{F^{s,-r}}}{\sqrt{T}} \leq M_5 M_3 \delta + M_6 \delta \sqrt{s-r} \bigg\}.
\end{aligned}
\]
Since $\frac{\norm{F^{s,-r}}}{\sqrt{T}} \leq \sqrt{s-r}$, it follows that $\Theta^s(\delta;\I_r) \subset \Theta^{s*}(\delta)$.

Now, define $\Lam_k^{0*}=[\Lamz \, \, \mathbf{0}] \in \R^{N \times s}$ and $F^{0*}=[F \, \, V] \in \R^{T \times s}$ with some $V \in \R^{T \times (s-r)}$ 
so that it satisfies $(F^{0*})\p F^{0*}/T = \I_s$. 
Then, if $\th^s \in \Theta^{s*}(\delta)$, 
\[
d^*(\th^s,\th_0^*) \leq M_5 M_3\delta +M_6 \cdot 2\sqrt{s-r}\cdot \delta, 
\]
since $\frac{\norm{F^{s,-r}-V}}{\sqrt{T}} \leq \frac{\norm{F^{s,-r}}}{\sqrt{T}}+\frac{\norm{V}}{\sqrt{T}} \leq 2\sqrt{s-r}$. 
Therefore, for $M_7=M_5 M_3 +M_6 \cdot 2\sqrt{s-r}$,
\[
\Theta^s(\delta;\I_r)\subset B_{d^*}(\theta_0^*,M_7\delta),
\]
and thus,
$D(\epsilon,d,\Theta^s(\delta)) \lesssim D(\frac{\epsilon}{4},d^*,B_{d^*}(M_7 \delta))\leq C(\frac{\epsilon}{8}, d^*, B_{d^*}(M_7 \delta))$.

We now calculate $C(\frac{\epsilon}{8}, d^*, B_{d^*}(\theta_0^*,M_7\delta))$. Let $\eta=\epsilon/8$ and $\th_1^*, \ldots, \th_J^* \in \R^{(KN+T)s}$ be a maximal set of points in $\Theta^{s*}(\delta)$ s.t. $d^*(\th_j^*, \th_\ell^*)>\eta$. 
From $\Theta^{s*}(\delta) \subset \underset{j}{\cup\,} B_{d^*}(\th_j^*,\eta)$, we get $C(\frac{\epsilon}{8}, d^*, B_{d^*}(\theta_0^*,M_7\delta)) \leq J$. Also, by triangular inequality,  
\[
\underset{j}{\cup\,} B_{d^*}(\th_j^*,\frac{\eta}{4} ) \subset B_{d^*}(\theta_0^*,M_7 \delta+\frac{\eta}{4}).
\]
For $\th^s,\tilde{\th}^s \in \Theta^s$, define $d^{**}(\th^s,\tilde{\th}^s)$ as
\[
\begin{aligned}
d^{**}(\th^s,\tilde{\th}^s)=\sqrt{\frac{\norm{F^{s,r}-\tilde{F}^{s,r}}^2}{T/M_5^2}+\sumk \frac{\norm{\Lamsr-\tilde{\Lam}_k^{s,r}}^2}{N/M_5^2}+\sumk \frac{\norm{\Lamsnr-\tilde{\Lam}_k^{s,-r}}^2}{N/M_5^2}+\frac{\norm{F^{s,-r}-\tilde{F}^{s,-r}}^2}{T/(M_6\delta)^2}}.
\end{aligned}
\]
Then, $d^{**}(\th^s,\tilde{\th}^s) \leq d^{*}(\th^s,\tilde{\th}^s) \leq 2 d^{**}(\th^s,\tilde{\th}^s)$ and it follows that 
\[
\underset{j}{\cup \,}{B_{d^{**}}\Big(\th_j^*,\frac{\eta}{4}}\Big) \subset \underset{j}{\cup\,}B_{d^*}\Big(\th_j^*,\frac{\eta}{2}\Big) \subset \ B_{d^*}\Big(\theta_0^*,M_7\delta+\frac{\eta}{2}\Big) \subset B_{d^{**}}\Big(\theta_0^*,M_7\delta+\frac{\eta}{2}\Big).
\]
Note that ${B_{d^{**}}(\th_j^*,\frac{\eta}{4}})$ for $j=1,\ldots,J$ are disjoint, so the volume of $\underset{j}{\cup \,}{B_{d^{**}}(\th_j^*,\frac{\eta}{4}})$ would be $J\cdot \operatorname{V}(B_{d^{**}}(\th_0^*,\frac{\eta}{4}))$, where $\operatorname{V}(\cdot)$ represents the volume. Also, considering the volume of $B_{d^{**}}(\theta_0^*,M_7\delta+\frac{\eta}{2})$, we get 
\[
\begin{aligned}
    J & \leq \frac{\operatorname{V}(B_{d^{**}}(\theta_0^*,M_7\delta+\frac{\eta}{2}))}{\operatorname{V}(B_{d^{**}}(\th_0^*,\frac{\eta}{4}))} = \left(\frac{4M_7\delta+2\eta}{\eta}\right)^{(KN+T)s}=\left(\frac{32M_7\delta+16\epsilon}{\epsilon}\right)^{(KN+T)s} \\
    & \leq \left(\frac{(32M_7+16)\delta}{\epsilon}\right)^{(KN+T)s}, \text{\, for } \epsilon \leq \delta.
\end{aligned}
\]
Therefore, 
\[
\int_0^{\delta} \sqrt{\log D(\epsilon,d,\Theta^r(\delta)}) d\epsilon \leq \int_0^{\delta} \sqrt{\log \, C(\frac{\epsilon}{8},d^*,B_{d^*}(M_7\delta)}) d\epsilon \leq \sqrt{(KN+T)s} \int_0^\delta \sqrt{\frac{(32M_7+16)\delta}{\epsilon}} d\epsilon.
\]
Since $\int_0^\delta \sqrt{\frac{(32M_7+16)\delta}{\epsilon}} d\epsilon=O(\delta)$, we have the desired result. 
\end{proof}
\vspace{5mm}

\begin{lemma} \label{lem:facnum_3} $d(\hat{\th}^s,\th_0)=O_p(L_{NT}^{-1})$.
\end{lemma}
\begin{proof}
Consider $\th_0^*$ as defined in Lemma \ref{lem:facnum_2}. By the definition of $\hat{\th}^s$, we have $\M_{NT}(\hat{\th}^s)=\M_{NT}(\hat{\th}^s)-\M_{NT}(\th_0^*)\leq 0$. Then, using a similar approach as in the proof of Lemma \ref{lem:consist_1}, it can be shown that $d(\hat{\th}^s,\th_0)=o_p(1)$. Following this, in a similar way to the proof of Theorem \ref{thm:consistency}, it can also be proved that $d(\hat{\th}^s,\th_0)=O_p(L_{NT}^{-1})$, and the details are omitted. 
\end{proof}
\vspace{5mm}

\noindent \textbf{Proof of Theorem 3:} First, consider $\ell<r$. We aim to show that $\M_{NT}(\hat{\th}^{\ell})-\M_{NT}(\hat{\th}^r)>C+o_p(1)$ for some constant $C>0$. We use 
\begin{equation} \label{eq:thm3start}
    \M_{NT}(\hat{\th}^{\ell})-\M_{NT}(\hat{\th}^r) = \M^*_{NT}(\hat{\th}^{\ell})-\M_{NT}^*(\hat{\th}^r)=\W_{NT}(\hat{\th}^\ell)+\E[\,M_{NT}(\hat{\th}^{\ell})\,]-\M_{NT}^*(\hat{\th}^r).
\end{equation}
It can be shown by a similar way in Lemma \ref{lem:consist_1} that $|\underset{\th \in \Theta^\ell}{\sup} \W_{NT}(\th)|=o_p(1)$ which then implies $\W_{NT}(\hat{\th}^\ell)=o_p(1)$. Next, for the last term $M_{NT}^*(\hat{\th}^r)$ in (\ref{eq:thm3start}), 
\begin{equation}
    \begin{aligned}
        |M_{NT}^*(\hat{\th}^r)| &=\frac{1}{NT} \left| \sumkit \rho_{\tau_k}(X_{it}-\hat{\lambda}_{k,i}^{\p} \hat{f}_t) -\rho_{\tau_k}(X_{it}-\lambda_{k,i}^{0\p} f_t^0) \right| \\
        & \leq \frac{1}{NT}\sumkit 2|\hat{\lambda}_{k,i}^{\p} \hat{f}_t-\lambda_{k,i}^{0\p} f_t^0| \leq 2 \sqrt{K}\cdot d(\hat{\th}^r,\th_0). 
    \end{aligned}
\end{equation}
Therefore, $M_{NT}^*(\hat{\th}^r)=o_p(1)$. The last inequality is caused by Cauchy-Schwarz inequality. Now, we consider the second term $\E[\,M_{NT}(\hat{\th}^{\ell})\,]$ in (\ref{eq:thm3start}). Like in the proof of Lemma \ref{lem:consist_1}, it can be shown that $\E[\,\M_{NT}(\hat{\th}^\ell)\,] \gtrsim d^2(\hat{\th}^\ell,\th_0)$. Also, we have
\[
\begin{aligned}
    \frac{1}{\sqrt{NT}} \norm{(\hat{\Lam}_{k^*}^\ell \hat{F}^{\ell\p}-\Lam_{k^*}^0 F^{0\p})M_{\hat{F}^\ell}} &= \frac{1}{\sqrt{NT}}\norm{\Lam_{k^*}^0 F^{0\p}M_{\hat{F}^\ell}}\\
    &=\frac{1}{\sqrt{NT}}\sqrt{Tr\left(\frac{\Lam_{k^*}^{0\p} \Lam_{k^*}^{0\p}}{N}\cdot \frac{F^{0\p} M_{\hat{F}^\ell}F^0}{T}\right)} \geq \sqrt{\sigma_{r}(N)}\cdot \frac{\norm{M_{\hat{F}^\ell}F^0}}{\sqrt{T}},
\end{aligned}
\]
and since
\[
\frac{1}{\sqrt{NT}} \norm{(\hat{\Lam}_{k^*}^\ell \hat{F}^{\ell\p}-\Lam_{k^*}^0 F^{0\p})M_{\hat{F}^\ell}} \lesssim \frac{1}{\sqrt{NT}} \norm{(\hat{\Lam}_{k^*}^\ell \hat{F}^{\ell\p}-\Lam_{k^*}^0 F^{0\p})} \leq d(\hat{\th}^\ell,\th_0),
\]
we have 
\[
 \frac{\norm{M_{\hat{F}^\ell}F^0}^2}{T} \lesssim d(\hat{\th}^\ell,\th_0)^2 \lesssim \E[\,\M_{NT}(\hat{\th}^\ell)\,].
\]
Note that 
\[
\frac{\norm{M_{\hat{F}^\ell}F^0}^2}{T} = \frac{1}{T}Tr\Big(F^{0\p}\Big(\I-\frac{\hat{F}^\ell \hat{F^{\ell\p}}}{T}\Big)F^0\Big)=Tr\Big(\I_r-\frac{F^{0\p} \hat{F}^{\ell} \hat{F}^{\ell} F^{0}}{T^2}\Big) \geq \rho_1\Big(\I_r-\frac{F^{0\p} \hat{F}^{\ell} \hat{F}^{\ell} F^{0}}{T^2}\Big).
\]
By Weyl's inequality, 
\[
\rho_1\Big(\I_r-\frac{F^{0\p} \hat{F}^{\ell} \hat{F}^{\ell} F^{0}}{T^2}\Big)+\rho_r\Big(\frac{F^{0\p} \hat{F}^{\ell} \hat{F}^{\ell} F^{0}}{T^2}\Big) \geq \rho_1(\I_r)=1.
\]
Since
$\rank(F^{0\p} \hat{F}^{\ell} \hat{F}^{\ell} F^{0}) \leq \rank (\hat{F}^{\ell})=\ell < r,$ it holds that $\rho_r(\frac{F^{0\p} \hat{F}^{\ell} \hat{F}^{\ell} F^{0}}{T^2})=0$, and thus, $\frac{\norm{M_{\hat{F}^\ell}F^0}^2}{T} \geq 1.$ Therefore, $\E[\,M_{NT}(\hat{\th}^{\ell})\,] \geq C $ for some constant $C>0$.
From (\ref{eq:thm3start}), we obtain $\M_{NT}(\hat{\th}^{\ell})-\M_{NT}(\hat{\th}^r)>C+o_p(1)$. 
Note that if $\hat{r}=\ell$, by definition of $\hat{r}$, we have
\[
\M_{NT}(\hat{\th}^{\ell})+\ell \cdot P_{NT} < \M_{NT}(\hat{\th}^r)+r \cdot P_{NT}.
\] Then,
\[
C+o_p(1)<\M_{NT}(\hat{\th}^{\ell})-\M_{NT}(\hat{\th}^r)<(r-\ell)\cdot P_{NT}.
\]
It implies that 
\[
P(\hat{r}=\ell) \leq P\left(C+o_p(1)<(r-\ell)\cdot P_{NT}\right) \leq P\left(\frac{C}{2}<o_p(1)\right)+P\left(\frac{C}{2}<(r-\ell)\cdot P_{NT}\right).
\]
Hence, $P(\hat{r}=\ell)  \rightarrow 0$ as $N,~T \rightarrow \infty$ for all $\ell<r$. Therefore, $P(\hat{r}<r)\rightarrow 0$ as $N,~T \rightarrow \infty$.

Now, we consider the case $\ell>r$. In this case, we aim to show that $\M_{NT}(\hat{\th}^{\ell})-\M_{NT}(\hat{\th}^r)=O_p(L_{NT}^{-2})$.
We use 
\begin{equation} \label{eq:thm3start2}
    \M_{NT}(\hat{\th}^{\ell})-\M_{NT}(\hat{\th}^r) = \W_{NT}(\hat{\th}^{\ell})-\W_{NT}(\hat{\th}^r)+\E[\,\M_{NT}^*(\hat{\th}^\ell)\,]-\E[\,\M_{NT}^*(\hat{\th}^r)\,].
\end{equation}
Using the inequality $\E[\underset{\th^\ell \in \Theta^\ell(\delta)}{\sup}|\W_{NT}(\th^\ell)|] \lesssim \frac{\delta}{L_{NT}}$ from Lemma \ref{lem:facnum_2}, and the fact that $d(\hat{\th}^\ell,\th_0)=O_p(L_{NT}^{-1})$, we show that $\W_{NT}(\hat{\th}^\ell)=O_p(L_{NT}^{-2})$. For $C_1,~C_2>0$,
\[
\begin{aligned}
P(L_{NT}^2|\W_{NT}(\hat{\th}^\ell)|>C_1)  &\leq P(L_{NT}^2|\W_{NT}(\hat{\th}^\ell)|>C_1, L_{NT}\cdot d(\hat{\th}^\ell,\th_0)<C_2) +P(L_{NT}\cdot d(\hat{\th}^\ell,\th_0)>C_2)\\
& \leq P\bigg(\underset{d(\th^l,\th_0)<\frac{C_2}{L_{NT}}}{\sup}L_{NT}^2|\W_{NT}(\hat{\th}^\ell)|>C_1\bigg)+P(L_{NT}\cdot d(\hat{\th}^\ell,\th_0)>C_2) \\
& \leq \frac{1}{C_1}L_{NT}^2 \E\left[ \, |\W_{NT}(\hat{\th}^\ell)|\,\right]+ P(L_{NT}\cdot d(\hat{\th}^\ell,\th_0)>C_2) \\
&\lesssim \frac{1}{C_1}L_{NT}^2\cdot \frac{C_2}{L_{NT}^2}+P(L_{NT}\cdot d(\hat{\th}^\ell,\th_0)>C_2).
\end{aligned}
\]
Given $\epsilon>0$, we can choose $C_1>0$ s.t. $P(L_{NT}^2|\W_{NT}(\hat{\th}^\ell)|>C_1) < \epsilon$ using the previous inequality and that $L_{NT}\cdot d(\hat{\th}^\ell,\th_0)=O_p(1)$. Therefore, $\W_{NT}(\hat{\th}^\ell)=O_p(L_{NT}^{-2})$.

By Lemma \ref{lem:consist_3}, $\E\left[\underset{\th^r \in \Theta^r(\delta)}{\sup}|\W_{NT}(\th^r)|\right] \lesssim \frac{\delta}{L_{NT}}$ and by Theorem \ref{thm:consistency}, $d(\hat{\th}^r,\th_0)=o_p(1)$. It again follows that $ \, \W_{NT}(\hat{\th}^r)=O_p(L_{NT}^{-2})$. Now, we consider $\E[\,\M_{NT}^*(\hat{\th}^\ell)\,]-\E[\,\M_{NT}^*(\hat{\th}^r)\,]$. Similar to the proof of Lemma \ref{lem:consist_1} and by Assumption 3, we have 
\[
|\E[\,\M_{NT}^*(\hat{\th}^\ell)\,]| \lesssim d^2(\hat{\th}^\ell,\th_0) \text{ and } |\E[\,\M_{NT}^*(\hat{\th}^r)\,]| \lesssim d^2(\hat{\th}^r,\th_0).
\]
Then,
\[
\left|\E[\,\M_{NT}^*(\hat{\th}^\ell)\,]-\E[\,\M_{NT}^*(\hat{\th}^r)\,]\right| \leq \left|\E[\,\M_{NT}^*(\hat{\th}^\ell)\,]\right|+\left|\E[\,\M_{NT}^*(\hat{\th}^r)\,]\right| \lesssim d^2(\hat{\th}^\ell,\th_0)+d^2(\hat{\th}^r,\th_0), 
\]
and therefore, $\E[\,\M_{NT}^*(\hat{\th}^\ell)\,]-\E[\,\M_{NT}^*(\hat{\th}^r)\,]=O_p(L_{NT}^{-2})$. Now, from (\ref{eq:thm3start2}), we have $\M_{NT}(\hat{\th}^{\ell})-\M_{NT}(\hat{\th}^r)=O_p(L_{NT}^{-2})$. 
Note that if $\hat{r}=\ell$,
\[
(\ell-r)\cdot P_{NT}<\M_{NT}(\hat{\th}^{r})-\M_{NT}(\hat{\th}^\ell).
\]
It implies that 
\[
P(\hat{r}=\ell) \leq P\left[(\ell-r)\cdot P_{NT}L_{NT}^2<L_{NT}^2(\M_{NT}(\hat{\th}^{r})-\M_{NT}(\hat{\th}^\ell))\right].
\]
Hence, $P(\hat{r}=\ell)  \rightarrow 0$ as $N,~T \rightarrow \infty$ for all $\ell<r$. Therefore, $P(\hat{r}<r)\rightarrow 0$ as $N,~T \rightarrow \infty$.

It follows from the results from the two cases that $P(\hat{r}=r)\rightarrow 1$ as $N,~T \rightarrow \infty$. 

\vspace{5mm}

\noindent \textbf{Proof of Theorem 4:}
By (\ref{eq:eigen for <r}) in the proof of Lemma \ref{lem:facnum_1}, $|\hat{\sigma}^{k^*}_j(N)-\nu_j(N)| \lesssim d(\hat{\th}^s,\th_0)$ for $j=1,\ldots,r$, where $\nu_j(N)$ is the $j$th eigenvalue of $\frac{\Lambda_{k^*}^{0\p}{\Lambda_{k^*}^{0}}}{\sqrt{N}}$. Thus, by Lemma \ref{lem:facnum_3}, $|\hat{\sigma}^{k^*}_j(N)-\nu_j(N)|=o_p(1)$. Also, since $|\nu_{j}(N)-\nu_j|=o(1)$ by Assumption 1(b), we have $|\hat{\sigma}^{k^*}_j(N)-\nu_j|=o_p(1)$ for $j=1,\ldots,r$. Then,
\[
\begin{aligned}
P(\hat{r}<r) &= P(\underset{k}{\max} \sum\limits_{j=1}^s I(\hat{\sigma}^{k}_{j}(N)>\kappa_{NT})<r) \\
& \leq P(\sum\limits_{j=1}^s I(\hat{\sigma}^{k^*}_j(N)>\kappa_{NT})<r))\\
& \leq P(\hat{\sigma}^{k^*}_r(N) \leq \kappa_{NT}).
\end{aligned}
\]
Also, if $\hat{\sigma}^{k^*}_r(N) \leq \kappa_{NT}$, then $\nu_{r}  \leq \kappa_{NT}+ |\hat{\sigma}^{k^*}_r(N)-\nu_r|$. Therefore, $P(\hat{r}<r) \leq P(\nu_r  \leq \kappa_{NT}+ |\hat{\sigma}^{k^*}_r(N)-\nu_r|)$. 
Since $\kappa_{NT}$ and $|\hat{\sigma}^{k^*}_r(N)-\nu_r|$ are $o_p(1)$,
\begin{equation} \label{eq:rhat<r}
    P(\hat{r}<r) \rightarrow 0 \text{ as } N,~T \rightarrow \infty.
\end{equation}
Next, by Lemma \ref{lem:facnum_3}, $d(\hat{\th}^s,\th_0)=O_p(L_{NT}^{-1})$. Fix $k \in \{1,\ldots, K\}$.\\ 
Define $A = \begin{bmatrix}
    \frac{\Lamzp \Lamz}{N} & 0 \\
     0 & 0 \\
\end{bmatrix} \in \mathbb{R}^{s \times s}$ and $D_{NT}=\frac{\hat{\Lam}_{k}^{s\p}\hat{\Lam}_{k}^{s}}{N}-A$. 
By Weyl's theorem for the Hermitian matrix,
\begin{equation} \label{eq:sigma<DNT}
    \big|\hat{\sigma}^{k}_{r+1}(N)-0\big|=\Big|\rho_{r+1}\Big(\frac{\hat{\Lam}_k^{s\p} \hat{\Lam}_k^s}{N}\Big)-\rho_{r+1}(A)\Big| \leq  \max\big[\,|\rho_1(D_{NT})|, |\rho_s(D_{NT})|\,\big] \leq \norm{D_{NT}}.
\end{equation}
Note that by Lemma \ref{lem:facnum_1},
\begin{equation} \label{eq:DNT<Op}
    \begin{aligned} 
\norm{D_{NT}} &\leq \norm{\frac{(\hat{\Lam}_k^{s,r})\p \hat{\Lam}_k^{s,r}}{N}-\frac{\Lamzp\Lamz}{N}}+\norm{\frac{(\hat{\Lam}_k^{s,r})\p\hat{\Lam}_k^{s,-r}}{N}}+\norm{\frac{(\hat{\Lam}_k^{s,-r})\p\hat{\Lam}_k^{s,r}}{N}}+\norm{\frac{(\hat{\Lam}_k^{s,-r})\p\hat{\Lam}_k^{s,-r}}{N}} \\
& \leq \frac{\norm{\hat{\Lam}_k^{s,r}}}{\sqrt{N}}\cdot \frac{\norm{\hat{\Lam}_k^{s,r}-\Lamz}}{\sqrt{N}} + \frac{\norm{\Lamz}}{\sqrt{N}}\cdot \frac{\norm{\hat{\Lam}_k^{s,r}-\Lamz}}{\sqrt{N}}   \\ 
&~~~~~~~~~~~ +\frac{\norm{\hat{\Lam}_k^{s,r}}}{\sqrt{N}}\cdot \frac{\norm{\hat{\Lam}_k^{s,-r}}}{\sqrt{N}}+\frac{\norm{\hat{\Lam}_k^{s,-r}}}{\sqrt{N}}\cdot \frac{\norm{\hat{\Lam}_k^{s,r}}}{\sqrt{N}}+\bigg(\frac{\norm{\hat{\Lam}_k^{s,-r}}}{\sqrt{N}}\bigg)^2\\
& \lesssim d(\hat{\th}^s,\th_0)=O_p(L_{NT}^{-1}). 
\end{aligned}
\end{equation}
Also, it holds that 
\begin{equation} \label{eq:Prob<Op}
\begin{aligned}
     P(\hat{r}>r) &\leq \sumk P(\sum\limits_{j=1}^s I(\hat{\sigma}^{k}_j(N)>\kappa_{NT})>r)\\
     & \leq \sumk P(\hat{\sigma}^{k}_{r+1}(N)>\kappa_{NT}) \\
    & \leq K \cdot P(\sumk \hat{\sigma}^{k}_{r+1}(N) >\kappa_{NT})\\
    &=K\cdot P\Big(L_{NT} \sumk \hat{\sigma}^{k}_{r+1}(N)>L_{NT}\kappa_{NT}\Big).
\end{aligned}
\end{equation}
By (\ref{eq:sigma<DNT}) and (\ref{eq:DNT<Op}), $\sumk \hat{\sigma}^{k}_{r+1}(N) = O_p(L_{NT}^{-1})$ for all $k$. Then, it follows from (\ref{eq:Prob<Op}) that 
\begin{equation} \label{eq:rhat>r}
    P(\hat{r}>r) \rightarrow 0 \text{ as } N,~T \rightarrow \infty.
\end{equation}
Finally, the desired result is followed by (\ref{eq:rhat<r}) and (\ref{eq:rhat>r}). 


\end{document}